\documentclass[12pt,letterpaper]{article}
\usepackage[paper=letterpaper,margin=.85in,headheight=14pt]{geometry}
\usepackage{microtype}
\usepackage{mathtools,amssymb,amsthm,mathrsfs}
\usepackage{enumitem,booktabs,tabularx}
\usepackage{needspace}
\usepackage[numbers,compress]{natbib}
\usepackage[dvipsnames]{xcolor}
\usepackage{xurl}
\usepackage[hidelinks,pdfencoding=auto,psdextra]{hyperref}
\usepackage{bookmark}

\setlist{itemsep=3pt,topsep=5pt,parsep=0pt}
\allowdisplaybreaks[1]
\numberwithin{equation}{section}

\newtheorem{theorem}{Theorem}[section]
\newtheorem{proposition}[theorem]{Proposition}
\newtheorem{lemma}[theorem]{Lemma}
\theoremstyle{definition}
\newtheorem{definition}[theorem]{Definition}
\theoremstyle{remark}

\newcommand{\R}{\mathbb R}
\newcommand{\Z}{\mathbb Z}
\newcommand{\HH}{\mathbb H}
\newcommand{\Ehat}{\widehat E}

\newcommand{\Uc}{\mathcal U_c}
\newcommand{\Repair}{\mathcal D}
\newcommand{\dd}{\,\mathrm d}

\DeclareMathOperator{\supp}{supp}
\DeclareMathOperator{\TV}{TV}

\hypersetup{pdftitle={Large gaps and BTZ entropy in modular spectra with positive integer degeneracies},
  pdfauthor={Chi-Ming Chang, Reiko Liu, Wen-Jie Ma},
  pdfsubject={Gap families, local repair, recursive construction, global consistency, and BTZ matching of smoothed densities of states}}

\begin{document}
\thispagestyle{empty}
\vspace*{1cm}
\begin{center}
  {\LARGE\bfseries Large gaps and BTZ entropy in modular spectra\\[0.2em]
    with positive integer degeneracies\par}
  \vspace{1cm}
  {\normalsize\bfseries Chi-Ming Chang\textsuperscript{a,b},\quad
    Reiko Liu\textsuperscript{c},\quad
    Wen-Jie Ma\textsuperscript{c,d}\par}
  \vspace{9mm}
  {\normalsize\rmfamily
    \textsuperscript{a}Yau Mathematical Sciences Center (YMSC),\\
    Tsinghua University, Beijing, China\par
    \textsuperscript{b}Beijing Institute of Mathematical Sciences and Applications (BIMSA),\\
    Beijing, China\par
    \textsuperscript{c}Shanghai Institute for Mathematics and Interdisciplinary Sciences (SIMIS),\\
    Shanghai 200433, China\par
    \textsuperscript{d}Fudan Center for Mathematics and Interdisciplinary Study,\\
    Fudan University, Shanghai 200438, China\par}
\end{center}
\vspace{4mm}
\begin{abstract}
We construct modular-invariant torus partition functions with a unique
vacuum, discrete energy levels and positive integer degeneracies by
recursively repairing the Maloney--Witten--Keller modular completion
of the Virasoro vacuum. Exact modular repairs and finite moment
matching discretize successive
spectral bands while preserving earlier levels and controlling
convergence. For $c_L=c_R=c$ and $a=(c-1)/12$, the construction realizes
primary dimension gaps $\Delta_1=(1+\kappa)a$ for sufficiently small
fixed $\kappa>0$, and $\Delta_1=a+\delta$ for any fixed $\delta\ge0$,
at every sufficiently large real $c$. Every nonvacuum primary satisfies
$h,\bar h\ge(c-1)/24$. In the fixed-$\delta$ family, spectra can be
chosen whose densities of states, smoothed with a fixed nonnegative
normalized smooth kernel of compact support, match the correspondingly
smoothed prediction of a single perturbative BTZ saddle through every
fixed finite order in $1/c$. The count includes all spins and Virasoro
descendants. Its logarithm reproduces the Bekenstein--Hawking entropy and its
corrections at every fixed positive $E/c$, where $E=\Delta-c/12$.
This includes $0<E<c/12$, where thermal AdS dominates the canonical
ensemble.
\end{abstract}

\clearpage

\setcounter{tocdepth}{2}
\tableofcontents
\clearpage

\section{Introduction}
\label{sec:introduction}

The relation between BTZ black-hole entropy
\cite{BanadosTeitelboimZanelli1992} and the asymptotic growth of
states in a two-dimensional conformal field theory~\cite{Cardy1986}
is a basic test of AdS$_3$/CFT$_2$~\cite{Strominger1998}.
For pure gravity, a microscopic
account also requires a spectrum with a unique vacuum, discrete energy levels
and positive integer degeneracies, compatible with modular invariance
and the absence of light matter. In this paper we construct genus-one
partition functions with these spectral properties and large primary
gaps. For one family, we show that their smoothed densities of states
reproduce the prediction of a perturbative BTZ saddle through every
fixed finite order in the inverse central charge, including energies
where that saddle does not dominate the canonical ensemble.

A natural starting point is the modular completion of the Virasoro
vacuum. The gravitational sum of Maloney and Witten and the Poincar\'e
construction of Keller and Maloney combine the vacuum contribution
with its modular images~\cite{MaloneyWitten2007,KellerMaloney2014}.
This gives an exactly modular-invariant function, but its nonvacuum
spectrum is continuous and contains a negative scalar contribution
at threshold. It also has negative spectral density in exponentially
narrow near-extremal windows at large odd spin
\cite{BenjaminEtAl2019,AldayBae2020}. These defects obstruct its
interpretation as the partition function of a compact unitary CFT.
Keller and Maloney argued that corrections subleading in the semiclassical
limit could render the spectrum positive and discrete, and discussed
their possible gravitational origin
\cite[Sec.~5]{KellerMaloney2014}. Our purpose is to give a recursive
construction in which discreteness, integer degeneracies, the gap
and the convergence of the completed modular function are controlled
simultaneously.

Several approaches illuminate parts of this problem. Conical-defect
orbifolds can repair negative densities by introducing primaries below
the black-hole threshold~\cite{BenjaminCollierMaloney2020}.
Resumming off-shell Seifert contributions in the near-extremal limit
yields a positive continuous density with a nonperturbatively shifted
spectral edge~\cite{MaxfieldTuriaci2020}. A stringy completion gives a
positive spectrum with primary dimension gap $(c-1)/12$, while retaining
a continuum and some subthreshold chiral weights
\cite{DiUbaldoPerlmutter2023Stringy}.
Discrete integral modular candidates also arise from vector-valued
modular forms~\cite{BenjaminEtAl2018Irrational}, with different gaps
and chiral content. The modular bootstrap places upper bounds on the
primary dimension gap through analytic arguments
\cite{Hellerman2009,HartmanMazacRastelli2019} and numerical
optimization~\cite{FriedanKeller2013,CollierLinYin2016,AfkhamiJeddiHartmanTajdini2019}.
Discreteness and integrality impose further constraints on modular
spectra~\cite{KaidiPerlmutter2020,FitzpatrickLi2023,ChiangEtAl2023}.
Recent numerical searches have also produced approximate truncated
spectra with integer degeneracies near $c=1$~\cite{BFLT2026}.
The construction below keeps every nonvacuum primary at or above the shifted
chiral threshold and imposes positive integer multiplicities on the
entire discrete spectrum.

Our method makes a prescribed spectral change in a bounded energy band
by an exact modular repair. It then replaces the continuous measure in
successive spin and energy intervals by finitely many states with
positive integer multiplicities. The replacement discrete measure is
chosen to match sufficiently many moments of the removed continuous
measure. These moment identities control the accompanying modular
correction at higher energies.
The order of these operations protects every previously constructed
level, while the sum of the complete repairs converges at every positive
temperature. Each bounded energy band eventually stabilizes to a finite
atomic spectrum. This provides the individual levels as solutions of
finite moment problems, as well as the estimates needed to pass modular
invariance to the infinite limit.

We take equal left and right central charges $c_L=c_R=c$, write
$\Delta=h+\bar h$, and set
\[
  a=\frac{c-1}{12},\qquad
  \Ehat=\Delta-a,\qquad J=h-\bar h.
\]
Every constructed nonvacuum primary satisfies $\Ehat\ge|J|$, or
$h,\bar h\ge(c-1)/24$. Writing $\Delta_1$ for the lowest
nonvacuum primary dimension, Theorem~\ref{thm:analytic} gives two
families, each available at every sufficiently large real $c$:
\begin{equation}
  \begin{aligned}
    \Delta_1&=(1+\kappa)a,
       &&0<\kappa<\kappa_0\ \text{fixed},\\
    \Delta_1&=a+\delta,
       &&\delta\ge0\ \text{fixed},
  \end{aligned}
  \label{eq:introduction-gap-families}
\end{equation}
where $\kappa_0>0$ is a sufficient constant determined by the
construction estimates. The first family has
$\Delta_1/c\to(1+\kappa)/12>1/12$; the second realizes any fixed
nonnegative shifted gap, including zero. These are primary dimension
gaps. The separate chiral bound allows primaries at large spin to
approach $\Ehat=|J|$. Throughout, we retain the distinction between
$\Ehat$ and the cylinder energy $E=\Delta-c/12=\Ehat-1/12$.

The removal of the odd-spin negative density is compatible with the
modular-ambiguity analysis of Alday and Bae: the correcting seeds
generated by our recursion do not satisfy its growth assumptions
\cite[Secs.~4.1 and 4.3]{AldayBae2020}.
We discuss the role of the infinite recursive limit in
Section~\ref{sec:discussion} and give quantitative estimates in
Appendix~\ref{app:global-near-extremal}.

The second main result concerns the entropy of the actual discrete
spectra. For a fixed nonnegative, compactly supported smooth kernel
$\phi\in C_c^\infty(\R)$, normalized to have integral one, and
states with cylinder energies $E_s$ and multiplicities $d_s$, define
$N_\phi(E)=\sum_s d_s\phi(E_s-E)$ and
$S_\phi(E)=\log N_\phi(E)$. This count includes all spins,
Virasoro descendants and the vacuum module. For the finite shifted-gap
family, Proposition~\ref{prop:btz-entropy} gives
\begin{equation}
  \begin{aligned}
    S_\phi(\lambda c)
      &=2\pi c\sqrt{\frac{\lambda}{3}}-\frac12\log c
         +C_{0,\phi}(\lambda)
         +\sum_{m=1}^{P}\frac{C_{m,\phi}(\lambda)}{c^m}
         +O_P(c^{-P-1})
  \end{aligned}
  \label{eq:introduction-entropy}
\end{equation}
for every fixed $\lambda>0$ and every fixed integer $P\ge0$.
All coefficients agree with the single
perturbative BTZ contribution, including its boundary gravitons,
smoothed with the same kernel
\cite{GiombiMaloneyYin2008,CotlerJensen2019}.
The proof uses the matched moments to compare the atomic spectrum with
the leading vacuum spectral density, then identifies the latter's
thermal transform with the BTZ contribution.

The energy range is significant. Under the sparseness assumptions
of Hartman, Keller and Stoica, the leading canonical free energy fixes
the Cardy entropy above $E=c/12$, while leaving the intermediate
range $0<E<c/12$ undetermined~\cite{HartmanKellerStoica2014}.
Our direct spectral estimates apply throughout this intermediate
range at every fixed positive $E/c$. There the relevant inverse
Laplace saddle lies at $\beta_*>2\pi$, and thermal AdS dominates
the canonical partition function.

Earlier all-order high-energy corrections to Cardy's formula were
derived using a smooth-density approximation
\cite{LoranSheikhJabbariVincon2010}. Our estimates control the
comparison with the actual discrete spectrum at large $c$ and fixed
positive $E/c$. Rigorous averaging and logarithmic corrections have
been studied using Tauberian methods
\cite{MukhametzhanovZhiboedov2019}.
All-order matching after smoothing also has a close precedent in the
large-spin theorem of Pal, Qiao and van Rees
\cite{PalQiaoVanRees2025}. Their limit keeps $c$ fixed and studies
the primary density at an individual large spin; ours takes
$c\to\infty$, includes descendants and sums over spins at fixed
positive $E/c$. Equation~\eqref{eq:introduction-entropy} concerns
the spectra constructed here, with a remainder controlled to every
fixed inverse power of $c$.

A macroscopic primary gap changes the answer near its lower edge.
Section~\ref{sec:entropy-positive-kappa} studies this behavior for the
explicit positive-$\kappa$ construction, using the compact smooth bump
kernel specified there. At energies an extensive distance below the gap,
the smoothed compact count contains only vacuum descendants. Just above the gap,
additional primary levels with large multiplicities and their descendants
produce an interval of exponential excess over the BTZ prediction.
At sufficiently high extensive energies, BTZ matching holds through
every fixed finite order in $1/c$.
Determining how the near-gap excess gives way to BTZ matching at higher
energies remains an open problem.

The construction also leaves freedom in the microscopic spectra.
Different finite shifted-gap spectra at the same central charge share
the expansion \eqref{eq:introduction-entropy}. This is compatible
with statistical descriptions of AdS$_3$ gravity motivated by torus
wormholes and ensembles of CFT data
\cite{CotlerJensen2020Random,ChandraEtAl2022Ensemble}.
Narain ensembles provide concrete examples in which an average over
discrete spectra has a bulk description involving Abelian Chern--Simons
theory
\cite{AfkhamiJeddiEtAl2020Narain,MaloneyWitten2020Narain}.
Our family does not itself specify a probability measure or connected
spectral correlations. It supplies spectral targets for asking whether
an extended gravitational path integral can reproduce individual
members, or select an ensemble on them. These bulk questions, including
the distinction between extra saddles and off-shell contributions,
are discussed in Section~\ref{sec:discussion}.

A natural next step is to determine which members of this family admit
OPE coefficients compatible with unitary correlation functions and
higher-genus sewing. Section~\ref{sec:discussion} outlines a route
through sphere crossing, torus one-point functions and genus-two
consistency, and asks how these conditions could constrain the freedom
in the spectral construction.

The two gap families in Theorem~\ref{thm:analytic} and the all-order
smoothed BTZ entropy expansion in Proposition~\ref{prop:btz-entropy}
have been formalized in Lean~4.
The companion repository\footnote{\url{https://github.com/ReikoAntoneva/modular-bootstrap-lean}.}
provides the formal statements and proofs, their correspondence with
the results of this paper, and instructions for reproducing the
verification.

Section~\ref{sec:setup} fixes conventions and states the main
results. Section~\ref{sec:local-repair} develops exact local
modular repair, and Section~\ref{sec:construction} gives the
integer discretization, the requirements on the full recursion and
one explicit schedule.
Section~\ref{sec:global-consistency} proves convergence and the
gap statements. Section~\ref{sec:entropy} derives the spectral
entropy results, and Section~\ref{sec:discussion} discusses their
physical interpretation and remaining questions.

\section{Setup and overview of results}
\label{sec:setup}

We study torus partition functions obtained by modifying the modular
completion of the Virasoro vacuum. The aim is to replace its signed
continuous spectrum by positive integer primary multiplicities while
preserving exact modular invariance and uncensored support. This section
specifies the spectral problem, states the gap theorem and the BTZ
entropy result for a smoothed density of states, and describes the
recursive construction.

\subsection{Characters, energies, and thresholds}
\label{sec:conventions}

Throughout, the central charges are $c_L=c_R=c>1$, so that the total
central charge is $2c$. For a state of weights $(h,\bar h)$, set
\begin{equation}
  \Delta=h+\bar h,\qquad J=h-\bar h,\qquad
  a=\frac{c-1}{12}.
  \label{eq:weights}
\end{equation}
On a spatial circle of circumference $2\pi$, we distinguish the cylinder
energy from the energy appearing in the primary character numerator:
\begin{equation}
  E:=\Delta-\frac{c}{12},\qquad
  \Ehat:=\Delta-a=E+\frac1{12}.
  \label{eq:energies}
\end{equation}
The hat always denotes this exact shift. The construction uses
$\Ehat$; $E$ is used for the physical cylinder Hamiltonian.

For $q=\exp(2\pi i\tau)$, the characters are
\begin{equation}
  \chi_0(\tau)=\frac{q^{-a/2}(1-q)}{\eta(\tau)},\qquad
  \chi_h(\tau)=\frac{q^{h-a/2}}{\eta(\tau)}\quad(h>0).
  \label{eq:characters}
\end{equation}
In particular, writing $\tau=x+iy$, a nonvacuum primary with both weights
positive contributes
\begin{equation}
  \chi_h(\tau)\overline{\chi_{\bar h}(\tau)}
  =\frac{\exp(-2\pi y\Ehat+2\pi i xJ)}{|\eta(\tau)|^2}.
  \label{eq:primary-character}
\end{equation}
This factorization makes $\Ehat$ the natural variable for the primary
spectral measure.

We call a nonvacuum primary \emph{uncensored} when
\begin{equation}
  h,\bar h\ge\frac{c-1}{24}
  \quad\Longleftrightarrow\quad
  \Ehat\ge |J|.
  \label{eq:uncensored}
\end{equation}
This is the convention used for the primary spectrum in the MWK
construction and subsequent pure-gravity modular-bootstrap discussions
\cite{KellerMaloney2014,BenjaminEtAl2019}. It is an additional spectral
condition, beyond the nonnegative weights required by unitarity.
The vacuum module is kept separately.

\Needspace{12\baselineskip}
The finite-$c$ distinction between the two energies is physical. In the
classical BTZ dictionary, $M\ell=E$ and the horizon condition is
$E\ge |J|$. The scalar uncensored threshold instead occurs at
$\Ehat=0$. Keller and Maloney interpret the shift of the primary threshold
as a one-loop shift of the lightest black-hole mass
\cite{KellerMaloney2014}. We use the following terminology:
\begin{center}
\begin{tabular}{@{}lccc@{}}
\toprule
Scalar configuration & $h=\bar h$ & $E$ & $\Ehat$\\
\midrule
Classical massless BTZ & $c/24$ & $0$ & $1/12$\\
Scalar uncensored threshold & $(c-1)/24$ & $-1/12$ & $0$\\
\bottomrule
\end{tabular}
\end{center}
Both rows have $J=0$. When discussing the MWK interpretation, we also
refer to the second row as the \emph{one-loop-shifted BTZ threshold}.
Our exact support condition is \eqref{eq:uncensored}, or equivalently
$E\ge |J|-1/12$.

\subsection{The genus-one spectral class}
\label{sec:spectral-class}

The data to be constructed consist of primary weights and their
multiplicities. Descendants are then supplied by the characters
\eqref{eq:characters}.

\begin{definition}[Discrete uncensored torus spectrum]
\label{def:class}
Let $\Uc$ be the class of partition functions
\begin{equation}
  Z(\tau,\bar\tau)=|\chi_0(\tau)|^2+
  \sum_\alpha d_\alpha\,
  \chi_{h_\alpha}(\tau)\overline{\chi_{\bar h_\alpha}(\tau)}
  \label{eq:partition-function}
\end{equation}
with the following properties:
\begin{enumerate}[label=(\roman*)]
  \item The vacuum module has multiplicity one. Every other primary has
    $h_\alpha,\bar h_\alpha\ge(c-1)/24$.
  \item The multiplicities satisfy $d_\alpha\in\Z_{>0}$ and the spins
    satisfy $J_\alpha\in\Z$.
  \item There are finitely many primaries below every finite dimension,
    counted with multiplicity, and
    \begin{equation}
      \sum_\alpha d_\alpha\exp(-t\Delta_\alpha)<\infty
      \qquad\text{for every }t>0.
      \label{eq:thermal-convergence}
    \end{equation}
  \item On $\HH$, the partition function is invariant under
    $\mathsf S:\tau\mapsto-1/\tau$ and
    $\mathsf T:\tau\mapsto\tau+1$, acting also on $\bar\tau$ by conjugation.
\end{enumerate}
\end{definition}

This definition states the genus-one spectral properties under study.
It does not require parity invariance or specify OPE coefficients,
correlation functions, or higher-genus amplitudes. In particular, proving
$Z\in\Uc$ does not by itself construct a full CFT. The uncensored
condition is an additional restriction, which general unitary compact
CFTs need not satisfy.

The dimension gap is the infimum over all nonvacuum primaries and all
spins,
\begin{equation}
  \Delta_1(Z):=\inf_\alpha\Delta_\alpha,\qquad
  \Ehat_1(Z)=\Delta_1(Z)-a.
  \label{eq:dimension-gap}
\end{equation}
The constructed spectra are nonempty and locally finite, so their gaps
are attained. The twist of a primary is
$\Delta-|J|=2\min(h,\bar h)$. The uncensored condition bounds this
quantity below by $a$, whereas a lower bound on $\Delta_1$ concerns
the total dimension. A primary near a chiral threshold can have
arbitrarily large dimension if its spin is large.

\Needspace{16\baselineskip}
\subsection{Main results}
\label{sec:analytic-result}

\paragraph{Gap families.}
The construction gives two large-$c$ regimes: a dimension gap exceeding
$c/12$ by an amount of order $c$, and a shifted gap equal to any
fixed nonnegative value. Both regimes are available for
all sufficiently large real central charges, parametrized by
\begin{equation}
  c=12a+1,\qquad a\in\R_{>0}.
  \label{eq:family}
\end{equation}

\begin{theorem}[Gap families]
\label{thm:analytic}
There exist families in $\mathcal U_c$ with the following properties.
\begin{enumerate}[label=(\roman*),ref=(\roman*)]
  \item \label{item:positive-fraction}
  There is a constant $\kappa_0>0$ such that, for every fixed
  $0<\kappa<\kappa_0$, a family $Z_{a,\kappa}\in\mathcal U_c$
  exists for all sufficiently large real $a$, with $b=\kappa a$ and
  \begin{equation}
    \Ehat_1(Z_{a,\kappa})=b,\qquad
    \Delta_1(Z_{a,\kappa})=a+b.
    \label{eq:gap-bounds}
  \end{equation}
  In particular, for each such fixed $\kappa$,
  \begin{equation}
    \lim_{a\to\infty}\frac{\Delta_1(Z_{a,\kappa})}{c}
    =\frac{1+\kappa}{12}>\frac1{12}.
    \label{eq:asymptotic-gap}
  \end{equation}
  \item \label{item:finite-gap}
  For every fixed real $\delta\ge0$, there are a threshold $a_\delta>0$
  and a family $\widetilde Z_{a,\delta}\in\mathcal U_c$ for every
  real $a\ge a_\delta$ such that
  \begin{equation}
    \Ehat_1(\widetilde Z_{a,\delta})=\delta.
    \label{eq:prescribed-gap}
  \end{equation}
  In particular,
  \begin{equation}
    \lim_{a\to\infty}\Ehat_1(\widetilde Z_{a,\delta})=\delta,
    \qquad
    \lim_{a\to\infty}
    \frac{\Ehat_1(\widetilde Z_{a,\delta})}{a}=0.
    \label{eq:finite-shifted-gap}
  \end{equation}
\end{enumerate}
\end{theorem}

The constant $\kappa_0$ specifies a sufficient positive range, with no
optimal value claimed. The lower thresholds on $a$ may depend on
$\kappa$ or $\delta$.
In part~\ref{item:finite-gap}, the normalized
limit follows from fixed $\delta$ and $a\to\infty$.
Equivalently, $\Delta_1(\widetilde Z_{a,\delta})=a+\delta$,
so its dimension-gap ratio tends to $1/12$. The target $\delta$ is
fixed before taking $a$ large. No identification
with a fixed-$a$ limit as $\kappa\to0^+$ is assumed.

The cylinder energies retain the finite-$c$ shift:
\begin{equation}
  E_1(Z_{a,\kappa})=b-\frac1{12},\qquad
  E_1(\widetilde Z_{a,\delta})=\delta-\frac1{12}.
  \label{eq:cylinder-gap-bound}
\end{equation}
The first family's primary dimension gap exceeds $c/12$ once $b>1/12$.
In the finite-gap family, $\delta=0$ gives scalar threshold primaries with
$h=\bar h=a/2>0$, whereas $\delta=1/12$
places the first scalar level at the classical massless BTZ dimension
$\Delta=c/12$.
For positive $\delta$ in this family, the scalar threshold
coefficient is zero. The chiral support bound remains $(c-1)/24$.

Local finiteness makes the threshold choice substantive. An exact gap
$\Delta_1=a$ must be attained, and uncensored support then forces
$\Ehat=J=0$ with positive integer multiplicity. Such a gap cannot coexist
with zero threshold coefficient. More generally, $\Ehat_1\ge0$ excludes
a negative asymptotic shifted-gap slope within $\Uc$.
The construction allows freedom in the spins and multiplicities at the
first level, as described in Section~\ref{sec:method-overview}.

\paragraph{BTZ entropy from a smoothed density of states.}
The families in part~\ref{item:finite-gap} can also reproduce
the spectral prediction of a perturbative BTZ saddle.
Choose a nonnegative kernel $\phi\in C_c^\infty(\R)$, with fixed
shape and width, normalized by $\int_{\R}\phi(u)\dd u=1$.
For states of cylinder energy $E_s$ and multiplicity $d_s$, define
\[
  N_\phi(E)=\sum_s d_s\,\phi(E_s-E),
  \qquad
  S_\phi(E)=\log N_\phi(E).
\]
The sum includes every spin, all Virasoro descendants and the vacuum
module. Thus $N_\phi$ is a smoothed density of states and
$S_\phi$ its spectral entropy.

\begin{proposition}[BTZ entropy]
\label{prop:btz-entropy}
The families in part~\ref{item:finite-gap} of Theorem~\ref{thm:analytic}
can be chosen so that, for every fixed $\delta\ge0$, the following
expansion holds at the exact center $E=\lambda c$ for every fixed
integer $P\ge0$:
\begin{equation}
  \begin{aligned}
    S_\phi(\lambda c)
      &=2\pi c\sqrt{\frac{\lambda}{3}}-\frac12\log c
         +C_{0,\phi}(\lambda)\\
      &\quad+\sum_{m=1}^{P}\frac{C_{m,\phi}(\lambda)}{c^m}
         +O_P(c^{-P-1}).
  \end{aligned}
  \label{eq:btz-entropy-leading}
\end{equation}
The estimate is uniform for $\lambda$ in any fixed closed interval
contained in $(0,\infty)$ as $a\to\infty$. On each such interval,
the error constant and the lower threshold on $a$ can be chosen
independently of $\lambda$.
The coefficients are those obtained from the single perturbative
BTZ saddle by using the same kernel and summing over all spins,
as defined in Section~\ref{sec:entropy-smooth}.
They depend only on $\phi$ and $\lambda$.
\end{proposition}

In particular,
\begin{equation}
  \lim_{a\to\infty}
    \frac{S_\phi(\lambda c)}{c}
       =2\pi\sqrt{\frac{\lambda}{3}}.
  \label{eq:btz-entropy-rate}
\end{equation}
The result includes $0<E<c/12$, or $c/12<\Delta<c/6$:
the intermediate range whose entropy is not fixed by the general
HKS assumptions~\cite{HartmanKellerStoica2014}.
The construction supplies the additional spectral information here.
The comparison concerns the perturbative contribution of one bulk
saddle at genus one and does not assume its canonical dominance.

Section~\ref{sec:entropy-matching} derives this matching, gives explicit
coefficients, and specifies the freedom in choosing discrete energy levels
without changing these coefficients.

Section~\ref{sec:entropy-positive-kappa} analyzes a specific realization
of the positive-$\kappa$ family in part~\ref{item:positive-fraction}.
It establishes the same all-order BTZ expansion at sufficiently high
energies. With the compact smooth
bump kernel specified there, it also identifies an interval of width
proportional to $c$ above the primary gap in which
the smoothed density exceeds the BTZ prediction exponentially. Thus
the agreement at higher energies coexists with a distinct spectral
regime above the engineered gap.

\subsection{From MWK to a discrete spectrum}
\label{sec:method-overview}

It is convenient to work with the reduced modular function
\begin{equation}
  F(\tau):=\sqrt y\,|\eta(\tau)|^2Z(\tau).
  \label{eq:reduced-function}
\end{equation}
The prefactor is modular invariant. We write $\delta_{(\Ehat_0,J_0)}$
for the unit Dirac measure at $(\Ehat_0,J_0)$. A nonzero
contribution $w\delta_{(\Ehat_0,J_0)}$ to a spectral measure is called
an \emph{atom} of weight $w$; it contributes
$w\sqrt y\exp(-2\pi y\Ehat_0+2\pi ixJ_0)$ to $F$.
In the final nonvacuum primary measure, $w$ is the positive integer
multiplicity at that energy and spin. Signed measures, including atoms
of negative weight, are allowed at intermediate stages of the construction.

\paragraph{The MWK starting point.}
Let $F_{\mathrm{MWK}}$ denote the canonical modular completion of the
Virasoro vacuum \cite{MaloneyWitten2007,KellerMaloney2014}. Its spectral
decomposition has three contributions:
\begin{equation}
  \begin{aligned}
    F_{\mathrm{MWK}}(\tau)
    &=F_{\mathrm{vac}}(\tau)-6\sqrt y\\
    &\quad+\sqrt y\sum_{J\in\Z}e^{2\pi ixJ}
      \int_{\Ehat>|J|}e^{-2\pi y\Ehat}
      \dd\mu_{\mathrm{MWK},J}^{>0}(\Ehat).
  \end{aligned}
  \label{eq:mwk-spectrum}
\end{equation}
Here $F_{\mathrm{vac}}(\tau)=\sqrt y\,e^{2\pi ay}|1-q|^2$ is the
reduced contribution of the Virasoro vacuum module, with multiplicity one.
The term $-6\sqrt y$ is a signed atom of weight $-6$ at
$(\Ehat,J)=(0,0)$, or $h=\bar h=(c-1)/24$. This is the only
nonvacuum atom. We write $\mu_{\mathrm{MWK}}^{>0}$ for the remaining
signed nonvacuum primary measure on positive energies, and
$\mu_{\mathrm{MWK},J}^{>0}$ for its spin-$J$ component. In each
integer-spin sector, this measure has a density with respect
to $\mathrm d\Ehat$ and has support in $\Ehat\ge|J|$. The superscript
$>0$ denotes the restriction to positive energy.

Thus MWK already has exact modular invariance, the required vacuum,
integer spin and uncensored nonvacuum support. Its defects are the
negative threshold atom, the possible negativity of the continuum,
and the absence of a discrete nonvacuum spectrum. In particular, the
continuum becomes negative sufficiently close to $\Ehat=|J|$ at large
odd spin~\cite{BenjaminEtAl2019}. Cancelling the scalar atom therefore
does not establish positivity. Our construction must also replace the
continuum by discrete levels with positive integer multiplicities.

\paragraph{Exact local modular repair.}
For a positive cutoff $B$, let $\nu$ be a finite signed measure
supported in $0<\Ehat\le B$,
$\Ehat\ge|J|$. We construct a modular function $\Repair_B[\nu]$
whose spectral measure in this band is precisely $\nu$. Its full
spectral support lies in the uncensored region $\Ehat\ge|J|$, and
its scalar threshold coefficient vanishes. Above $B$, its spectral
measure admits a density with respect to $\mathrm d\Ehat$ in each
spin sector. The construction also provides quantitative bounds on
this exterior density.

The repair is constructed in two steps. First, we determine the seed
measure whose modular completion produces the desired band measure
$\nu$. Completing a seed adds an induced spectrum to the seed itself,
and this induced contribution generally extends into the band. We must
therefore invert the map from a seed supported in $0<\Ehat\le B$ to
its complete spectral measure in that band. For the modular completion
used in our construction, we establish this inverse on finite signed
measures, including Dirac measures. This allows us to prescribe
discrete spectral contributions as well as spectral densities.

Second, we control the scalar threshold independently of the
positive-energy band. The completion obtained in the first step may
have a nonzero coefficient at $(\Ehat,J)=(0,0)$. We construct an
auxiliary modular function $\mathcal A_B$, called a threshold anchor,
with coefficient one at this point and zero spectral measure in
$0<\Ehat\le B$. Its remaining spectral contribution has uncensored support and is described
by densities above $B$. Subtracting the threshold coefficient times
this anchor cancels the threshold atom while preserving modularity and
the prescribed band measure $\nu$.

\paragraph{Clearing the low-energy band.}
All the initializations below use the same band-clearing operation:
\begin{equation}
  F_B^{\mathrm{clear}}=F_{\mathrm{MWK}}+6\mathcal A_B+
  \Repair_B\bigl[-\mu_{\mathrm{MWK}}^{>0}|_{0<\Ehat\le B}\bigr].
  \label{eq:band-clearing}
\end{equation}
The term $6\mathcal A_B$ cancels the MWK threshold coefficient, while
the zero-threshold repair clears the positive-energy band. Thus
$F_B^{\mathrm{clear}}$ retains the original vacuum, has zero scalar
threshold coefficient, and has no spectral support in $0<\Ehat\le B$.
Its remaining nonvacuum spectral measure is continuous and lies above $B$.
The exterior densities of both corrections enter the subsequent estimates.

\paragraph{Choosing the first level.}
At a target shifted energy $\Ehat_*\ge0$, choose a finite nonempty set of integer
spins $S$ satisfying $|J|\le \Ehat_*$, and positive integer multiplicities
$d_J$. Fix an integer $M_*\ge1$, independent of $a$, bounding their
total. Write the prescribed measure as
\begin{equation}
  \nu_{\mathrm{first}}=\sum_{J\in S}d_J\delta_{(\Ehat_*,J)},
  \qquad d_{\mathrm{first}}=\sum_{J\in S}d_J\le M_*.
  \label{eq:first-level-measure}
\end{equation}
The target energy is $\Ehat_*=b$ in part~\ref{item:positive-fraction} and
$\Ehat_*=\delta$ in part~\ref{item:finite-gap}.
For $0\le \Ehat_*<1$, integer spin
and uncensored support force $S=\{0\}$, while its multiplicity may
be any positive integer. No pairing of $J$ with $-J$ is required.
The family symbols suppress these construction choices. We initialize
the spectrum as follows:
\begin{itemize}[leftmargin=1.5em,itemsep=6pt]
\item \textbf{Part (i): fixed positive $\kappa$.}
Keep $S$ and $d_J$ fixed as $a$ grows. Set $b=\kappa a$, with
$b\ge\max_{J\in S}|J|$, and place the
prescribed first level at the upper endpoint of the cleared band:
\begin{equation}
  F_0=F_b^{\mathrm{clear}}+\Repair_b[\nu_{\mathrm{first}}].
  \label{eq:positive-fraction-initialization}
\end{equation}
The measure inverse permits atoms at this closed upper endpoint,
including nonzero-spin edges $b=|J|$.
Here the cleared-band scale grows with $a$ at fixed $\kappa$.
In the recursive discretization described below, all replacement atoms
have $\Ehat>b$ and uncensored support, so the first level retains
exactly its prescribed spins and multiplicities.

\item \textbf{Part (ii): $\delta>0$.}
Fix $\delta>0$ and choose $B>\delta$, independently
of $a$. For the admissible first-level data at $\Ehat_*=\delta$, set
\begin{equation}
  F_0=F_B^{\mathrm{clear}}+
    \Repair_B[\nu_{\mathrm{first}}].
  \label{eq:positive-gap-initialization}
\end{equation}
The complete measure of $F_0$ in $0<\Ehat\le B$ is precisely the
measure $\nu_{\mathrm{first}}$, and its scalar threshold coefficient is zero.
The cutoff $B$ and the attained
gap $\delta$ are separate quantities.

\item \textbf{Part (ii): $\delta=0$.}
Fix any $B>0$, independently of $a$, and set
\begin{equation}
  F_0=F_B^{\mathrm{clear}}+d_0\mathcal A_B.
  \label{eq:threshold-initialization}
\end{equation}
Here $S=\{0\}$. The anchor supplies scalar threshold multiplicity $d_0$, while the
positive-energy band through $B$ remains empty. The inverse never
acts on an origin atom.
\end{itemize}

\paragraph{Recursive discretization.}
After initialization, the remaining continuum may still be signed.
We discretize it one cell at a time, where a cell is an energy interval
in a single spin sector. At each step, we replace the selected
continuous cell measure by a finite atomic measure with positive integer weights
and matching moments, using the exact modular repair. The repair uses
a single cutoff for all spin sectors, chosen above the current cell
and all previously processed cells. Up to this cutoff, the spectral
change is exactly the prescribed cell replacement; additional continuous
corrections are confined to higher energies. Each step therefore advances
discretization only in the selected spin sector.

\paragraph{Moment matching and integer multiplicities.}
At step $n$, choose a half-open cell $I_n=(L_n,V_n]$ in spin row $j_n$.
We use $J$ for a general spin label and $j_n$ for the spin of the
selected row. We write $j$ when the step index is suppressed.
Here $L_n$ is the upper endpoint of the energy range already processed
in that row. Processing the cell extends this range to $V_n$.
Cells in the same spin row are disjoint.
Let $\mu_n^{\mathrm{cell}}$ be the finite signed continuous measure
on this cell. The upper cell endpoint $V_n$ and moment degree are chosen
so that this measure can be replaced by a positive integer atomic
measure $Q_n$.
The energy positions of its atoms are called \emph{nodes}, and the
atomic weights give the primary multiplicities. We choose every node
in the same half-open cell, $L_n<\Ehat_{n,\ell}\le V_n$, and require
\begin{equation}
  \int (\Ehat-|j_n|)^{r/2}
  \dd\bigl(Q_n-\mu_n^{\mathrm{cell}}\bigr)=0,
  \qquad r=0,\ldots,k_n.
  \label{eq:moment-matching}
\end{equation}
These are polynomial moments in $u=\sqrt{\Ehat-|j_n|}$.
The construction ensures that
\begin{equation}
  \int_{I_n}p(u)\,\dd\mu_n^{\mathrm{cell}}
  =\int_{I_n}p(u)\,\dd Q_n\ge0
  \label{eq:cell-moment-positivity}
\end{equation}
for every polynomial $p$ of degree at most $k_n$ that is nonnegative
on the cell's $u$-interval. The equality follows from
\eqref{eq:moment-matching}; the inequality is required for a
positive atomic replacement. This condition can hold even when
the cell density changes sign.

The reason is that polynomials of bounded degree have a limited ability
to concentrate in regions of negative density while remaining small
in regions of positive density. Polynomial estimates quantify this
limitation. Combined with the large-$c$ density estimates and our
choices of upper cell endpoints and moment degrees, they ensure that the
positive contribution to the integral dominates the possible negative
contribution for every such polynomial. The resulting quantitative lower
bounds on the polynomial moments provide the positivity needed for the
atomic replacement.

We choose the upper cell endpoint $V_n$ so that the total weight
$M_n=\mu_n^{\mathrm{cell}}(I_n)$ of the cell measure is a positive integer.
Together with these lower bounds, this allows a quadrature theorem
to produce $Q_n$ as a sum of exactly $M_n$ unit Dirac measures matching all the
required moments. Their positions may coincide; the weight at each
distinct node is the number of summands supported there.

The node selection must also enforce the strict lower bound
$\Ehat_{n,\ell}>L_n$. Direct minimization over the half-open node
domain need not attain a minimum. We therefore restrict the allowed
node positions to a smaller closed interval contained in $I_n$.
The quantitative lower bounds on polynomial moments described above
also ensure that a sum of $M_n$ unit Dirac measures supported in this
smaller interval can still match all required moments of the entire
continuous cell, including its leftmost portion.
The admissible node vectors then form a nonempty compact set.
Every vector in this set gives a valid $Q_n$; for the explicit
prescription below, a fixed minimization rule selects one of them.

\paragraph{Modular update and spectral protection.}
All these cell replacements have positive-energy support. Choose a
cutoff $B_n$ above both the cell and the previously fixed region and set
\begin{equation}
  \nu_n=Q_n-\mu_n^{\mathrm{cell}},\qquad
  F_{n+1}=F_n+\Repair_{B_n}[\nu_n].
  \label{eq:recursive-update}
\end{equation}
Below $B_n$ this makes exactly the intended replacement and preserves
earlier atoms. Exact moment cancellation controls the exterior error.
In part~\ref{item:finite-gap}, each initial cell in row $j$ starts at
$\max(B,|j|)$, and all primary energies introduced by the replacements
satisfy $\Ehat>B$. Every recursive repair has cutoff $B_n>B$,
zero measure change in $0<\Ehat\le B$, and zero threshold output.
It therefore preserves the first-level measure in
\eqref{eq:positive-gap-initialization} or \eqref{eq:threshold-initialization}.
Thus the spectrum in $0\le\Ehat\le B$ is fixed, and every other
nonvacuum primary lies strictly above $B$.
No origin atom enters a cell residual.

\paragraph{Schedule and convergence.}
Starting from $F_0$, we use \eqref{eq:recursive-update} at every step.
Section~\ref{sec:construction} first states the requirements on the
recursion and then gives an explicit schedule satisfying them.
The schedule specifies which spin row to process and the allowed
cell-endpoint interval, moment degree and repair cutoff at each step.
The explicit schedule begins by processing one cell in each of finitely
many low-spin rows.
The upper endpoints and moment degrees of these initial cells are chosen
so that their high-energy portions supply the quantitative positivity
needed for integer moment matching.
Thereafter, we repeatedly select a row whose processed range has the
lowest upper endpoint and use tail cells whose energy widths are bounded
above and below by fixed positive constants, independently of $c$, the
step and the spin row. Since only finitely many spin
rows meet any bounded energy region, this processing order completes its
discretization in finitely many steps.

A summable error bound keeps these positivity bounds valid for
subsequent cells. Absolute thermal convergence of the complete
corrections then gives a modular limit with a locally finite integer spectrum.
At a finite stage the remaining continuum is still part of $F_n$;
the atomic prefix alone need not be modular.

\paragraph{The two gap regimes.}
For part~\ref{item:positive-fraction}, the cutoff $b=\kappa a$ grows
with $a$, so the construction requires quantitative control of the
inverse and anchor as the cutoff increases.

For part~\ref{item:finite-gap}, keep $B>0$ and $M_*$ fixed as $a$ grows.
The repair and cell moment estimates, together with the uniform reference
and summable tail bounds, give a threshold $a_0(B,M_*)$ such that the
construction is well defined and converges for every real
$a\ge a_0(B,M_*)$. This threshold is uniform over all admissible
first-level data with $\delta\in[0,B)$ and total multiplicity at most
$M_*$.

In both regimes, all later replacement atoms lie above the cleared band,
so the prescribed first-level measure $\nu_{\mathrm{first}}$ remains
unchanged. Section~\ref{sec:global-gaps} uses this preservation to obtain
the exact gaps and their large-$c$ limits.

\section{Exact local modular repair}
\label{sec:local-repair}

The construction requires a modular correction with a prescribed spectrum
in a finite energy band. Completing the desired measure directly would
also induce additional spectral weight in that band. We first solve for
a seed whose complete low-energy output is the desired measure, and then
use a threshold anchor to cancel its scalar threshold coefficient.
This section gives the resulting repair formula and the estimates that
control its exterior spectrum. The full technical proofs are deferred
to Appendices~\ref{app:canonical} and~\ref{app:inversion-repair}.

Throughout this section, the measures are signed primary spectral
measures for the reduced function $F$ of \eqref{eq:reduced-function}.
We use $(\Ehat,J)$ as general spectral coordinates. When distinguishing
input and output in a kernel, $(\Ehat,J)$ denotes the output and
$(\Ehat',J')$ the input. Kernel indices and arguments are ordered with
output first and input second. Both energy variables use the shifted
convention of \eqref{eq:energies}. Fix $B>0$ and define
\begin{equation}
  \mathcal B_B=\{(\Ehat,J):0<\Ehat\le B,\ J\in\Z,
    \ |J|\le\Ehat\}.
  \label{eq:repair-band}
\end{equation}
We write $\mathcal M_B$ for the finite signed measures on this band,
and $\nu_J$ for the spin-$J$ component of $\nu$. The total-variation
norm is $\|\nu\|_{\TV}=\sum_J\|\nu_J\|_{\TV}$.\footnote{The
decomposition $\nu_J=\nu_J^+-\nu_J^-$ into nonnegative measures
concentrated on disjoint measurable sets defines the total-variation measure
$|\nu_J|=\nu_J^++\nu_J^-$ and the norm
$\|\nu_J\|_{\TV}=\int\dd|\nu_J|$.
For a measure consisting of an energy density and atoms at distinct
energies, this is the integral of the absolute density plus the sum
of the absolute atomic weights.}
We use $[\mu]_{\le B}$ for restriction to $\mathcal B_B$.
The band includes its upper boundary $\Ehat=B$ and the nonzero-spin
edges $\Ehat=|J|>0$. The scalar origin is treated separately.

\subsection{Canonical modular completion}
\label{sec:canonical-completion}

\paragraph{The completed seed.}
Let $\Gamma=PSL(2,\Z)$ and let $\Gamma_\infty$ be its translation
subgroup. For $\Ehat'>0$ and integer $J'$ with $|J'|\le\Ehat'$, set
\begin{equation}
  \mathcal P_{\Ehat',J'}(\tau;s)
  =\sum_{\gamma\in\Gamma_\infty\backslash\Gamma}
  (\operatorname{Im}\gamma\tau)^s
  \exp\!\left[-2\pi\Ehat'\operatorname{Im}\gamma\tau
       +2\pi iJ'\operatorname{Re}\gamma\tau\right],
  \qquad \operatorname{Re}s>1.
  \label{eq:poincare-seed}
\end{equation}
We use the canonical meromorphic continuation, regular at $s=1/2$,
with the normalization of Keller--Maloney
\cite[Secs.~3--4]{KellerMaloney2014}, and denote its value there by
$\mathcal P_{\Ehat',J'}(\tau)$. Its direct term is
$\sqrt y\exp(-2\pi y\Ehat'+2\pi ixJ')$.
On each compact subset of $\HH$, the functions
$\mathcal P_{\Ehat',J'}(\tau)$ are bounded uniformly for
$(\Ehat',J')\in\mathcal B_B$. Since every $\sigma\in\mathcal M_B$
has finite total variation, we may integrate the completed seeds against
$\sigma$; the resulting function is well defined and modular invariant.
All seed completions below use this canonical continuation.

To describe the induced continuum, introduce the reference measures
\begin{equation}
  \dd\omega_J(\Ehat)=
  \frac{\dd\Ehat}{\sqrt{\Ehat^2-J^2}},\qquad \Ehat>|J|.
  \label{eq:repair-reference-measure}
\end{equation}
In particular, $\dd\omega_0=\dd\Ehat/\Ehat$.
If $\dd\mu_J=\rho_J\dd\omega_J$, then $\rho_J$ is the density
with respect to $\dd\omega_J$. The density with respect to $\dd\Ehat$ is
$\rho_J/\sqrt{\Ehat^2-J^2}$. We abbreviate the collection of
measures $\rho_J\dd\omega_J$ to $\rho\dd\omega$.
Any atomic contributions are written separately from this continuous part.

For a finite signed seed measure $\sigma$ with bounded uncensored support
and no atom at the origin, the spectral measure of its canonical completion
$\sum_{J'}\int\mathcal P_{\Ehat',J'}(\tau)\dd\sigma_{J'}(\Ehat')$
is
\begin{equation}
  \sigma+(Q\sigma)\dd\omega+t(\sigma)\delta_{(0,0)},
  \qquad
  (Q\sigma)_J(\Ehat)
  =\sum_{J'}\int Q_{J,J'}(\Ehat,\Ehat')\dd\sigma_{J'}(\Ehat').
  \label{eq:canonical-spectral-measure}
\end{equation}
Here $Q\sigma$ denotes the collection of induced densities with respect
to the reference measures $\dd\omega_J$, and $(Q\sigma)\dd\omega$
denotes the induced continuous measure. The scalar threshold coefficient is
\begin{equation}
  t(\sigma)=\sum_{J'} c_{J'}\int\dd\sigma_{J'},\qquad
  c_0=-1,\qquad c_{J'}=2d(|J'|)\quad(J'\ne0).
  \label{eq:completion-threshold}
\end{equation}
Here $d(n)$ is the number of positive divisors of $n$; the coefficients
$c_{J'}$ are independent of the central charge.
The induced measure has uncensored support. Its continuous part is
defined on the open intervals $\Ehat>|J|$ and assigns zero weight to
the points $\Ehat=|J|$. In particular, it produces no atoms at nonzero-spin
edges. The sole induced atom is the scalar term in
\eqref{eq:canonical-spectral-measure}.

\paragraph{The induced kernel.}
Define the Kloosterman sum and its continued zeta function by
\begin{equation}
  S(J,J';m)=\sum_{\substack{r\bmod m\\\gcd(r,m)=1}}
       \exp\!\left[\frac{2\pi i}{m}(Jr+J'r^{-1})\right],
  \qquad
  \mathcal Z_{J,J'}(s)=\sum_{m=1}^{\infty}\frac{S(J,J';m)}{m^{2s}}.
  \label{eq:repair-kloosterman}
\end{equation}
The second series is initially defined for $\operatorname{Re}s>1$
and then continued; $r^{-1}$ is the inverse modulo $m$, and
$S(J,J';1)=1$. The kernel producing the induced density with respect
to $\dd\omega_J$ is
\begin{align}
  Q_{J,J'}(\Ehat,\Ehat')={}&
    2\mathbf 1_{JJ'\ne0}\mathcal Z_{J,J'}(1/2)\nonumber\\
  &+2\sum_{m=1}^{\infty}\frac{S(J,J';m)}m
    \Biggl[
    \cos\!\left(\frac{2\pi}{m}
       \sqrt{(\Ehat'+J')(\Ehat+J)}\right)\nonumber\\[-2pt]
  &\hspace{38mm}\times
    \cos\!\left(\frac{2\pi}{m}
       \sqrt{(\Ehat'-J')(\Ehat-J)}\right)-1
    \Biggr].
  \label{eq:repair-Q-kernel}
\end{align}
The first term is present only when $JJ'\ne0$.
Each cosine of a square root is interpreted by its even entire power
series. For fixed spins $J,J'$, the bracket is $O(m^{-2})$ uniformly
on each compact set of complex energies $(\Ehat,\Ehat')$.
Since $S(J,J';m)$ is a sum of at most $m$ phases of unit modulus,
$|S(J,J';m)|\le m$. Each full summand is therefore also $O(m^{-2})$,
uniformly on the same compact set.
The sum over $m$ thus converges absolutely and uniformly on every
such compact set and defines a function holomorphic in both energy
variables. The subtraction of $1$ inside the bracket is essential
for this $m^{-2}$ decay and must be retained in the sum.

The Fourier calculation fixes both this normalization and the threshold
coefficients in \eqref{eq:completion-threshold}. A joint bound on the
continued arithmetic term gives, for an absolute constant $C$,
\begin{equation}
  |Q_{J,J'}(\Ehat,\Ehat')|
    \le C\log(2+B)\sqrt{\Ehat\Ehat'},\qquad
  0<\Ehat'\le B,\quad |J'|\le\Ehat',\quad |J|\le\Ehat.
  \label{eq:repair-compact-kernel-bound}
\end{equation}
At scalar output, an additional bound controls the behavior near
the threshold $\Ehat=0$:
$|Q_{0,J'}(\Ehat,\Ehat')|\le C\Ehat\Ehat'$.
These bounds make the continuous measure locally integrable, including
at the scalar origin and the nonzero-spin edges, and give absolute
thermal convergence after summing over output spins.
The derivation and the joint arithmetic estimate are assigned to
Appendix~\ref{app:canonical-arithmetic}.

\paragraph{The completion used for inversion.}
We add a modular constant to the canonical completion so that the
low-band operator has a positive quadratic form. Set
\begin{equation}
  \begin{aligned}
    g_J(\Ehat)&=\delta_{J,0}\sqrt{\Ehat},\qquad
    m_g(\sigma)=\int\sqrt{\Ehat'}\dd\sigma_0(\Ehat'),\\
    R_{J,J'}(\Ehat,\Ehat')
      &=Q_{J,J'}(\Ehat,\Ehat')+12g_J(\Ehat)g_{J'}(\Ehat').
  \end{aligned}
  \label{eq:repair-R-kernel}
\end{equation}
The corrected completion is
\begin{equation}
  \mathscr C[\sigma](\tau)
    =\sum_{J'}\int\mathcal P_{\Ehat',J'}(\tau)\dd\sigma_{J'}(\Ehat')
       +\frac{12}{\sqrt2}\,m_g(\sigma).
  \label{eq:corrected-completion}
\end{equation}
The added constant represents a scalar continuum, as follows from
\begin{equation}
  \sqrt y\int_0^\infty (\Ehat)^{-1/2}
                \exp(-2\pi y\Ehat)\dd\Ehat=\frac1{\sqrt2}.
  \label{eq:constant-continuum}
\end{equation}
Thus this contribution replaces the kernel $Q$ by $R$ and carries
no threshold atom. Writing $\mu_F$ for the spectral measure of a
reduced function $F$, we obtain
\begin{equation}
  \mu_{\mathscr C[\sigma]}
    =\sigma+(R\sigma)\dd\omega+t(\sigma)\delta_{(0,0)}.
  \label{eq:corrected-completion-measure}
\end{equation}
The compact-input bound \eqref{eq:repair-compact-kernel-bound} holds
for $R$ as well, after increasing $C$.

The same constant makes the integral of $\mathscr C[\sigma]$ over
the modular fundamental domain, with invariant measure $\dd x\dd y/y^2$,
vanish; see Appendix~\ref{app:band-positivity}.
A related normalization fixes the additive constant in the regularized
$s=1$ Eisenstein series of the scalar Rademacher completion so that
the real part of its regularized modular integral vanishes
\cite[Sec.~4.3]{BenjaminChangFitzpatrickRamella2025}.

\subsection{Inverting the low-energy band}
\label{sec:band-inverse}

The seed needed for a prescribed band measure $\nu$ must solve
\begin{equation}
  w+[(Rw)\dd\omega]_{\le B}=\nu,\qquad w\in\mathcal M_B.
  \label{eq:band-equation}
\end{equation}
This equation concerns the positive-energy band; the threshold
$t(w)$ will be handled in Section~\ref{sec:threshold-anchor}.
We first invert the corresponding operator on functions that are
square-integrable with respect to the reference measures $\dd\omega_J$,
and then use that inverse to solve
\eqref{eq:band-equation} for ordinary measures, including atoms.

\paragraph{Positivity and the Hilbert-space inverse.}
The relevant Hilbert space and operator are
\begin{equation}
  \mathcal H_B=\bigoplus_{|J|<B}
       L^2\bigl((|J|,B),\dd\omega_J\bigr),\qquad
  R_Bf=(R[f\dd\omega])_{\le B},\qquad P_B=I+R_B.
  \label{eq:band-Hilbert-operator}
\end{equation}
Its norm is
$\|f\|_{\mathcal H_B}^2=
\sum_{|J|<B}\int_{|J|}^B|f_J|^2\dd\omega_J$;
the associated inner product is linear in its second argument.
For functions, the subscript $\le B$ restricts each spin component to
the open interval $(|J|,B)$, with $|J|<B$; values at the interval
endpoints do not affect the $L^2$ class. For measures, the same
subscript retains atoms at $\Ehat=B$ and at the nonzero-spin
edges $\Ehat=|J|>0$.

Square integrability alone does not make $f\dd\omega$ a finite
measure. In the scalar channel, $\dd\omega_0=\dd\Ehat/\Ehat$
has infinite measure near zero, so $\int_0^B|f_0|\dd\omega_0$
can diverge even for $f\in\mathcal H_B$. Nevertheless, the
compact-input bound \eqref{eq:repair-compact-kernel-bound} supplies
a factor $\sqrt{\Ehat'}$ in the kernel. Cauchy--Schwarz then
ensures that the kernel integral defining $R_Bf$ converges absolutely.

The kernel is real and symmetric, and the kernel bound above makes
$R_B$ bounded and compact on $\mathcal H_B$. Thus $P_B=I+R_B$ is a
bounded self-adjoint operator.
The correction with coefficient $12$ cancels a negative term in a
decomposition of the quadratic form. After this cancellation,
$\langle f,P_Bf\rangle$ can be expressed as a sum and an integral of
absolute squares, and is therefore nonnegative for every
$f\in\mathcal H_B$.
The following proposition strengthens this nonnegativity to a strictly
positive lower bound, which ensures that $P_B$ has a bounded inverse.

\begin{proposition}[Controlled band inverse]
\label{prop:controlled-band-inverse}
For every $B>0$, $P_B$ is boundedly invertible on $\mathcal H_B$.
There is a finite constant $C_B$ such that
\begin{equation}
  \langle f,P_Bf\rangle\ge C_B^{-1}\|f\|_{\mathcal H_B}^2,
  \qquad \|P_B^{-1}\|_{\mathcal H_B\to\mathcal H_B}\le C_B.
  \label{eq:band-coercivity}
\end{equation}
For $B\ge2$, one may take $C_B\le\exp(CB)$ with an absolute
constant $C$.
\end{proposition}

To obtain this strictly positive lower bound, we use the analytic
dependence of the kernel on energy. Extend $f$ by zero to a larger band.
A small value of
$\langle f,P_Bf\rangle$ then forces the induced response $R[f\dd\omega]$ to
be small on an exterior interval. Analytic propagation bounds its
size on the original band in terms of that exterior response.
Together with $f=P_Bf-R_Bf$, this prevents the quadratic form
from being arbitrarily small on unit vectors. Quantitative propagation
gives the exponential bound in \eqref{eq:band-coercivity}.
Appendices~\ref{app:band-positivity} and
\ref{app:strict-band-inverse} provide the representation of the quadratic form by absolute
squares and the analytic estimates used in this argument.

For $0<B<2$, the same conclusion follows by extending functions in
$\mathcal H_B$ by zero into $\mathcal H_2$ and restricting the quadratic
form of $P_2$. This proves invertibility of the actual operator $P_B$.
The equation is solved on $\mathcal H_B$ itself; compressing
$P_2^{-1}$ would in general give a different operator.
Throughout this section, constants denoted by $C_B$ are finite
at fixed $B$ and may increase between estimates. The exponential
bounds for $B\ge2$ are stated separately because they will be
needed when the repair cutoff grows with the central charge.

\paragraph{The inverse on finite measures.}
Define
\begin{equation}
  \mathscr R_B\nu=[(R\nu)\dd\omega]_{\le B},\qquad
  u=P_B^{-1}(R\nu)_{\le B},\qquad
  \mathcal I_B[\nu]=\nu-u\dd\omega.
  \label{eq:ordinary-measure-inverse}
\end{equation}
Here $u$ is extended by zero outside the band. We must check that
$\mathcal I_B[\nu]$ is a finite signed measure on the band.
The kernel bounds give $(R\nu)_{\le B}\in\mathcal H_B$, so the
inverse from Proposition~\ref{prop:controlled-band-inverse} defines
$u\in\mathcal H_B$. Square integrability of $u$ does not by itself
guarantee finite total variation of $u\dd\omega$, since the scalar
reference measure $\dd\omega_0=\dd\Ehat/\Ehat$ has infinite measure
near zero. The kernel bounds also show that $\mathscr R_B\nu$ and
$(R_Bf)\dd\omega$ have finite total variation for every
$f\in\mathcal H_B$. Taking $f=u$ and using
$P_Bu=(R\nu)_{\le B}$, we obtain
\begin{equation}
  u\dd\omega=\mathscr R_B\nu-(R_Bu)\dd\omega
      \in\mathcal M_B.
  \label{eq:inverse-finite-variation}
\end{equation}
Since $\nu\in\mathcal M_B$, it follows that
$\mathcal I_B[\nu]=\nu-u\dd\omega\in\mathcal M_B$.

Substituting \eqref{eq:ordinary-measure-inverse} into
\eqref{eq:band-equation} gives
\begin{equation}
  (I+\mathscr R_B)\mathcal I_B[\nu]=\nu,\qquad
  \|\mathcal I_B[\nu]\|_{\TV}\le C_B\|\nu\|_{\TV}.
  \label{eq:inverse-measure-bound}
\end{equation}
A homogeneous solution satisfies $w=-\mathscr R_Bw$, so it has a
density $f\in\mathcal H_B$ with $w=f\dd\omega$ and vanishes by
Proposition~\ref{prop:controlled-band-inverse}.
Thus $\mathcal I_B=(I+\mathscr R_B)^{-1}$ on $\mathcal M_B$.
For $B\ge2$, the constant in \eqref{eq:inverse-measure-bound}
can again be chosen at most $\exp(CB)$.

In each spin sector, the subtracted measure $u\dd\omega$ has the
integrable density $u_J(\Ehat)/\sqrt{\Ehat^2-J^2}$ with respect to
$\dd\Ehat$. Thus $\mathcal I_B$ changes only the part of $\nu$
represented by energy densities. In particular, every atom of $\nu$
in the band retains its position and weight, including atoms at
$\Ehat=B$ or $\Ehat=|J|>0$. When $B$ is an integer, the rows
$|J|=B$ have no open interval in $\mathcal H_B$, but their atoms
are still included in the input measure $\nu$ and in the kernel
integral defining $R\nu$.
The inverse solves a signed cancellation equation. Positivity of
the final atomic spectrum will come from the moment construction.

An explicit operator representation is also available. Choose
$\lambda_B>\|P_B\|_{\mathcal H_B\to\mathcal H_B}$.
Appendix~\ref{app:strict-band-inverse} gives an explicit admissible choice from the kernel bounds.
The strictly positive lower bound in
Proposition~\ref{prop:controlled-band-inverse} then gives
\begin{equation}
  P_B^{-1}=\frac1{\lambda_B}
      \sum_{r=0}^{\infty}(I-P_B/\lambda_B)^r,
  \qquad \|I-P_B/\lambda_B\|<1.
  \label{eq:band-inverse-series}
\end{equation}
The series converges in operator norm on $\mathcal H_B$.
Applying this series to $(R\nu)_{\le B}$ gives the function $u$
in \eqref{eq:ordinary-measure-inverse}.

\subsection{The threshold anchor}
\label{sec:threshold-anchor}

The completion of $w=\mathcal I_B[\nu]$ has the correct band measure
but generally has threshold coefficient $t(w)$. To adjust that
coefficient independently, we construct a modular function
$\mathcal A_B$ whose spectral measure for $\Ehat\le B$ consists
of a single unit atom at $(\Ehat,J)=(0,0)$.

\paragraph{A seed outside the band.}
Start with a scalar seed above $B$. Its induced low-band measure can
be removed using $\mathcal I_B$. For a unit seed at $(\Ehat',J')$
above the band, the threshold left after this removal is
\begin{equation}
  \zeta_B(\Ehat',J')=c_{J'}-
    t\!\left(\mathcal I_B
      \bigl[[R_{\cdot,J'}(\cdot,\Ehat')\dd\omega]_{\le B}\bigr]
      \right).
  \label{eq:anchor-response}
\end{equation}
The input to $\mathcal I_B$ is the finite signed measure induced
in $\mathcal B_B$ by this unit seed.
The formula also defines $\zeta_B$ for other positive
uncensored inputs. For any finite seed $\psi$ with bounded support above $B$,
linearity gives
\begin{equation}
  t\!\left(\psi-\mathcal I_B[\mathscr R_B\psi]\right)
    =\sum_{J'}\int\zeta_B(\Ehat',J')\dd\psi_{J'}(\Ehat').
  \label{eq:anchor-response-identity}
\end{equation}
In this equation, $\mathscr R_B\psi$ means the same low-output
restriction as in \eqref{eq:ordinary-measure-inverse}, now with
a bounded seed outside the band.

To normalize the anchor, we need $\zeta_B(\Ehat',0)$ to be nonzero
somewhere on $[2B,3B]$. Define $G_B(z)=\zeta_B(Bz^2,0)$ for real
$z>0$. Appendix~\ref{app:threshold-anchor} shows that $G_B$ extends analytically to the
whole complex plane and satisfies $G_B(0)=-1$.
If $\zeta_B(\Ehat',0)$ vanished throughout $[2B,3B]$, then $G_B$
would vanish on $[\sqrt2,\sqrt3]$. The identity theorem would force
$G_B$ to vanish identically, contradicting $G_B(0)=-1$.
Since the scalar response is real and continuous on positive energies,
we therefore have, for every $B>0$,
\begin{equation}
  N_B=\int_{2B}^{3B}\zeta_B(\Ehat',0)^2
                           \frac{\dd\Ehat'}{\Ehat'}>0.
  \label{eq:anchor-normalization-integral}
\end{equation}

\paragraph{Band cancellation and unit threshold.}
Choose the signed scalar seed
\begin{equation}
  \dd\psi_{B,J'}(\Ehat')=
    \delta_{J',0}\mathbf1_{[2B,3B]}(\Ehat')
       \frac{\zeta_B(\Ehat',0)}{N_B}\frac{\dd\Ehat'}{\Ehat'},
  \qquad
  v_B=\mathcal I_B[\mathscr R_B\psi_B],\qquad
  \theta_B=\psi_B-v_B.
  \label{eq:anchor-seed-construction}
\end{equation}
The inverse equation \eqref{eq:inverse-measure-bound} gives
$v_B+\mathscr R_Bv_B=\mathscr R_B\psi_B$.
Since $\psi_B$ is supported on the scalar interval $[2B,3B]$ and
$v_B$ is supported in $\mathcal B_B$, the band measure of the
completion of $\theta_B$ is
\begin{equation}
  [\theta_B+(R\theta_B)\dd\omega]_{\le B}
    =-v_B+\mathscr R_B\psi_B-\mathscr R_Bv_B=0.
  \label{eq:anchor-band-cancellation}
\end{equation}
This cancellation holds in every spin sector, including the
nonzero-spin components induced by the scalar seed $\psi_B$.

For the threshold coefficient, applying
\eqref{eq:anchor-response-identity} to $\psi_B$ and substituting
its density from \eqref{eq:anchor-seed-construction} gives
\begin{equation}
  \begin{aligned}
    t(\theta_B)
      &=\int_{2B}^{3B}\zeta_B(\Ehat',0)\dd\psi_{B,0}(\Ehat')\\
      &=\frac1{N_B}\int_{2B}^{3B}
          \zeta_B(\Ehat',0)^2\frac{\dd\Ehat'}{\Ehat'}=1.
  \end{aligned}
  \label{eq:anchor-unit-threshold}
\end{equation}
The last equality uses the definition of $N_B$ in
\eqref{eq:anchor-normalization-integral}.
Thus the normalization fixes the threshold of the completed seed.
It does not prescribe the total weight of $\psi_B$.

The modular threshold anchor is
\begin{equation}
  \mathcal A_B=\mathscr C[\theta_B].
  \label{eq:modular-threshold-anchor}
\end{equation}
The seed $\theta_B$ has an integrable density with respect to
$\dd\Ehat$ in each spin sector: $\psi_B$ and $\mathscr R_B\psi_B$
have such densities, and $\mathcal I_B$ preserves this property.
It is supported in $\mathcal B_B$ together with the scalar interval
$[2B,3B]$. Its completed spectrum has a unit atom at $(0,0)$
and zero measure in $\mathcal B_B$. Above $B$, its spectral
measure has an energy density in each spin sector and no atoms.
All of this support is uncensored. The energy densities of
$\theta_B$ and those of the completed anchor $\mathcal A_B$
above $B$ may take either sign.

We will need control of the cost of this normalization. Write the scalar
density of the seed $\psi_B$ with respect to $\dd\omega_0$ as
\begin{equation}
  \rho_B^\psi(\Ehat)
     =\mathbf1_{[2B,3B]}(\Ehat)\frac{\zeta_B(\Ehat,0)}{N_B},
  \qquad
  \dd\psi_{B,J}=\delta_{J,0}\rho_B^\psi\,\dd\omega_J.
  \label{eq:anchor-direct-numerator}
\end{equation}
Appendix~\ref{app:threshold-anchor} proves the following bound for every fixed $B>0$:
\begin{equation}
  \|\theta_B\|_{\TV}
   +\sum_{J'}\int\sqrt{\Ehat'}\dd|\theta_{B,J'}|(\Ehat')
   +\sup_{\Ehat}\bigl|\rho_B^\psi(\Ehat)\bigr|
       \le C_B<\infty.
  \label{eq:anchor-size-bound}
\end{equation}
For $B\ge2$, one may take $C_B\le\exp(CB)$, with a constant $C$
independent of $B$.

\subsection{Exact repair and exterior control}
\label{sec:repair-exterior}

\paragraph{The repair formula.}
For a desired signed band change $\nu\in\mathcal M_B$, define
\begin{equation}
  w=\mathcal I_B[\nu],\qquad
  \sigma_B[\nu]=w-t(w)\theta_B,\qquad
  \Repair_B[\nu]=\mathscr C[\sigma_B[\nu]]
                =\mathscr C[w]-t(w)\mathcal A_B.
  \label{eq:exact-repair-formula}
\end{equation}
The inverse determines the positive-energy band and the anchor removes
the resulting threshold coefficient. Both operations are linear in
$\nu$.

\begin{proposition}[Exact local modular repair]
\label{prop:exact-local-repair}
For every $B>0$ and $\nu\in\mathcal M_B$, the function
$\Repair_B[\nu]$ is modular invariant and has spectral support only
in the uncensored region $\Ehat\ge |J|$.
Its band measure and scalar threshold are
\begin{equation}
  [\mu_{\Repair_B[\nu]}]_{\le B}=\nu,\qquad
  \mu_{\Repair_B[\nu]}(\{(0,0)\})=0.
  \label{eq:exact-repair-identities}
\end{equation}
Above $B$, its measure has a density with respect to $\dd\Ehat$
in each spin sector.
\end{proposition}

The band and threshold identities \eqref{eq:exact-repair-identities}
follow by substituting
\eqref{eq:anchor-band-cancellation} and \eqref{eq:anchor-unit-threshold} into
\eqref{eq:exact-repair-formula}:
\begin{align*}
  [\sigma_B[\nu]+(R\sigma_B[\nu])\dd\omega]_{\le B}
    &=\nu-t(w)[\theta_B+(R\theta_B)\dd\omega]_{\le B}=\nu,\\
  t(\sigma_B[\nu])&=t(w)-t(w)t(\theta_B)=0.
\end{align*}
Modularity follows from \eqref{eq:corrected-completion}.
Both $w$ and $\theta_B$ have support only in $\Ehat\ge |J|$,
and the completion preserves this support condition.
Adding the repair to $F_{\mathrm{MWK}}$ therefore leaves the Virasoro vacuum
contribution unchanged.
Above $B$, both $w$ and $v_B$ have no support.
The direct seed contribution there is therefore $-t(w)\psi_B$,
which has an energy density in each spin sector.
The induced measure $(R\sigma_B[\nu])\dd\omega$ also has an energy
density in each spin sector.

\paragraph{The complete exterior density.}
For $\Ehat>B$, write
$\dd\mu_{\Repair_B[\nu],J}=r_{B,J}[\nu](\Ehat)\dd\omega_J$.
The full density with respect to $\dd\omega_J$ is
\begin{equation}
  r_{B,J}[\nu](\Ehat)
    =-t(w)\delta_{J,0}\rho_B^\psi(\Ehat)
       +(R\sigma_B[\nu])_J(\Ehat),\qquad
       \Ehat>B,\quad\Ehat>|J|.
  \label{eq:complete-repair-exterior}
\end{equation}
The first term is the direct anchor density on $[2B,3B]$; the second
is the induced continuum of the complete correcting seed. Both enter
the exterior estimates. The density with respect to $\dd\Ehat$
is $r_{B,J}[\nu](\Ehat)/\sqrt{\Ehat^2-J^2}$.

The inverse bound \eqref{eq:inverse-measure-bound} controls
$\|w\|_{\TV}$. Since $|J'|\le B$ on the support of $w$, the
definition of $t$ also gives $|t(w)|\le C_B\|\nu\|_{\TV}$.
Combining these bounds with \eqref{eq:anchor-size-bound} yields
\[
  \sum_{J'}\int\sqrt{\Ehat'}\,
    \dd|\sigma_{B,J'}[\nu]|(\Ehat')
       \le C_B\|\nu\|_{\TV}.
\]
The seed has energies at most $3B$. Applying the kernel bound
\eqref{eq:repair-compact-kernel-bound} to $R$ with input cutoff
$3B$ therefore bounds the induced term by
$C_B\|\nu\|_{\TV}\sqrt{\Ehat}$.
The direct term is bounded by $C_B\|\nu\|_{\TV}$ using
\eqref{eq:anchor-size-bound}. Since $\Ehat>B$, it satisfies the
same square-root bound after increasing $C_B$.
We therefore obtain the uniform-in-spin estimate
\begin{equation}
  |r_{B,J}[\nu](\Ehat)|
     \le C_B\|\nu\|_{\TV}\sqrt{\Ehat},\qquad
       \Ehat>B,\quad\Ehat>|J|.
  \label{eq:general-exterior-bound}
\end{equation}
For $B\ge2$, $C_B$ can be taken at most $\exp(CB)$.

For every fixed $\beta=2\pi\operatorname{Im}\tau>0$, the complete
repair is absolutely thermally convergent:
\begin{equation}
  \sum_{J\in\Z}\int e^{-\beta\Ehat}\,
    \dd|\mu_{\Repair_B[\nu],J}|(\Ehat)<\infty.
  \label{eq:repair-thermal-convergence}
\end{equation}
The band contribution is bounded by $\|\nu\|_{\TV}$.
For the exterior, writing $\Ehat=|J|+s$ and using
\eqref{eq:general-exterior-bound} bounds each spin integral by
$C_B\|\nu\|_{\TV}\sqrt{\pi/\beta}\,e^{-\beta|J|}$.
The integral is finite at the spin edge, and the remaining spin sum
is geometric. Appendix~\ref{app:canonical-finite-seeds} gives the integral estimate.

This establishes convergence of each fixed repair. Convergence of the
infinite recursion requires additional control of successive corrections,
supplied by moment matching.

\paragraph{How moment matching reduces the exterior change.}
Suppose $\nu$ lies in a single input spin row $J'$ on a cell $(L,V]$
with $0<L<V<B$ and $|J'|\le L$, and suppose
\begin{equation}
  u=\sqrt{\Ehat'-|J'|},\qquad
  \int u^r\dd\nu_{J'}(\Ehat')=0\quad(0\le r\le k),\qquad k\ge1.
  \label{eq:repair-zero-moments}
\end{equation}
By \eqref{eq:completion-threshold}, the zeroth-moment condition gives
$t(\nu)=0$. The first-moment condition also cancels the scalar rank-one
contribution in $R$.

With $v_\nu=\mathcal I_B[\mathscr R_B\nu]$, equations
\eqref{eq:inverse-measure-bound} and \eqref{eq:exact-repair-formula}
give $w=\nu-v_\nu$, so the correcting seed takes the form
\begin{equation}
  \sigma_B[\nu]=\nu-v_\nu+t(v_\nu)\theta_B.
  \label{eq:moment-repair-decomposition}
\end{equation}
Using \eqref{eq:inverse-measure-bound}, the definition of $t$, and
\eqref{eq:anchor-size-bound}, we obtain
\[
  \|v_\nu\|_{\TV}+\|t(v_\nu)\theta_B\|_{\TV}
     \le C_B\|\mathscr R_B\nu\|_{\TV}.
\]
This controls the total variation of the additional seed measure
$\sigma_B[\nu]-\nu$. Its input energies are at most $3B$, so applying
the kernel bound \eqref{eq:repair-compact-kernel-bound} to
$R(\sigma_B[\nu]-\nu)$ bounds its induced density by
$C_B\sqrt{\Ehat}\|\mathscr R_B\nu\|_{\TV}$.
The direct anchor term in \eqref{eq:complete-repair-exterior} has
coefficient $t(v_\nu)$ and satisfies the same bound for $\Ehat>B$,
by \eqref{eq:anchor-size-bound}. Combining both contributions gives
\begin{equation}
  |r_{B,J}[\nu](\Ehat)|
    \le |(R\nu)_J(\Ehat)|
       +C_B\sqrt{\Ehat}\,\|\mathscr R_B\nu\|_{\TV},
  \qquad \Ehat>B,\quad \Ehat>|J|.
  \label{eq:moment-exterior-bound}
\end{equation}
We therefore need to control $R\nu$ both in the band, to estimate
$\|\mathscr R_B\nu\|_{\TV}$, and above $B$, to estimate the
remaining exterior term. The moment conditions give both estimates
through polynomial approximation of the kernel.

For fixed output $(\Ehat,J)$, regard the input kernel as a function
of $u$ on the full closed interval
$[u_-,u_+]=[\sqrt{L-|J'|},\sqrt{V-|J'|}]$. For every polynomial
$p$ of degree at most $k$, the moment conditions
\eqref{eq:repair-zero-moments} and the total-variation bound give
\begin{equation}
  \begin{aligned}
    |(R\nu)_J(\Ehat)|
      &=\left|\int_{(L,V]}
          \bigl[R_{J,J'}(\Ehat,\Ehat')-p(u)\bigr]
          \,\dd\nu_{J'}(\Ehat')\right|\\
      &\le\|\nu\|_{\TV}
        \sup_{u\in[u_-,u_+]}
        \bigl|R_{J,J'}(\Ehat,|J'|+u^2)-p(u)\bigr|.
  \end{aligned}
  \label{eq:kernel-moment-approximation}
\end{equation}
For a fixed cell and output $(\Ehat,J)$, the analyticity of the
kernel in $u$ yields polynomial approximations on $[u_-,u_+]$ whose
errors decrease geometrically with $k$. Appendix~\ref{app:complete-repair-estimates} gives quantitative
bounds that track the dependence on the cell and output variables.

We apply these approximation bounds to both terms in the exterior
estimate \eqref{eq:moment-exterior-bound}.
The recursive construction chooses the cell widths, moment degrees,
and repair cutoffs together so that the resulting exterior corrections
are summable, as proved in Appendix~D.

\section{Constructing a discrete integer spectrum}
\label{sec:construction}

We now use the local repair to replace the MWK continuum by discrete
primary levels. We first specify the initial spectral data and describe
the local steps, allowing the cell endpoints, moment degrees and atomic
replacements to vary. We then state the requirements on the full
recursion and give one explicit schedule. Section~5 verifies these
requirements for that schedule and proves convergence to a discrete
integer spectrum, with the technical estimates supplied in Appendix~D.
Throughout, the real parameter $a$ and $c=12a+1$ are fixed during
the recursion.

\subsection{Initial spectra}
\label{sec:construction-initial}

We use $B>0$ for the energy through which the initial continuum is
cleared. In part~\ref{item:positive-fraction} of
Theorem~\ref{thm:analytic}, this cutoff is $B=b=\kappa a$;
in part~\ref{item:finite-gap}, choose $B>\delta$ independently of $a$.
The later repair cutoffs are denoted by $B_n$.
We use the initial functions
\eqref{eq:positive-fraction-initialization}--\eqref{eq:threshold-initialization}
from Section~\ref{sec:method-overview}.

Let $\mu_0$ denote the nonvacuum primary spectral measure of the
initialized function $F_0$. We write
\begin{equation}
  \mu_0=\mu_0^{\mathrm{at}}+\rho_0\dd\omega,
  \label{eq:construction-initial-measure}
\end{equation}
where $\mu_0^{\mathrm{at}}$ is the sum of its Dirac contributions
and $\rho_0\dd\omega$ is its continuous part, using the reference
measures \eqref{eq:repair-reference-measure}. In each initialization,
$\mu_0^{\mathrm{at}}=\nu_{\mathrm{first}}$ is the prescribed first-level
measure of \eqref{eq:first-level-measure}. Its total multiplicity is
at most the fixed bound $M_*$.

\paragraph{Data for the recursion.}
In each case, $\rho_{0,J}(\Ehat)=0$ for $|J|<\Ehat\le B$.
The initial densities above $B$ follow by adding the continuous
spectral contributions of the terms defining $F_0$.
The anchor contributes the direct density of
\eqref{eq:anchor-direct-numerator} and its induced continuum;
each repair contributes the full density
\eqref{eq:complete-repair-exterior}. These contributions can have
either sign.

For each spin $J$, let $f_{n,J}$ denote the upper endpoint of the
energy range already processed in that row before step $n$. Set
\begin{equation}
  f_{0,J}=\max\{B,|J|\},\qquad J\in\Z.
  \label{eq:construction-initial-fronts}
\end{equation}
The continuum vanishes between the spin threshold and this endpoint.
Every prescribed first-level atom has energy at most $B$.

\subsection{Integer moment matching on a cell}
\label{sec:construction-moments}

Consider the current continuous measure $\rho_j\dd\omega_j$ in
one spin row $j$. Here $j$ denotes the selected spin $j_n$ with
the step index suppressed. For this subsection, take a left endpoint
$L\ge|j|$ with $L>0$, an allowed endpoint interval
$[V^-,V^+]$ with $L<V^-<V^+$, and a moment degree $k\ge1$
as given. The cell will be $I=(L,V]$, with $V$ chosen in that
interval. We regard its measure as supported in the single spin row
$j$ when combining it with measures on the full spectrum.

\paragraph{An endpoint with integer weight.}
Define
\begin{equation}
  m(v)=\int_{(L,v]}\rho_j(\Ehat)\dd\omega_j(\Ehat),
  \qquad M=\lceil m(V^-)\rceil,
  \qquad m(V)=M,\quad V\in[V^-,V^+].
  \label{eq:construction-integer-endpoint}
\end{equation}
We require $m(V^-)>0$, $\rho_j>0$ almost everywhere on
$[V^-,V^+]$, and $m(V^+)-m(V^-)>1$. The first condition makes
$M$ a positive integer. The second makes $m$ strictly increasing
on the endpoint interval; together with continuity and the third
condition, this gives a unique solution of $m(V)=M$. If $m(V^-)$
is already an integer, the rule selects $V=V^-$. Positivity is
required only on the movable endpoint interval; the density
elsewhere in the cell may still change sign.

\paragraph{The finite moment functional.}
Let $\mu^{\mathrm{cell}}$ be the restriction of
$\rho_j\dd\omega_j$ to $I$. Normalize the square-root energy
coordinate to the unit interval:
\begin{equation}
  \begin{aligned}
    u_-&=\sqrt{L-|j|},\qquad u_+=\sqrt{V-|j|},\qquad
    s(\Ehat)=\frac{\sqrt{\Ehat-|j|}-u_-}{u_+-u_-},\\
    \Lambda(p)&=\int_{(L,V]}p\bigl(s(\Ehat)\bigr)
                      \dd\mu^{\mathrm{cell}}(\Ehat).
  \end{aligned}
  \label{eq:construction-cell-functional}
\end{equation}
Write $\Pi_k$ for the real polynomials of degree at most $k$.
Matching $\Lambda$ on $\Pi_k$ is equivalent to matching moments
in $u=\sqrt{\Ehat-|j|}$ through degree $k$, since the change
from $u$ to $s$ is affine. For these polynomials, define
\begin{equation}
  \operatorname{Var}_{[0,1]}p=\int_0^1|p'(s)|\dd s.
  \label{eq:construction-polynomial-variation}
\end{equation}
The integer $M$ fixes the total multiplicity of the atomic
replacement. We seek $M$ atoms of unit weight whose moments agree
with $\Lambda$ through degree $k$, with every energy strictly
above $L$. The following proposition gives a sufficient lower
bound on $\Lambda(p)$ in terms of the polynomial variation.
This lower bound is strong enough both to obtain unit weights
and to restrict the allowed positions to a closed interval
separated from the left endpoint.

\begin{proposition}[Unit-weight moment matching]
\label{prop:construction-integer-matching}
Let $k\ge1$ and let $\Lambda$ be a real linear functional on
$\Pi_k$ satisfying
\begin{equation}
  \Lambda(1)=M\in\Z_{>0},\qquad
  \Lambda(p)\ge\operatorname{Var}_{[0,1]}p
  \quad\text{for every }p\in\Pi_k\text{ nonnegative on }[0,1].
  \label{eq:construction-moment-reserve}
\end{equation}
Set $h=[64M(k+1)^3]^{-1}$. There are $M$ points
$s_1,\ldots,s_M\in[h,1]$, with repetitions allowed, such that
\begin{equation}
  \sum_{\ell=1}^M p(s_\ell)=\Lambda(p)
  \qquad\text{for every }p\in\Pi_k.
  \label{eq:construction-unit-moments}
\end{equation}
\end{proposition}

We outline the proof of Proposition~\ref{prop:construction-integer-matching},
leaving the quantitative details to Appendix~\ref{app:moment-designs}. First, consider the
normalized moment vector $(\Lambda(s),\ldots,\Lambda(s^k))/M$.
If it lay outside the convex hull of the curve $(s,\ldots,s^k)$,
$0\le s\le1$, a separating hyperplane would give a polynomial
$p\in\Pi_k$ that is nonnegative on $[0,1]$ but satisfies
$\Lambda(p)<0$, contradicting \eqref{eq:construction-moment-reserve}.
The vector is therefore a convex combination of points on this curve.
This gives a positive atomic measure on $[0,1]$ with total weight $M$
and the prescribed moments, even when the original cell density
changes sign. The weights in this auxiliary representation need not
be integers.

Next, we seek a representing measure supported on $[h,1]$, so that
the corresponding energies lie strictly above $L$. The functional
$\Lambda$ remains that of the entire original cell, including $0<s<h$.
To repeat the separation argument on $[h,1]$, we need $\Lambda(p)\ge0$
for every $p\in\Pi_k$ nonnegative on this smaller interval.
Such polynomials may be negative on $[0,h)$, so
\eqref{eq:construction-moment-reserve} does not apply to them directly.

A derivative estimate controls how far these polynomials can fall
below zero on $[0,h)$. For the stated choice of $h$, this estimate
and \eqref{eq:construction-moment-reserve} give the required positivity.
Applying the separation argument on $[h,1]$ then supplies a positive
representing measure with the same total weight $M$ and moments.

Finally, we apply Kane's interval design theorem
\cite[Prop.~20]{Kane2012} to this positive measure normalized by $M$.
The quantitative variation bound retained on $[h,1]$ ensures that
the theorem permits exactly $M$ points, which yield the unit-weight
moment identities \eqref{eq:construction-unit-moments} after restoring
total weight $M$.

\paragraph{Choosing atomic positions.}
For the functional \eqref{eq:construction-cell-functional}, form
the compact set
\begin{equation}
  \mathcal X=
  \left\{(s_1,\ldots,s_M)\in\R^M:
    \begin{array}{l}
      h\le s_1\le\cdots\le s_M\le1,\\[2pt]
      \displaystyle\sum_{\ell=1}^M s_\ell^r=\Lambda(s^r),
                       \quad 1\le r\le k
    \end{array}\right\},
  \qquad h=\frac1{64M(k+1)^3}.
  \label{eq:construction-node-set}
\end{equation}
Proposition~\ref{prop:construction-integer-matching} makes this
set nonempty. Choose any vector in $\mathcal X$, and define
\begin{equation}
  \Ehat_\ell=|j|+
      \bigl[u_-+(u_+-u_-)s_\ell\bigr]^2,
  \qquad
  Q=\sum_{\ell=1}^M\delta_{(\Ehat_\ell,j)}.
  \label{eq:construction-cell-atoms}
\end{equation}
Every new energy satisfies $L<\Ehat_\ell\le V$, including
when $L=|j|$. Coincident positions give a primary multiplicity
equal to their number of occurrences. An atom at $V$ belongs to
this cell; the next cell in the same row excludes that endpoint.

Every vector in $\mathcal X$ yields an atomic measure with the
required support, positive integer multiplicities and moments.
The value of $h$ in Proposition~\ref{prop:construction-integer-matching}
is one sufficient choice for ensuring that this set is nonempty;
Appendix~\ref{app:moment-selector} gives the conditions under which $h$ can be varied.

\paragraph{An explicit atom selector.}
For the explicit prescription used below, we keep this $h$ and choose
the lexicographically least vector in $\mathcal X$: minimize the first
coordinate, then the second among vectors with the chosen first
coordinate, and continue through all $M$ coordinates. Each minimum is
attained on a nonempty compact set, so this fixes one vector uniquely.
Other choices in $\mathcal X$ satisfy the same local requirements but
can change the individual energies and multiplicities.

\subsection{The recursive update}
\label{sec:construction-update}

At step $n\ge0$, the state consists of $F_n$, its nonvacuum measure
$\mu_n$, and the processed endpoints $f_{n,J}$. We maintain
\begin{equation}
  \mu_n=\mu_n^{\mathrm{at}}+\rho_n\dd\omega,
  \qquad
  \mu_n^{\mathrm{at}}=\mu_0^{\mathrm{at}}+
       \sum_{m=0}^{n-1}Q_m,
  \qquad
  \rho_{n,J}(\Ehat)=0\quad(|J|<\Ehat\le f_{n,J}).
  \label{eq:construction-spectral-state}
\end{equation}
The function and measure descriptions are related by
\begin{equation}
  F_n(\tau)=F_{\mathrm{vac}}(\tau)+
    \sqrt y\sum_{J\in\Z}e^{2\pi ixJ}
    \int_{[|J|,\infty)}e^{-2\pi y\Ehat}
                       \dd\mu_{n,J}(\Ehat).
  \label{eq:construction-state-transform}
\end{equation}
The measure integral includes the prescribed threshold atom when
$\delta=0$; the continuous part uses $\dd\omega_J$ on
$\Ehat>|J|$. At every finite stage, that continuum remains part
of the modular function.

Choose a spin row $j_n$ and set $L_n=f_{n,j_n}$. Applying
Section~\ref{sec:construction-moments} with an endpoint interval
and degree satisfying its local conditions gives $V_n$, the cell
$I_n=(L_n,V_n]$, and its replacement $Q_n$. The cell contains
no earlier atom: all atoms previously inserted in this row have
energy at most $L_n$, while every prescribed first-level atom has
energy at most $B\le L_n$. Define the residual by
\begin{equation}
  \mu_n^{\mathrm{cell}}
      =\bigl[\rho_{n,j_n}\dd\omega_{j_n}\bigr]_{I_n},
  \qquad
  \nu_n=Q_n-\mu_n^{\mathrm{cell}}.
  \label{eq:construction-cell-residual}
\end{equation}
It is supported in the selected spin row and satisfies
\begin{equation}
  \begin{aligned}
    \int (\Ehat-|j_n|)^{r/2}\dd\nu_n(\Ehat)&=0,
                       &&0\le r\le k_n,\\
    \|\nu_n\|_{\TV}&\le
          2\|\mu_n^{\mathrm{cell}}\|_{\TV}.&&
  \end{aligned}
  \label{eq:construction-residual-properties}
\end{equation}
The second line follows because the positive measure $Q_n$ has
total weight $M_n=\mu_n^{\mathrm{cell}}(I_n)$.
The first line supplies exactly the cancellations used in
\eqref{eq:repair-zero-moments}--\eqref{eq:kernel-moment-approximation}.

\paragraph{Updating the function and spectrum.}
Take a repair cutoff satisfying
\begin{equation}
  B_n>\max\{B,V_0,\ldots,V_n\}.
  \label{eq:construction-protecting-cutoff}
\end{equation}
Then $\nu_n\in\mathcal M_{B_n}$.
By Proposition~\ref{prop:exact-local-repair}, the update has
the following spectral description:
\begin{equation}
  \begin{aligned}
    F_{n+1}&=F_n+\Repair_{B_n}[\nu_n],\\
    \mu_{n+1}^{\mathrm{at}}&=\mu_n^{\mathrm{at}}+Q_n,\\
    \rho_{n+1,J}(\Ehat)
      &=\bigl(1-\delta_{J,j_n}\mathbf1_{I_n}(\Ehat)\bigr)
                                      \rho_{n,J}(\Ehat)\\
      &\quad+\mathbf1_{(B_n,\infty)}(\Ehat)
                                      r_{B_n,J}[\nu_n](\Ehat).
  \end{aligned}
  \label{eq:construction-full-update}
\end{equation}
Here $r_{B_n,J}[\nu_n]$ is the full exterior density
\eqref{eq:complete-repair-exterior}, including the direct anchor
term. The last term is zero for $\Ehat\le B_n$.
The density identity holds almost everywhere with respect to
$\dd\omega_J$. Finally, advance only the selected row:
\begin{equation}
  f_{n+1,J}=
  \begin{cases}
    V_n,&J=j_n,\\
    f_{n,J},&J\ne j_n.
  \end{cases}
  \label{eq:construction-front-update}
\end{equation}

\paragraph{Requirements on the full recursion.}
Starting from one of the initial spectra in
Section~\ref{sec:construction-initial}, every step must satisfy the
local endpoint and moment conditions of
Section~\ref{sec:construction-moments}, together with the protection
condition \eqref{eq:construction-protecting-cutoff}. The processing
order must also exhaust every bounded energy band: for each $W\ge0$
there must be a finite step $N(W)$ such that
\begin{equation}
  f_{N(W),J}\ge W\qquad\text{for every }|J|\le W.
  \label{eq:construction-exhaustion}
\end{equation}
Finally, the absolute spectral changes from the complete repairs must
be summable with the Boltzmann weight $e^{-\beta\Ehat}$ for every
$\beta>0$. Section~\ref{sec:global-limit} gives the precise thermal
norm and shows how exhaustion, protection and this summability yield
a locally finite modular limit. These requirements constrain the cell
sizes, moment degrees, repair cutoffs and processing order jointly.

\subsection{An explicit admissible schedule}
\label{sec:construction-schedule}

We now specify one schedule for which the local conditions, protection,
finite-band exhaustion and thermal summability can all be established.
We use two sets of cell-selection rules, according to the step
index $n$. The initial stage processes one cell in each spin row
$|J|\le T$, so it consists of $N_{\mathrm{init}}=2\lfloor T\rfloor+1$
steps, $0\le n<N_{\mathrm{init}}$. The tail stage,
$n\ge N_{\mathrm{init}}$, uses cells of bounded energy width.

Here $T$ sets the initial spin range, while $U$ places the initial
right endpoints in $[U,U+1]$. Their values are specified below.
Fix also a sufficiently large constant $A_{\mathrm{tail}}>0$.
Table~\ref{tab:construction-cell-parameters} gives both sets of
rules; $[V_n^-,V_n^+]$ is the interval in which
\eqref{eq:construction-integer-endpoint} selects the right endpoint.
\begin{table}[htbp]
  \centering
  \small
  \renewcommand{\arraystretch}{1.18}
  \begin{tabularx}{\textwidth}{@{}lXX@{}}
    \toprule
    Quantity & Initial stage, $0\le n<N_{\mathrm{init}}$
             & Tail stage, $n\ge N_{\mathrm{init}}$\\
    \midrule
    Spin $j_n$ & Next row with $|j_n|\le T$
               & Row with least $f_{n,J}$\\
    Left endpoint $L_n$ & $\max\{B,|j_n|\}$ & $f_{n,j_n}$\\
    $[V_n^-,V_n^+]$ & $[U,U+1]$ & $[L_n+\tfrac12,L_n+1]$\\
    Degree $k_n$ & $\lfloor\sqrt{aU}\rfloor$
                 & $\lceil A_{\mathrm{tail}}(a+L_n)\rceil$\\
    Cutoff $B_n$ & $2U+4$ & $2+\max\{U+2,L_n\}$\\
    \bottomrule
  \end{tabularx}
  \caption{One choice of cell parameters for the two stages.
  The scales $T$ and $U$ are chosen by
  \eqref{eq:construction-initial-scales}.}
  \label{tab:construction-cell-parameters}
\end{table}

\paragraph{Initial scales.}
We set
\begin{equation}
  (T,U)=
  \begin{cases}
    (R_0b,R_{\mathrm{init}}b),&B=b=\kappa a,\\
    (5B^\sharp,1024B^\sharp),&B\text{ fixed},
  \end{cases}
  \qquad B^\sharp=\max\{2,B\}.
  \label{eq:construction-initial-scales}
\end{equation}
In the growing-cutoff case, we choose fixed constants $R_0>5$
and $R_{\mathrm{init}}\ge16R_0$ sufficiently large to satisfy the
admissibility conditions in Section~5 and Appendix~D. In both
regimes, the positive contribution from the high-energy part of each
initial cell ensures the variation inequality in
\eqref{eq:construction-moment-reserve}, even if the density is negative
at lower energies.

In the tail stage, the cell widths lie between $1/2$ and $1$.
Keeping them bounded helps control the polynomial approximation
errors in the repair estimates. We choose the moment degrees $k_n$
and repair cutoffs $B_n$ together with these widths to control the
complete exterior corrections. Section~5 and Appendix~D establish
these estimates and show that the absolute spectral contributions
of successive updates are summable with the Boltzmann weight
$e^{-\beta\Ehat}$ for every fixed $\beta>0$.

\paragraph{Processing order and protection.}
We process the initial rows $|J|\le T$ once each, in the order
$0,-1,+1,-2,+2,\ldots$. The remaining rows retain their initial
endpoints $f_{0,J}=|J|$.

Thereafter choose a row with the smallest processed endpoint
$f_{n,J}$. If several rows share this minimum, choose the first
among them in the order $0,-1,+1,-2,+2,\ldots$.
Since $f_{n,J}\ge|J|$, only finitely many rows have a processed
endpoint below any fixed energy. The minimum is therefore attained.
After processing the selected cell, we advance its row endpoint
according to
\eqref{eq:construction-front-update}.

For a tail cell, the endpoint window gives
$\tfrac12\le V_n-L_n\le1$, so each step advances its selected
endpoint by at least $1/2$. Only finitely many rows have endpoints
in a bounded energy band, and each can be selected only finitely
many times before its endpoint leaves that band. The
smallest-endpoint rule therefore ensures that every bounded band
is eventually processed.

The selected $L_n$ is nondecreasing in the tail stage, so every earlier
tail endpoint is at most $L_n+1$. Together with the initial endpoint
bound $U+1$, this proves that the repair cutoffs in
Table~\ref{tab:construction-cell-parameters} satisfy
\eqref{eq:construction-protecting-cutoff}. It also controls the size of
the protecting cutoff relative to the current cell energy, as needed
in the exterior estimates of Appendix~\ref{app:global-tail}.

\paragraph{Parameter choices.}
Section~5 and Appendix~D verify the requirements of
Section~\ref{sec:construction-update} for this prescription.
In the growing-cutoff case, the auxiliary constants
are fixed first, then $0<\kappa<\kappa_0$, and finally a real $a$
is chosen sufficiently large for the fixed first-level data, including
$\kappa a\ge\max_{J\in S}|J|$. For fixed $B>0$ and $M_*$, one
lower bound on $a$ works for all $\delta\in[0,B)$ and all admissible
first-level spins and multiplicities with total at most $M_*$.
The tail constant $A_{\mathrm{tail}}$ can be fixed independently
of these data, $B$, $\delta$, $\kappa$ and $a$.

Integer spins label the rows, while the endpoint rule and moment
matching produce integer weights. The degree and row cutoffs are
rounded explicitly. These operations impose no arithmetic restriction
on the real parameter $a$.

\paragraph{Admissible choices.}
The atomic replacement $Q_n$ may be chosen from any vector in
$\mathcal X$. The local and exterior estimates use only its support,
total weight and matched moments, and hold uniformly over these choices.

The ceiling rule can also be replaced by another fixed rule selecting
an endpoint with positive integer total weight within the same
allowed window. The moment degrees may vary within the ranges
justified in Appendix~\ref{app:global-tail},
possibly after increasing the lower bound on $a$.
The estimates there verify the requirements of
Section~\ref{sec:construction-update} for both kinds of variation.

The auxiliary parameters and selection rules therefore form
part of the definition of a particular spectrum. Changing them
can change individual energies and multiplicities while
preserving the genus-one properties and the prescribed gap.
Constants used only to prove estimates do not enter the selection
of cells or atoms.

\paragraph{Construction summary.}
For fixed admissible input parameters, the explicit prescription is:
\begin{enumerate}[label=\arabic*.,leftmargin=1.6em]
  \item Form $F_0$ in the desired gap regime, extract
    $\mu_0^{\mathrm{at}}$ and $\rho_0$, and initialize
    $f_{0,J}=\max\{B,|J|\}$ and $n=0$.
  \item Select $j_n$, $L_n$, $[V_n^-,V_n^+]$, $k_n$ and $B_n$
    using Table~\ref{tab:construction-cell-parameters} and the
    prescribed row order.
  \item Use \eqref{eq:construction-integer-endpoint} to select
    $M_n$ and $V_n$. Form the normalized moments of the entire
    cell and select the lexicographically least vector in
    $\mathcal X$ of \eqref{eq:construction-node-set}, as described in
    Section~\ref{sec:construction-moments}. Construct $Q_n$ from this
    vector using \eqref{eq:construction-cell-atoms}.
  \item Set $\nu_n=Q_n-\mu_n^{\mathrm{cell}}$, apply
    \eqref{eq:construction-full-update}, and advance $f_{n,j_n}$
    to $V_n$. Repeat with $n+1$.
\end{enumerate}
These rules determine the sequence of modular functions and its
accumulating atomic spectrum. Section~5 proves that the sequence
has a limit in the required genus-one class.

\section{Global consistency}
\label{sec:global-consistency}

We now verify the requirements of Section~\ref{sec:construction-update}
for the explicit prescription in Section~\ref{sec:construction-schedule}
and show that its limit belongs to $\mathcal U_c$. The argument
has three parts. We first control the evolving continuum so that every
cell admits the required integer replacement. We then prove convergence
of the complete modular functions and identify their limiting spectrum.
Finally, we read off the two gaps in Theorem~\ref{thm:analytic}.
Throughout the recursion, the real parameter $a$ and the initial cutoff
$B$ are fixed; the large-$a$ conditions ensure that the estimates hold
at every step.
Appendix~\ref{app:integer-moments} proves the polynomial and
integer-matching inputs;
Appendix~\ref{app:global-estimates} supplies the global estimates.

\subsection{Uniform control of the recursion}
\label{sec:global-control}

For the explicit schedule, we establish two estimates at each step: the current cell
must admit the integer replacement of
Proposition~\ref{prop:construction-integer-matching}, and the complete
modular repair must be small enough to preserve this property for later
cells. We first collect the estimates used below in a single lemma,
then use them to prove that every step is defined.

\paragraph{Reference density and the inductive bound.}
The four seed terms in the vacuum character give a positive reference
for the unprocessed continuum. Their contribution from the $m=1$
term of the cosine sum in \eqref{eq:repair-Q-kernel} is
\begin{equation}
  \begin{aligned}
    \rho^{(1)}_{a,J}(\Ehat)
      &=2D_a(\Ehat+J)D_a(\Ehat-J),\\
    D_a(z)&=\cosh(2\pi\sqrt{az})
          -\cosh(2\pi\sqrt{(a-2)z}),\qquad z\ge0.
  \end{aligned}
  \label{eq:global-leading-density}
\end{equation}
This density is taken with respect to $\dd\omega_J$.
It is positive for $\Ehat>|J|$ and vanishes at the spin threshold.
Set $H_a(\Ehat)=e^{7\sqrt{a\Ehat}}$. We will prove by induction
that the prescription defines $F_n$ for every $n\ge0$, with its
unprocessed continuous density $\rho_{n,J}(\Ehat)$ satisfying
\begin{equation}
  |\rho_{n,J}(\Ehat)-\rho^{(1)}_{a,J}(\Ehat)|
       \le H_a(\Ehat),
  \qquad \Ehat>f_{n,J},\quad \Ehat\ge T.
  \label{eq:global-invariant}
\end{equation}
Inequalities between densities are understood almost everywhere with
respect to $\dd\omega_J$. The restriction to unprocessed energies is
essential: the continuum is already zero below $f_{n,J}$.
The reference grows faster than $H_a$ away from the spin threshold,
but its zero at that threshold prevents a pointwise positivity argument
throughout every cell. Instead, \eqref{eq:global-invariant} will give
the positive polynomial moments needed for the replacement.

For each completed step, let $\varepsilon_n\ge0$ bound its exterior
error relative to $H_a$:
\begin{equation}
  |r_{B_n,J}[\nu_n](\Ehat)|
       \le\varepsilon_n H_a(\Ehat),
  \qquad \Ehat>B_n,\quad\Ehat>|J|.
  \label{eq:global-step-error}
\end{equation}
This includes the entire repair in
\eqref{eq:complete-repair-exterior}, including the direct anchor density.
The general exterior bound allows a finite choice of $\varepsilon_n$;
exact moment matching gives the small bounds used in the induction.

\paragraph{Estimates used in the induction.}
For the prescription in Table~\ref{tab:construction-cell-parameters},
we establish the density bound \eqref{eq:global-invariant} by
induction using the following lemma. The bounds in the lemma are
uniform over all choices of atomic positions in the set
$\mathcal X$ defined in \eqref{eq:construction-node-set}.
Its proof, including the
separate fixed-cutoff and growing-cutoff estimates, is given in
Appendix~\ref{app:global-estimates}.

\begin{lemma}[Initial and tail estimates]
\label{lem:global-step-estimates}
Fix the analytic estimate constants, and then choose
$R_0$, $R_{\mathrm{init}}$ and $A_{\mathrm{tail}}$ sufficiently large,
independently of $a$ and $\kappa$. There is a constant $\kappa_0>0$
with the following properties. For fixed $B>0$, the conclusions below
hold for every sufficiently large real $a$, with one threshold uniform in
$\delta\in[0,B)$ and all admissible first-level data of total multiplicity
at most a fixed $M_*$. For $B=b=\kappa a$, they hold for each fixed
$0<\kappa<\kappa_0$ and fixed finite set of first-level spins, for every
sufficiently large real $a$. The bound on $a$ may depend on $M_*$ and
the spin set, but $R_0$, $R_{\mathrm{init}}$, $A_{\mathrm{tail}}$
and the sufficient $\kappa_0$ can be kept independent of them.
\begin{enumerate}[label=(\roman*),leftmargin=*]
  \item \textbf{Initial stage.}
  The initialized density satisfies
  \begin{equation}
    |\rho_{0,J}(\Ehat)-\rho^{(1)}_{a,J}(\Ehat)|
         \le\tfrac12 H_a(\Ehat),
    \qquad \Ehat\ge T,\quad \Ehat>|J|.
    \label{eq:global-initial-reference}
  \end{equation}
  In each initial spin row, the continuous measure of $F_0$ satisfies
  the endpoint conditions of Section~\ref{sec:construction-moments}
  on $[U,U+1]$. Its cell functional satisfies the variation inequality
  in \eqref{eq:construction-moment-reserve} for every candidate endpoint
  $V\in[U,U+1]$. After selecting the integer endpoint and a matching
  replacement, form the residual from this cell of $F_0$.
  For the complete repairs of these initial residuals, the errors can be
  chosen to satisfy
  \begin{equation}
    \sum_{n<N_{\mathrm{init}}}\varepsilon_n\le\tfrac18.
    \label{eq:global-initial-budget}
  \end{equation}

  \item \textbf{Tail step.}
  Suppose a finite prefix has been constructed through $F_n$, with
  $n\ge N_{\mathrm{init}}$, and satisfies
  \eqref{eq:global-invariant}. The next selected row satisfies the
  endpoint conditions on $[L_n+1/2,L_n+1]$, and its cell functional
  satisfies the variation inequality for every candidate endpoint in
  this interval. After the integer replacement, the complete exterior error can be
  bounded with
  \begin{equation}
    \varepsilon_n\le\frac1{64}(2+L_n)^{-4}.
    \label{eq:global-tail-error}
  \end{equation}
  The same lower threshold on $a$ makes these conclusions valid for
  every tail cell with $L_n\ge T$.
\end{enumerate}
\end{lemma}

The variation inequality controls the signed cell measure. The other
condition in \eqref{eq:construction-moment-reserve}, positive integer
total weight, follows after the endpoint rule is applied. The lemma
therefore supplies precisely the hypotheses of
Proposition~\ref{prop:construction-integer-matching}. In the initial
stage, the large positive contribution at the upper end of each cell
overcomes any negative contribution at lower energies. In the tail,
the inductive bound supplies this variation inequality, while moment
cancellation controls the effect of the repair on later cells.

\paragraph{Starting the induction.}
Every initial endpoint is at most $U+1$, whereas all initial repairs
use the cutoff $2U+4$. Their exterior outputs therefore lie above
every initial cell, and a direct replacement changes only its own
cell. Thus, when an initial row is processed, its cell still has the
continuous measure of $F_0$. Part~(i) of
Lemma~\ref{lem:global-step-estimates} and
Proposition~\ref{prop:construction-integer-matching} consequently
establish all initial replacements and their common error bound.

At an unprocessed energy, no earlier direct replacement contributes.
Only the exterior outputs can change $\rho_0$ there, and an exterior
output is zero through its cutoff. Hence, for every
$0\le n\le N_{\mathrm{init}}$ and
$\Ehat>f_{n,J}$ with $\Ehat\ge T$,
\[
  |\rho_{n,J}(\Ehat)-\rho^{(1)}_{a,J}(\Ehat)|
    \le\left(\tfrac12+\sum_{m<n}\varepsilon_m\right)H_a(\Ehat)
    \le\tfrac58H_a(\Ehat)<H_a(\Ehat).
\]
This proves \eqref{eq:global-invariant} throughout the initial stage,
including the starting state for the tail recursion.

\paragraph{Continuing the induction.}
Suppose $F_n$ is defined and obeys \eqref{eq:global-invariant}, with
$n\ge N_{\mathrm{init}}$. Part~(ii) of the lemma gives the endpoint
and moment conditions for the next cell. Its integer replacement and
complete repair therefore define $F_{n+1}$.

The lower bound on tail-cell widths and the constraint
$L_m\ge|j_m|$ control the accumulated errors of any such finite
prefix. For each integer $q\ge0$, at most $2(2q+1)$ tail
cells have $L_m\in[q,q+1)$: only spins $|j_m|\le q$ can occur,
and each row advances by at least $1/2$ on each visit.
Consequently,
\begin{equation}
  \sum_{m=N_{\mathrm{init}}}^{n}\varepsilon_m
    \le\frac1{64}\sum_{q\ge0}\frac{2(2q+1)}{(q+2)^4}
    \le\tfrac1{32}.
  \label{eq:global-tail-budget}
\end{equation}
At every energy still unprocessed in $F_{n+1}$, we may again sum only
the exterior outputs. Equations~\eqref{eq:global-initial-reference},
\eqref{eq:global-initial-budget} and \eqref{eq:global-tail-budget}
then give
\[
  |\rho_{n+1,J}-\rho^{(1)}_{a,J}|
    \le\left(\tfrac12+\tfrac18+\tfrac1{32}\right)H_a<H_a.
\]
The bound is thus preserved at the next step. This closes the
induction and proves that the endpoint rule and integer replacement
are defined throughout the infinite recursion.

\subsection{Convergence and genus-one consistency}
\label{sec:global-limit}

For any recursion satisfying the requirements of
Section~\ref{sec:construction-update}, every bounded energy band
is completely processed after finitely many steps, and subsequent
repairs leave its spectrum unchanged. We then verify thermal
summability for the explicit schedule and show how this condition
allows modular invariance to pass to the limit.

\paragraph{Finite stabilization of the spectrum.}
Fix $W\ge0$, and choose a finite step $N(W)$ as in
\eqref{eq:construction-exhaustion}. By that step, every row with
$|J|\le W$ has been processed through energy $W$; rows with
$|J|>W$ have no support in this band. The protection condition
\eqref{eq:construction-protecting-cutoff} and the exact repair
identities \eqref{eq:exact-repair-identities} then give
\begin{equation}
  \mu_n\big|_{[0,W]}=
  \left.\left(\mu_0^{\mathrm{at}}+
               \sum_{m=0}^{N(W)-1}Q_m\right)\right|_{[0,W]},
  \qquad n\ge N(W).
  \label{eq:global-finite-stabilization}
\end{equation}
All restrictions here and below include every spin row. The continuum
has been removed from this band, and the remaining atoms are unchanged
by all later steps.

These compatible restrictions define the limiting spectral measure
\begin{equation}
  \mu_\infty=\mu_0^{\mathrm{at}}+\sum_{n=0}^{\infty}Q_n.
  \label{eq:global-atomic-limit}
\end{equation}
Only finitely many cells contribute to any bounded band, and each
contributes finitely many unit atoms. Thus $\mu_\infty$ is locally
finite, counting multiplicity. Coincident atoms combine into positive
integer weights, and every atom has integer spin and $\Ehat\ge|J|$.
Finite stabilization also rules out accumulation at a finite energy.

\paragraph{Absolute thermal convergence.}
For a signed nonvacuum measure $\sigma$, define
\begin{equation}
  \|\sigma\|_\beta=
  \sum_{J\in\Z}\int_{[|J|,\infty)}e^{-\beta\Ehat}
                        \dd|\sigma_J|(\Ehat),
  \qquad\beta>0.
  \label{eq:global-thermal-norm}
\end{equation}
This norm includes a possible scalar threshold atom and sums the
magnitudes of the spectral contributions with their Boltzmann weights.
In this norm, positive and negative spectral contributions are counted
without cancellation. Throughout this argument, $a$, $B$ and $\beta$
are fixed.

We verify the required summability for the explicit prescription.
For $n\ge N_{\mathrm{init}}$, its remaining continuum lies above $T$.
The decomposition \eqref{eq:construction-spectral-state} and the
density bound \eqref{eq:global-invariant} give
\[
  \begin{aligned}
    \|\mu_n\|_\beta
      &=\|\mu_n^{\mathrm{at}}\|_\beta+
        \sum_{J\in\Z}\int_{|J|}^\infty e^{-\beta\Ehat}
          |\rho_{n,J}(\Ehat)|\dd\omega_J(\Ehat)\\
      &\le\|\mu_n^{\mathrm{at}}\|_\beta+
        \sum_{J\in\Z}\int_{\max\{T,|J|\}}^\infty e^{-\beta\Ehat}
          \bigl(\rho^{(1)}_{a,J}(\Ehat)+H_a(\Ehat)\bigr)
          \dd\omega_J(\Ehat).
  \end{aligned}
\]
The atomic contribution is a finite sum at every finite stage.
The sum of integrals in the upper bound is independent of $n$.
It is the thermal norm of the positive comparison measure with spin
components $(\rho^{(1)}_{a,J}+H_a)\dd\omega_J$, restricted to
$\Ehat\ge T$. This norm is finite, including the sum over spins.
The density $\rho^{(1)}_{a,J}+H_a$ with respect to $\dd\omega_J$
grows at most exponentially in $\sqrt{\Ehat}$, and this growth is
suppressed by any fixed Boltzmann factor. The scalar
integral starts at $T>0$, while the square-root singularities at
nonzero spin thresholds are integrable. The remaining spin sum
converges exponentially.
Appendix~\ref{app:global-thermal-comparison} gives the explicit
integral and spin-sum bounds.

The finitely many initial repairs have finite thermal norms by
\eqref{eq:repair-thermal-convergence}. We now estimate the sum over
the tail steps.

For a tail cell $I_n=(L_n,V_n]$, the positive replacement $Q_n$
has total weight
$M_n=\mu_n^{\mathrm{cell}}(I_n)
\le|\mu_n^{\mathrm{cell}}|(I_n)$.
Its atoms all lie in $I_n$. Since $V_n-L_n\le1$, the Boltzmann
factor varies by at most a factor $e^\beta$ across the cell, giving
\[
  \|Q_n\|_\beta
    \le e^{-\beta L_n}M_n
    \le e^\beta\|\mu_n^{\mathrm{cell}}\|_\beta.
\]
The triangle inequality for $\nu_n=Q_n-\mu_n^{\mathrm{cell}}$
adds the thermal norm of the removed continuum, producing the factor
$1+e^\beta$. Using the density bound \eqref{eq:global-invariant}
on $I_n$, we obtain
\begin{equation}
  \|\nu_n\|_\beta\le(1+e^\beta)
       \int_{I_n}e^{-\beta\Ehat}
       \bigl(\rho^{(1)}_{a,j_n}(\Ehat)+H_a(\Ehat)\bigr)
       \dd\omega_{j_n}(\Ehat).
  \label{eq:global-direct-thermal-bound}
\end{equation}
The tail cells lie above $T$ and are disjoint within each spin row.
Summing the integrals in \eqref{eq:global-direct-thermal-bound} over
all tail steps therefore gives at most the finite thermal norm of
the comparison measure, including all spins. The common factor
$1+e^\beta$ is fixed, so the direct residuals are summable in
thermal norm.

The complete spectral change includes both the direct residual and
the exterior repair:
\begin{equation}
  (\mu_{n+1}-\mu_n)_J=(\nu_n)_J+
     \mathbf1_{(B_n,\infty)}\,
       r_{B_n,J}[\nu_n]\dd\omega_J.
  \label{eq:global-complete-step}
\end{equation}
For a tail step, the exterior term is bounded in absolute value by
$\varepsilon_nH_a\dd\omega_J$ above $B_n>T$.
Its thermal norm is therefore at most $\varepsilon_n$ times the
finite thermal norm of the comparison measure. The budget
\eqref{eq:global-tail-budget} then proves summability of these
exterior terms in thermal norm.
Together with the direct estimate
\eqref{eq:global-direct-thermal-bound} and the finite initial
contribution, this proves
\begin{equation}
  \sum_{n=0}^{\infty}\|\mu_{n+1}-\mu_n\|_\beta<\infty
  \qquad\text{for every }\beta>0.
  \label{eq:global-complete-thermal-sum}
\end{equation}
The vacuum estimates and the bounds on the finite initialization
repairs and anchors also give $\|\mu_0\|_\beta<\infty$.
Equation~\eqref{eq:global-complete-thermal-sum} therefore implies
that $\mu_n$ converges in thermal norm for every $\beta>0$.

To identify the limit, fix a bounded energy band $[0,W]$.
On this band, total variation is bounded by $e^{\beta W}$ times
the thermal norm, so the convergence also holds in total variation.
By \eqref{eq:global-finite-stabilization},
the restrictions of $\mu_n$ to this band eventually equal those
of the atomic measure $\mu_\infty$ defined in
\eqref{eq:global-atomic-limit}. Since this holds for every $W$,
the thermal limit is $\mu_\infty$ on the full spectrum.
In particular, $\|\mu_\infty\|_\beta<\infty$ for every $\beta>0$.

\paragraph{The modular limit.}
For any recursion with the exhaustion, protection and thermal
summability just described, return to the modular functions and define
\begin{equation}
  F_\infty=F_0+\sum_{n=0}^{\infty}\Repair_{B_n}[\nu_n],
  \qquad
  Z_\infty(\tau)=\frac{F_\infty(\tau)}
                         {\sqrt y\,|\eta(\tau)|^2}.
  \label{eq:global-modular-limit}
\end{equation}
Each summand is a complete modular repair with the spectral change
\eqref{eq:global-complete-step}. These repairs preserve the vacuum
and the scalar threshold coefficient. For a compact set
$\mathcal K\subset\HH$, let $0<y_-\le y_+<\infty$ be the
minimum and maximum of $\operatorname{Im}\tau$ on $\mathcal K$.
The spectral transform \eqref{eq:construction-state-transform}
gives
\[
  \sup_{\tau\in\mathcal K}
       |\Repair_{B_n}[\nu_n](\tau)|
    \le\sqrt{y_+}\,
          \|\mu_{n+1}-\mu_n\|_{2\pi y_-}.
\]
Equation~\eqref{eq:global-complete-thermal-sum} therefore gives
absolute and uniform convergence of the repair series on every
compact subset of $\HH$. The image of such a set under either
modular generator is again compact, so the finite-stage modular
identities pass to $F_\infty$. The invariant prefactor in
\eqref{eq:reduced-function} gives modular invariance of $Z_\infty$.

The limiting nonvacuum spectrum lies entirely in the uncensored
region. The vacuum remains unique, with multiplicity one.
Local finiteness and positive integer multiplicities were established
in \eqref{eq:global-atomic-limit}. Using $\Delta=a+\Ehat$, the
thermal bound also gives \eqref{eq:thermal-convergence}. Hence
$Z_\infty\in\mathcal U_c$.

\subsection{Completion of the gap theorem}
\label{sec:global-gaps}

It remains to identify the first primary in the limit.
Every $Q_n$ is supported at $\Ehat>L_n\ge B$, so
\eqref{eq:global-atomic-limit} gives
\begin{equation}
  \mu_\infty\big|_{[0,B]}=\mu_0^{\mathrm{at}}
      =\nu_{\mathrm{first}}.
  \label{eq:global-prescribed-band}
\end{equation}
Thus the first energy has exactly the prescribed spin-$J$ multiplicity
$d_J$ for $J\in S$, with no other primary at that energy. Its total
multiplicity is $d_{\mathrm{first}}$.

For part~\ref{item:positive-fraction}, fix $0<\kappa<\kappa_0$
as in Lemma~\ref{lem:global-step-estimates} and take
$B=b=\kappa a$. For every sufficiently large real $a$, the
prescribed first level remains at $b$, giving
$\Ehat_1(Z_{a,\kappa})=b$. The dimension gap
\eqref{eq:gap-bounds} and its limit \eqref{eq:asymptotic-gap}
follow immediately.

For part~\ref{item:finite-gap}, fix $\delta\ge0$ and choose a
fixed $B>\delta$. Lemma~\ref{lem:global-step-estimates} supplies
one threshold $a_0(B,M_*)$ valid for all $\delta\in[0,B)$ and all
admissible first-level data of total multiplicity at most $M_*$, including
the scalar threshold case. Taking $a_\delta=a_0(B,M_*)$,
\eqref{eq:global-prescribed-band} gives
$\Ehat_1(\widetilde Z_{a,\delta})=\delta$, with every other nonvacuum
primary strictly above $B$. The limits
\eqref{eq:finite-shifted-gap} follow from fixed $\delta$ and
$a\to\infty$. Together with genus-one consistency from
Section~\ref{sec:global-limit}, this completes the proof of
Theorem~\ref{thm:analytic}.

\section{Spectral entropy and BTZ state counting}
\label{sec:entropy}

We now study the smoothed density of states of the discrete partition
functions constructed above. For fixed-$B$ families using the schedule
in Table~\ref{tab:construction-cell-parameters}, the spectral entropy agrees
with the prediction of a single perturbative BTZ saddle through every
finite order in $1/c$. Exact moment matching makes this
comparison possible: it controls the smoothed density even before the
individual atom positions have been evaluated. A gap growing with $c$
leads to a different energy dependence. For the explicit schedule and
lexicographic atom choice, Section~\ref{sec:entropy-positive-kappa}
establishes an interval without BTZ matching above this gap, as well as
matching to every finite order at sufficiently high energies.

Throughout Sections~\ref{sec:entropy-smooth}--\ref{sec:entropy-matching},
$B>0$ and the first-level multiplicity bound $M_*$ are fixed as the real
parameter $a\to\infty$, with $c=12a+1$. Unless a more restricted
statement is specified, the prescribed shifted gap may be any
$\delta\in[0,B)$, with any admissible first-level spins and multiplicities
of total at most $M_*$. We retain
$a=(c-1)/12$ and $\Ehat=E+1/12$ exactly until taking the large-$c$
limit. The uniform estimates and coefficient algorithm used below are
proved in Appendices~\ref{app:entropy-transfer} and
\ref{app:entropy-reference-expansion}.

\subsection{Smoothed density of states and the BTZ prediction}
\label{sec:entropy-smooth}

\paragraph{The smoothed density of states.}
Let $s$ label distinct pairs of cylinder energy $E_s$ and total spin
$J_s$, with multiplicity $d_s$. Choose a nonnegative smooth kernel
with compact support, $\phi\in C_c^\infty(\R)$, whose shape and width
are independent of $c$, and normalize it by
$\int_{\R}\phi(u)\dd u=1$. Define
\begin{equation}
  N_\phi(E)=\sum_s d_s\,\phi(E_s-E),
  \qquad
  S_\phi(E)=\log N_\phi(E).
  \label{eq:entropy-smooth-observable}
\end{equation}
The sum includes every spin, all Virasoro descendants and the vacuum
module. With this normalization, $N_\phi$ is a smoothed density of
states, and $S_\phi$ is its spectral entropy. The logarithm is taken
where $N_\phi>0$, as guaranteed in the large-$c$ regime below.
The entropy is defined using the natural logarithm, and we use this
convention throughout the paper.

Equivalently, using the limiting nonvacuum primary measure
$\mu_\infty$ in \eqref{eq:global-atomic-limit},
\begin{equation}
  \begin{aligned}
    N_\phi(E)={}&N_\phi^{\mathrm{vac}}(E)\\
      &+\sum_{N=0}^{\infty}p_2(N)\sum_{J\in\mathbb Z}
        \int_{[|J|,\infty)}
        \phi\!\left(\Ehat-\frac1{12}+N-E\right)
        \dd\mu_{\infty,J}(\Ehat).
  \end{aligned}
  \label{eq:entropy-primary-count}
\end{equation}
Here $N_\phi^{\mathrm{vac}}(E)$ is the vacuum-module contribution
with the same kernel, $J$ is the primary spin, and $N$ is the total
descendant level. The number of pairs of left- and right-moving
descendants at level $N$ is
\begin{equation}
  p_2(N)=\sum_{\ell=0}^{N}p(\ell)\,p(N-\ell),
  \label{eq:entropy-descendant-multiplicity}
\end{equation}
where $p(\ell)$ is the integer partition number, with $p(0)=1$.
The kernel tests their cylinder energy $\Ehat-1/12+N$.

The energy kernel enters the bulk comparison through
\begin{equation}
  K_\phi(\beta)=\int_{\R}\phi(u)e^{\beta u}\dd u.
  \label{eq:entropy-kernel-transform}
\end{equation}
Normalization fixes $K_\phi(0)=1$; the value at the relevant inverse
temperature still depends on the kernel. Rescaling an unnormalized
kernel by a positive constant rescales $N_\phi$ by the same constant
and adds its logarithm to $S_\phi$.

For example, for any fixed $w>0$, a normalized smooth bump is
\begin{equation}
  \phi_w(u)=\frac{1}{Z_w}
  \begin{cases}
    \displaystyle\exp\!\left[-\frac{1}{1-(u/w)^2}\right],
      & |u|<w,\\[6pt]
    0, & |u|\ge w,
  \end{cases}
  \label{eq:entropy-bump-kernel}
\end{equation}
where $Z_w$ normalizes the integral to one. The function and all its
derivatives vanish at $u=\pm w$, so it joins smoothly to zero outside
$[-w,w]$.

A Gaussian kernel of fixed width can also be used. Although it has
unbounded support, the matching result below extends to this case using
the additional tail estimates in Appendix~\ref{app:entropy-transfer}.

\paragraph{The perturbative BTZ prediction.}
For pure Einstein AdS$_3$ gravity with Brown--Henneaux boundary
conditions~\cite{BrownHenneaux1986}, the contribution of a single
Euclidean BTZ saddle at
zero angular potential is
\begin{equation}
  Z_{\mathrm{BTZ}}^{\mathrm{pert}}(\beta)
    =\exp\left(\frac{\pi^2c}{3\beta}\right)
       V\left(\frac{4\pi^2}{\beta}\right),
  \qquad
  V(u)=\prod_{m=2}^{\infty}(1-e^{-mu})^{-2}.
  \label{eq:entropy-btz-partition}
\end{equation}
Here $c$ is the central charge appearing in the Virasoro vacuum character,
the same $c$ used in our partition functions.
The product is the boundary-graviton determinant; it starts at $m=2$
because of the vacuum null states. This is the modular transform of
the vacuum character, with the one-loop exact structure of the
single-saddle torus contribution~\cite{GiombiMaloneyYin2008,CotlerJensen2019}.

To compare densities at the same energy resolution, apply the same kernel
before performing the inverse Laplace integral:
\begin{equation}
  N_\phi^{\mathrm{BTZ}}(E)
    =\int_{\gamma-i\infty}^{\gamma+i\infty}
       \frac{\dd\beta}{2\pi i}\,
       e^{\beta E}K_\phi(\beta)
       Z_{\mathrm{BTZ}}^{\mathrm{pert}}(\beta),
  \qquad \gamma>0.
  \label{eq:entropy-btz-observable}
\end{equation}
Indeed, integrating a density at $E+u$ against $\phi(u)$ multiplies
its inverse Laplace integrand by $K_\phi(\beta)$. Thus this factor
records the chosen energy resolution.
Zero angular potential corresponds to summing over all spins.
A count at fixed total spin requires an additional projection.

\paragraph{The large-$c$ regime.}
We evaluate the smoothed density of states at $E=\lambda c$, with
$\lambda$ in a compact subset of $(0,\infty)$, and define
\begin{equation}
  r=\sqrt{12\lambda},
  \qquad
  \beta_*=\frac{2\pi}{r},
  \qquad
  \widetilde\beta_*=\frac{4\pi^2}{\beta_*}=2\pi r.
  \label{eq:entropy-saddle-parameters}
\end{equation}
The leading exponent in \eqref{eq:entropy-btz-observable} is
$c(\lambda\beta+\pi^2/(3\beta))$, whose positive saddle is $\beta_*$.
Its value gives the nonrotating BTZ entropy
\begin{equation}
  S_{\mathrm{BTZ}}(E)=2\pi\sqrt{\frac{cE}{3}}
                    =\frac{\pi cr}{3}.
  \label{eq:entropy-btz-action}
\end{equation}
The relation to state growth is the usual BTZ entropy
comparison~\cite{Strominger1998}. Gaussian integration around $\beta_*$
gives a prefactor proportional to $c^{-1/2}$, while higher terms in the
saddle-point expansion give corrections in powers of $1/c$. These arise
from the inverse Laplace integral~\eqref{eq:entropy-btz-observable}, even
though the single-saddle partition function is one-loop exact.

Our comparison includes $0<E/c<1/12$, the intermediate range whose
entropy is not determined by the general HKS assumptions
\cite[Sections~2.4 and 4.2]{HartmanKellerStoica2014}.
In this range, $\beta_*>2\pi$, and the BTZ saddle is subdominant to
thermal AdS in the canonical ensemble at $\beta_*$. Nevertheless, its
smoothed density of states agrees with that of our constructed partition
functions through every finite order in $1/c$. The comparison remains
at genus one.

\subsection{Matching to every order in \texorpdfstring{$1/c$}{1/c}}
\label{sec:entropy-matching}

We prove Proposition~\ref{prop:btz-entropy} by first comparing the
discrete spectrum with an integer-spin reference. We then compare
that reference with the BTZ prediction and evaluate the resulting
saddle expansion to obtain the entropy coefficients.

\begin{lemma}[Comparison with the integer-spin reference]
\label{lem:entropy-reference-comparison}
Consider a fixed-$B$ family constructed using the schedule in
Table~\ref{tab:construction-cell-parameters}, with fixed admissible
auxiliary parameters. Let $\phi$ be a compactly supported
kernel as in Section~\ref{sec:entropy-smooth}.
Define $N_\phi^{\mathrm{ref}}(E)$ from the nonvacuum term in
\eqref{eq:entropy-primary-count} by replacing $\dd\mu_{\infty,J}$
with $\rho^{(1)}_{a,J}\,\dd\omega_J$, using the leading density
\eqref{eq:global-leading-density}:
\begin{equation}
  N_\phi^{\mathrm{ref}}(E)
    =\sum_{N=0}^{\infty}p_2(N)\sum_{J\in\mathbb Z}
      \int_{|J|}^{\infty}
      \phi\!\left(\Ehat-\frac1{12}+N-E\right)
      \rho^{(1)}_{a,J}(\Ehat)\,\dd\omega_J(\Ehat).
  \label{eq:entropy-reference-count}
\end{equation}
Then, for every fixed positive integer $P$,
\begin{equation}
  \frac{N_\phi(\lambda c)}{N_\phi^{\mathrm{ref}}(\lambda c)}
    =1+O_P(c^{-P}).
  \label{eq:entropy-reference-comparison}
\end{equation}
The estimate is uniform for $\lambda$ in any fixed closed interval
contained in $(0,\infty)$, over admissible first-level spins and multiplicities
with $\delta\in[0,B)$ and total multiplicity at most $M_*$, and over
all atomic choices in $\mathcal X$ at each step. The error constants
may depend on $B$, $M_*$, the fixed auxiliary parameters, $\phi$,
$P$ and the energy-ratio interval.
\end{lemma}

We use the density estimate \eqref{eq:global-invariant} established
in Section~\ref{sec:global-control}. The proof of
Lemma~\ref{lem:entropy-reference-comparison} has three
steps: moment matching on a cell, summation, and conversion to relative
error. Appendix~\ref{app:entropy-transfer} supplies the detailed
uniform estimates.

At fixed $c$ and $E$, compact support restricts the contributing
nonvacuum primaries to a bounded energy band.
Section~\ref{sec:global-limit} shows that this band reaches its final
discrete spectrum after finitely many steps. We apply the uniform
estimates at that stage before taking the large-$c$ limit.

\paragraph{Moment matching on a single cell.}
For a tail cell $I_n=(L_n,V_n]$ in spin row $j_n$, write the
normalized square-root coordinate of
\eqref{eq:construction-cell-functional} as
\begin{equation}
  z(\Ehat)=
  \frac{\sqrt{\Ehat-|j_n|}-\sqrt{L_n-|j_n|}}
       {\sqrt{V_n-|j_n|}-\sqrt{L_n-|j_n|}}.
  \label{eq:entropy-cell-coordinate}
\end{equation}
The cell corresponds to $0<z\le1$, and we use the closed interval
$[0,1]$ for polynomial approximation. Let $\Ehat_n(z)$ denote the
inverse map. For a fixed total descendant level $N\ge0$, the tested
function on this cell is $h_n(z)=\phi(\Ehat_n(z)-1/12+N-E)$.
The replacement $Q_n$ and the current cell measure
$\mu_n^{\mathrm{cell}}$ have the same moments through degree $k_n$
and the same mass $M_n$. Hence
\begin{equation}
  \left|\int_{I_n}h_n(z(\Ehat))\,\dd(Q_n-\mu_n^{\mathrm{cell}})(\Ehat)\right|
  \le
  \bigl(M_n+\|\mu_n^{\mathrm{cell}}\|_{\mathrm{TV}}\bigr)
  \inf_{p\in\Pi_{k_n}}\|h_n-p\|_{\infty,[0,1]}.
  \label{eq:entropy-smooth-moment-bound}
\end{equation}
For any $p\in\Pi_{k_n}$, the moment identities allow us to replace
$h_n$ by $h_n-p$ in the integral on the left. We bound the integrals
against $Q_n$ and $\mu_n^{\mathrm{cell}}$ separately by
$\|h_n-p\|_{\infty,[0,1]}$ times their respective total variations.
Since $Q_n$ is positive with total
weight $M_n$, its total variation is $M_n$. Taking the infimum over
$p$ gives \eqref{eq:entropy-smooth-moment-bound}.

In the extensive-energy cells and central spin region that dominate
the count, \eqref{eq:global-invariant} makes the current density
positive and exponentially close to $\rho^{(1)}_{a,J}$.
On these cells, the current measure is positive, so its total variation
equals its total weight $M_n$. The prefactor in
\eqref{eq:entropy-smooth-moment-bound} is therefore $2M_n$.

Table~\ref{tab:construction-cell-parameters} gives tail-cell widths
at most one and $k_n\ge A_{\mathrm{tail}}a$, so the moment degree
grows at least linearly with $c$.
The width bound controls the derivatives of $\Ehat_n(z)$.
For each fixed derivative order, the derivatives of $h_n$ on $[0,1]$
are bounded by a constant depending only on that order and the fixed
kernel $\phi$. Polynomial approximation gives
an error bounded by a constant times $k_n^{-P}$, and hence by a
constant times $c^{-P}$, for every fixed $P$. Combining this with
\eqref{eq:entropy-smooth-moment-bound} gives the bound $C_P M_n c^{-P}$.
This is an absolute error bound, proportional to the cell mass $M_n$.
Here $C_P$ is independent of $c$, the cell, its spin and the descendant
level; it may depend on $P$, $\phi$ and $A_{\mathrm{tail}}$.
Only the test function is
differentiated; no smoothness of the evolving error density is required.

\paragraph{Summing over cells, spins and descendants.}
For each fixed descendant level, compact support confines the tested
primary energies to an interval of fixed length. The lower width bound
$V_n-L_n\ge1/2$ from Table~\ref{tab:construction-cell-parameters}
then bounds the number of relevant cells in each spin row independently
of $c$ and the descendant level. In the retained extensive-energy
region, summing the relevant cell masses over central spins gives at
most a constant multiple of the primary reference count tested at the
same descendant level. The Gaussian spin profile of width of order
$\sqrt c$ controls both sums, so summing the cell errors introduces
no additional growing power of $c$.

At descendant level $N$, the primary energy tested by the kernel is
near $E+1/12-N$. While this energy remains in the retained
extensive-energy range, its decrease supplies a suppression factor
$e^{-\beta_0N}$ for some $c$-independent $\beta_0>0$. This factor
comes from primary state growth at the reduced energy and bounds the
sum over descendant multiplicities uniformly in $c$.

The invariant \eqref{eq:global-invariant} controls the error from
replacing the current cell measures by
$\rho^{(1)}_{a,J}\,\dd\omega_J$. After summation, this error is
exponentially small relative to $N_\phi^{\mathrm{ref}}(E)$. The same
holds for the remaining tail contributions, from the outer spin
regions and larger descendant levels. Combining these bounds with
the cell approximation estimates gives a total tail error at most
$C_Pc^{-P}N_\phi^{\mathrm{ref}}(E)$.

\paragraph{From absolute to relative error.}
The initial packets, the prescribed first level and their descendants,
together with the vacuum module, contribute at most
$\exp(O(\sqrt c))$ at these energies. Appendix~\ref{app:entropy-transfer}
derives the initial-weight estimate \eqref{eq:entropy-initial-mass}
and combines it with the descendant bound to obtain this contribution.
The contribution to the reference count from primary energies
$\Ehat\le U+1$, including their descendants, obeys the same bound.
Including these contributions gives the absolute error estimate
\begin{equation}
  \bigl|N_\phi(E)-N_\phi^{\mathrm{ref}}(E)\bigr|
  \le C_Pc^{-P}N_\phi^{\mathrm{ref}}(E)
       +\exp\!\bigl(O(\sqrt c)\bigr).
  \label{eq:entropy-reference-transfer-bound}
\end{equation}
The reference count grows exponentially in $c$: for $E=\lambda c$
in the stated range,
\begin{equation}
  \log N_\phi^{\mathrm{ref}}(E)
    =S_{\mathrm{BTZ}}(E)+O(\log c).
  \label{eq:entropy-reference-growth}
\end{equation}
The coefficient of $c$ in $S_{\mathrm{BTZ}}(\lambda c)$ is bounded
below by a positive constant on the chosen $\lambda$-interval.
Hence, for every fixed $P$, we have
\begin{equation}
  \exp\!\bigl(O(\sqrt c)\bigr)
    =o\!\left(c^{-P}N_\phi^{\mathrm{ref}}(E)\right).
  \label{eq:entropy-initial-relative-smallness}
\end{equation}
Dividing \eqref{eq:entropy-reference-transfer-bound} by the reference
count therefore proves Lemma~\ref{lem:entropy-reference-comparison}.

\paragraph{The exact reference transform.}
To obtain the entropy expansion in Proposition~\ref{prop:btz-entropy},
we next establish the BTZ comparison: for every fixed positive
integer $P$,
\begin{equation}
  \frac{N_\phi(\lambda c)}{N_\phi^{\mathrm{BTZ}}(\lambda c)}
    =1+O_P(c^{-P}).
  \label{eq:entropy-smooth-matching}
\end{equation}
The estimate applies to the same families, with the same uniformity
and parameter dependence, as Lemma~\ref{lem:entropy-reference-comparison}.
This comparison identifies the two asymptotic expansions through
every fixed finite order. It asserts neither convergence of the
series nor an exponentially small error.

By Lemma~\ref{lem:entropy-reference-comparison}, it suffices to compare
the integer-spin reference with the BTZ prediction. We begin with the
Laplace transform of the smoothed count. For real $\beta>0$, define
\begin{equation}
  Z_{\mathrm{ref}}^{\mathbb Z}(\beta)
    =\frac{1}{K_\phi(\beta)}
      \int_{\R}e^{-\beta E}N_\phi^{\mathrm{ref}}(E)\dd E.
  \label{eq:entropy-reference-laplace}
\end{equation}
The superscript records integer primary spins. The factor
$K_\phi(\beta)>0$ removes the effect of smoothing. Substituting
\eqref{eq:entropy-reference-count} and using
\[
  \int_{\R}e^{-\beta E}
    \phi\!\left(\Ehat-\frac1{12}+N-E\right)\dd E
    =K_\phi(\beta)e^{-\beta(\Ehat-1/12+N)}
\]
gives
\begin{equation}
  Z_{\mathrm{ref}}^{\mathbb Z}(\beta)
    =e^{\beta/12}G(\beta)
      \sum_{J\in\mathbb Z}\int_{|J|}^{\infty}
      e^{-\beta\Ehat}\rho^{(1)}_{a,J}(\Ehat)\,
      \dd\omega_J(\Ehat),
  \label{eq:entropy-integer-reference-transform}
\end{equation}
where the descendant multiplicities in
\eqref{eq:entropy-descendant-multiplicity} give
\begin{equation}
  G(\beta)=\sum_{N=0}^{\infty}p_2(N)e^{-\beta N}
    =\prod_{m=1}^{\infty}(1-e^{-m\beta})^{-2}.
  \label{eq:entropy-descendant-product}
\end{equation}

Extend $\rho^{(1)}_{a,J}$ in \eqref{eq:global-leading-density} to real
$J$ and replace the primary-spin sum in
\eqref{eq:entropy-integer-reference-transform} by an integral. The
resulting continuous-spin transform is
\begin{equation}
  \begin{aligned}
    Z_{\mathrm{ref}}^{\mathbb R}(\beta)
      &=e^{\beta/12}G(\beta)Z_{\mathrm{ref,prim}}(\beta),\\
    Z_{\mathrm{ref,prim}}(\beta)
      &=\int_0^\infty e^{-\beta\Ehat}
        \int_{-\Ehat}^{\Ehat}
        \frac{\rho^{(1)}_{a,J}(\Ehat)}
             {\sqrt{\Ehat^2-J^2}}\dd J\dd\Ehat\\
      &=\frac{2\pi}{\beta}\,e^{4\pi^2a/\beta}
        (1-e^{-4\pi^2/\beta})^2.
  \end{aligned}
  \label{eq:entropy-continuous-reference-transform}
\end{equation}
The substitution $u=\Ehat+J$, $v=\Ehat-J$ reduces the primary
integral to a product of two elementary Gaussian integrals.

At the level of the smoothed counts, Poisson summation controls this
replacement with relative error $O_P(c^{-P})$ for every fixed $P$,
at $E=\lambda c$ in the range of Proposition~\ref{prop:btz-entropy}.
Appendix~\ref{app:entropy-transfer} proves the uniform estimates,
including the descendant sum.

With $\beta'=4\pi^2/\beta$, the eta identity gives
\begin{equation}
  G(\beta)=\frac{\beta}{2\pi}
             e^{(\beta'-\beta)/12}G(\beta'),
  \qquad
  (1-e^{-\beta'})^2G(\beta')=V(\beta').
  \label{eq:entropy-eta-identity}
\end{equation}
Combining these factors yields the exact identity
\begin{equation}
  Z_{\mathrm{ref}}^{\mathbb R}(\beta)
    =e^{\beta/12}G(\beta)Z_{\mathrm{ref,prim}}(\beta)
    =e^{\pi^2c/(3\beta)}V(4\pi^2/\beta)
    =Z_{\mathrm{BTZ}}^{\mathrm{pert}}(\beta).
  \label{eq:entropy-reference-btz}
\end{equation}
The identity extends to $\Re\beta>0$, where the thermal integrals
are holomorphic.

Returning to the smoothed counts at $E=\lambda c$, the Poisson
comparison and inverse Laplace transformation give, for any $\gamma>0$,
\begin{equation}
  \begin{aligned}
    N_\phi^{\mathrm{ref}}(E)
      &=\bigl(1+O_P(c^{-P})\bigr)
        \int_{\gamma-i\infty}^{\gamma+i\infty}
        \frac{\dd\beta}{2\pi i}\,
        e^{\beta E}K_\phi(\beta)Z_{\mathrm{ref}}^{\mathbb R}(\beta)\\
      &=\bigl(1+O_P(c^{-P})\bigr)N_\phi^{\mathrm{BTZ}}(E).
  \end{aligned}
  \label{eq:entropy-reference-inversion}
\end{equation}
The second equality uses \eqref{eq:entropy-reference-btz} and the
definition \eqref{eq:entropy-btz-observable}. Together with
Lemma~\ref{lem:entropy-reference-comparison}, this proves
\eqref{eq:entropy-smooth-matching}. Taking logarithms and applying
the saddle-point expansion of \eqref{eq:entropy-btz-observable} then
proves Proposition~\ref{prop:btz-entropy}. We evaluate the first
coefficients below.

\paragraph{Entropy and its first corrections.}
Define
\begin{equation}
  A_\phi(\beta)=K_\phi(\beta)V(4\pi^2/\beta).
  \label{eq:entropy-smooth-amplitude}
\end{equation}
The inverse Laplace saddle gives
\begin{equation}
  N_\phi(\lambda c)
    =\frac{\sqrt6\,A_\phi(\beta_*)}
              {\sqrt c\,r^{3/2}}\,
          e^{\pi cr/3}
          \bigl(1+O(c^{-1})\bigr).
  \label{eq:entropy-smooth-count-expansion}
\end{equation}
Taking logarithms gives the constant term in
\eqref{eq:btz-entropy-leading}:
\begin{equation}
  C_{0,\phi}(\lambda)
    =\frac12\log6-\frac32\log r
       +\log K_\phi(\beta_*)+\log V(\widetilde\beta_*).
  \label{eq:entropy-smooth-constant}
\end{equation}
The logarithmic term comes from fixing the energy by one inverse
Laplace integral. The constant separates the saddle prefactor,
the energy kernel and the boundary-graviton factor.

Continuing the saddle expansion to the next order gives
\begin{equation}
  C_{1,\phi}(\lambda)
    =-\frac{3\beta_*}{16\pi^2}
       \left(
          3+12\beta_*\frac{A_\phi'(\beta_*)}{A_\phi(\beta_*)}
            +4\beta_*^2\frac{A_\phi''(\beta_*)}{A_\phi(\beta_*)}
       \right).
  \label{eq:entropy-smooth-first-coefficient}
\end{equation}
Higher entropy coefficients follow by carrying the saddle expansion
to higher orders and taking the logarithm.
The general differential operator and recurrence, together with a
bound on the rest of the inversion contour, are given in
Appendix~\ref{app:entropy-reference-expansion}.
The coefficients depend on $\lambda$ and the chosen kernel. They are
independent of the first-level data and admissible atomic choices
for the families in Lemma~\ref{lem:entropy-reference-comparison}.

\subsection{A macroscopic gap and BTZ matching}
\label{sec:entropy-positive-kappa}

We now turn to the growing-cutoff family of
part~\ref{item:positive-fraction} of Theorem~\ref{thm:analytic}, using
the schedule in Table~\ref{tab:construction-cell-parameters} and the
lexicographic atom selector of Section~\ref{sec:construction-moments}.
Fix the admissible auxiliary parameters, then fix a
sufficiently small $\kappa>0$. The first level has the measure
\eqref{eq:first-level-measure}, with total multiplicity at most a fixed
$M_*$ and spins in a fixed interval $|J|\le J_*$.
We take $a$ large enough that $\kappa a\ge J_*$.
The relevant scales are
\begin{equation}
  b=\kappa a,\qquad
  \vartheta=R_{\mathrm{init}}\kappa,\qquad U=\vartheta a,
  \qquad
  \Delta_g=(1+\kappa)a,\qquad E_g=\kappa a-\frac1{12}.
  \label{eq:entropy-kappa-scales}
\end{equation}
Here $\Delta_g$ is the exact first nonvacuum primary dimension, $E_g$
is its cylinder energy, and $U$ is the initial endpoint scale in shifted
energy. We use the same all-state counts $N_\phi(E)$ as in
\eqref{eq:entropy-smooth-observable}, with a fixed compact smooth kernel.
All limits in this subsection keep $\kappa$, the construction parameters,
$M_*$ and $J_*$ fixed. The estimates are uniform over the prescribed
first-level data within these bounds.

\paragraph{Below the primary gap.}
Let the center in dimension be $t=cx$, with $x$ in a compact subset of
$0<x<(1+\kappa)/12$. Only vacuum descendants contribute for
sufficiently large $c$. Their dimensions are integers;
smoothing at fixed width can therefore retain a dependence on the
fractional part of $t$. Define the integer sampling factor
\begin{equation}
  P_\phi(t)=\sum_{n\in\mathbb Z}\phi(n-t).
  \label{eq:entropy-kappa-periodization}
\end{equation}
The vacuum multiplicities give
\begin{equation}
  N_\phi\!\left(t-\frac c{12}\right)
    =\frac{\pi^2}{4\,3^{7/4}}\,
       t^{-9/4}e^{2\pi\sqrt{t/3}}P_\phi(t)
       \bigl(1+O(t^{-1/2})\bigr).
  \label{eq:entropy-kappa-vacuum-count}
\end{equation}
For the bump kernel $\phi=\phi_w$ in \eqref{eq:entropy-bump-kernel},
take $w>1/2$ fixed. Then $P_\phi(t)>0$ for every $t$, and
$\log P_\phi(t)$ is a bounded periodic function. The entropy therefore
has leading growth $2\pi\sqrt{t/3}$, logarithmic term
$-\tfrac94\log t$, and a bounded term containing $\log P_\phi(t)$.

\paragraph{Near the primary gap.}
The prescribed first level has bounded total multiplicity. The initial
discretization produces primary levels of exponentially large
multiplicity just above it. For every fixed primary spin $J$, their
dimensions approach $\Delta_g$ from above and their multiplicities
are asymptotic to $M_a e^{-\chi|J|}$, where
\begin{equation}
  M_a=\exp\!\left(aH-\{a\sqrt\vartheta\}\log r_{\mathrm{gap}}+C_{\mathrm{sc}}\right).
  \label{eq:entropy-kappa-packet-scale}
\end{equation}
The constants $H$, $r_{\mathrm{gap}}$, $C_{\mathrm{sc}}$ and $\chi$
depend on $\kappa$ and $R_{\mathrm{init}}$ for the atom selector and
initial degree prescription fixed above, and are independent of the
prescribed first-level spins and multiplicities. They satisfy
$4\pi\sqrt\kappa<H<4\pi\sqrt\vartheta$, $0<r_{\mathrm{gap}}<1$,
and $\chi>0$. Their explicit expressions and the derivation of the
bounds on $H$ are given in Appendix~\ref{app:entropy-positive-kappa}.
In particular, even the leading entropy rate $H$ depends on the construction.
The braces denote the fractional part and arise from the initial
degree $\lfloor a\sqrt\vartheta\rfloor$.
Appendix~\ref{app:entropy-positive-kappa} shows that changing
the bounded first-level data preserves these asymptotics for this selector
and degree rule.

Take the bump kernel $\phi=\phi_w$ in \eqref{eq:entropy-bump-kernel}
with fixed $w>1/2$, and evaluate the count at a fixed offset $u\ge0$
from $E_g$. The descendants of these levels are sampled by
\begin{equation}
  \mathcal B_\phi(u)=\sum_{N=0}^\infty p_2(N)\phi(N-u),
  \qquad
  \frac{N_\phi(E_g+u)}{M_a}
       \longrightarrow\coth(\chi/2)\mathcal B_\phi(u).
  \label{eq:entropy-kappa-gap-count}
\end{equation}
The factor $\coth(\chi/2)=\sum_{J\in\mathbb Z}e^{-\chi|J|}$
includes all primary spins, while $p_2(N)$, defined in
\eqref{eq:entropy-descendant-multiplicity}, counts their descendants
at total level $N$. The prescribed first level and its descendants
contribute only $d_{\mathrm{first}}\mathcal B_\phi(u)$, which vanishes
after division by $M_a$ at fixed $u$. For this kernel,
$\mathcal B_\phi(u)>0$ for every
$u\ge0$, so taking the logarithm gives
\begin{equation}
  \begin{aligned}
    S_\phi(E_g+u)
      ={}&aH-\{a\sqrt\vartheta\}\log r_{\mathrm{gap}}+C_{\mathrm{sc}}\\
       &+\log\coth(\chi/2)+\log\mathcal B_\phi(u)+o(1).
  \end{aligned}
  \label{eq:entropy-kappa-gap-entropy}
\end{equation}
There is no $\log a$ term, and the $O(1)$ term need not converge
because it depends on the fractional part $\{a\sqrt\vartheta\}$.

\paragraph{An interval without BTZ matching.}
The same high-multiplicity primary levels contribute descendants at
dimensions a distance of order $a$ above the gap. Their contribution
alone establishes a mismatch for the full spectrum.

\begin{proposition}[An interval without BTZ matching]
\label{prop:entropy-kappa-nonbtz}
For the construction above, take the bump kernel $\phi=\phi_w$ of
\eqref{eq:entropy-bump-kernel} with fixed $w>1/2$, and define
\begin{equation}
  \epsilon_H=\frac{H^2}{16\pi^2},
  \qquad \kappa<\epsilon_H<\vartheta.
  \label{eq:entropy-kappa-nonbtz-endpoint}
\end{equation}
The inequalities follow from the bounds on $H$ given above.
Fix $\kappa<\epsilon<\epsilon_H$ and set
$E=a\epsilon-1/12$, equivalently $\Delta=a(1+\epsilon)$.
Then
\begin{equation}
  \liminf_{a\to\infty}\frac1a
       \log\frac{N_\phi(a\epsilon-1/12)}
                     {N_\phi^{\mathrm{BTZ}}(a\epsilon-1/12)}
       \ge H-4\pi\sqrt\epsilon>0.
  \label{eq:entropy-kappa-nonbtz-rate}
\end{equation}
\end{proposition}

Indeed, the descendant level is of order $a(\epsilon-\kappa)$,
and the large multiplicity $M_a$ gives an entropy contribution
$aH+O(\sqrt a)$. The BTZ reference at the same cylinder energy has
entropy $4\pi a\sqrt\epsilon+O(\log a)$. Positivity of the spectrum
then gives the stated lower bound. Appendix~\ref{app:entropy-positive-kappa}
retains the positive $\sqrt a$ correction and records the required
uniform descendant estimate. The proposition is a statement about
the full smoothed density, not a claim that these descendants exhaust
it. The bound does not determine the full intermediate-energy profile
or identify $\epsilon_H$ as a transition point.

\paragraph{BTZ matching at higher energies.}
Above the initial endpoint scale, the all-order comparison can be
recovered, with a separate estimate for the initial primary levels
and all their descendants.

\begin{proposition}[BTZ matching above the initial scale]
\label{prop:entropy-kappa-btz}
For the fixed positive-$\kappa$ construction above, let $\phi$ be any
fixed nonnegative normalized kernel in $C_c^\infty(\R)$, and let
$E=\lambda c$ exactly, with $\lambda$ in a compact subset of
$(\vartheta/12,\infty)$. For every fixed integer $P\ge0$,
\begin{equation}
  \frac{N_\phi(\lambda c)}{N_\phi^{\mathrm{BTZ}}(\lambda c)}
       =1+O_P(c^{-P}).
  \label{eq:entropy-kappa-btz-matching}
\end{equation}
The entropy has the expansion \eqref{eq:btz-entropy-leading}, with
the same coefficients $C_{m,\phi}(\lambda)$ as in the fixed-$B$
family, to every fixed finite order.
\end{proposition}

The total primary multiplicity from the initial stage is bounded by
a polynomial in $a$ times $e^{4\pi a\sqrt\vartheta}$. Its descendants
at energies of order $c$ add at most $e^{O(\sqrt a)}$. Since
$12\lambda>\vartheta$ by a fixed positive margin, these contributions
are exponentially smaller than the BTZ count. On the remaining
short cells, exact moments and smooth polynomial approximation give
the all-order comparison. The spin integral and exact reference
transform are then those of Section~\ref{sec:entropy-matching}.
Appendix~\ref{app:entropy-positive-kappa} records the additional estimates.

The coefficients are independent of $\kappa$ and the construction
parameters at a common admissible energy and kernel, while the error
bounds and the range proved here can depend on those parameters.
No uniform limit as $\kappa\to0$ is taken. The condition
$12\lambda>\vartheta$ is sufficient for matching; it does not locate
a sharp boundary between matching and nonmatching spectra.

The asymptotic formulas and error bounds in this subsection also hold
for a Gaussian kernel of fixed width, with the corresponding
kernel-dependent factors. In this case, $P_\phi(t)$ has a positive
lower bound independent of $t$, and $\mathcal B_\phi(u)>0$ for every fixed $u$.
The additional tail estimates are given in
Appendix~\ref{app:entropy-positive-kappa}.

\section{Discussion}
\label{sec:discussion}

The construction turns the modular completion of the Virasoro vacuum
\cite{MaloneyWitten2007,KellerMaloney2014} into a locally finite spectrum
with positive integer multiplicities. Exact local repair allows prescribed
changes in a bounded energy band, while moment matching controls their
effect at higher energies. Iterating these operations preserves the vacuum,
modular invariance and uncensored support. The resulting spectra let us
study both large primary gaps and BTZ entropy within the same genus-one
framework.

\paragraph{Large gaps and genus-one constraints.}
Theorem~\ref{thm:analytic} shows that the spectral conditions in
Definition~\ref{def:class} allow a primary dimension gap with
$\Delta_1/c\to(1+\kappa)/12$ for a nonempty range of fixed
$\kappa>0$. They also allow $\Delta_1=a+\delta$ for every fixed
$\delta\ge0$, where $a=(c-1)/12$. Both constructions apply at every
sufficiently large real $c$, with a prescribed finite set of spins and
positive integer multiplicities at the first primary level. The total
first-level multiplicity is bounded independently of $c$, and the
positive-$\kappa$ construction allows any fixed finite spin set once
$\kappa a$ exceeds its largest absolute spin.
Thus positivity and integer multiplicities are compatible with both an extensive excess above $c/12$ and a finite shifted gap.
Our construction gives a lower bound on the maximal primary dimension
gap attainable within the genus-one spectral class, complementing the
upper bounds obtained from the modular bootstrap.
The sufficient range $0<\kappa<\kappa_0$ comes from the estimates used
in the construction; determining its optimal range remains open.
Changing the energy region used for local repair, the initial cells
or the moment degrees offers ways to improve this sufficient range.
Each modification must preserve the positivity bounds needed for integer
moment matching and the complete exterior estimates. It can also change
the near-gap spectrum, so its entropy requires an analysis of the same
modified prescription.

The primary dimension gap and the chiral support condition play different
roles. All nonvacuum primaries obey $\Ehat\ge|J|$, irrespective of
the chosen gap. In the finite-gap family, the first primary has cylinder energy
$E_1=\delta-1/12$: $\delta=0$ gives the scalar uncensored threshold,
whereas $\delta=1/12$ gives the classical massless BTZ dimension.
This finite shift matters when comparing individual levels with the
bulk threshold, even though both choices have $\Delta_1/c\to1/12$.

\paragraph{Near-extremal negativity and modular ambiguities.}
The odd-spin negative region of the vacuum completion
\cite{BenjaminEtAl2019,AldayBae2020} persists at every finite stage
of our recursion. Fixing $a$ and the number of completed steps,
then taking odd $J\to\infty$, gives an exponentially large
negative weight in an exponentially narrow interval above
$\Ehat=J$. Positivity of the completed spectrum therefore requires
the remaining steps to contribute at least this weight in those
intervals. Such a contribution is compatible with convergence in
every fixed-temperature thermal norm: the Boltzmann factor suppresses
large $J$, whereas the near-extremal comparison follows a shrinking
window whose energy tends to infinity.

Appendix~\ref{app:global-near-extremal} proves these statements and
shows that the natural correcting seeds have total absolute weight
growing faster than any polynomial. This places their sum outside
the seed-growth control used in the uncensored-ambiguity argument
of Alday and Bae~\cite[Secs.~4.1 and 4.3]{AldayBae2020}.
The finite-seed asymptotics cannot be passed through the recursive
limit uniformly on the near-extremal windows. The same appendix
quantifies how the specified lexicographic choice produces primaries
approaching the chiral threshold at arbitrarily large spin.
This is compatible with a large primary dimension gap and local
finiteness, since the accumulating states have divergent total
dimension; it is also consistent with the twist-accumulation
theorem of Pal, Qiao and Rychkov~\cite{PalQiaoRychkov2022}.

\paragraph{BTZ entropy and the primary gap.}
For the family with fixed shifted gap $\delta$, constructed with
$B>\delta$ held fixed, Proposition~\ref{prop:btz-entropy} identifies
the smoothed density of states with the prediction of a single
perturbative BTZ saddle through every fixed finite order in $1/c$.
The resulting entropy reproduces the Bekenstein--Hawking term, the
logarithmic correction and the subsequent coefficients, including the
boundary-graviton contribution of
\eqref{eq:entropy-btz-partition}~\cite{GiombiMaloneyYin2008,CotlerJensen2019}.
This gives a realization of the BTZ state-growth relation
\cite{Strominger1998} in positive integer modular spectra with specified
primary gaps and first-level spin multiplicities.

The result applies at $E=\lambda c$ for every fixed $\lambda>0$.
In particular, it includes
$0<E/c<1/12$, where the general HKS assumptions leave the entropy
undetermined~\cite{HartmanKellerStoica2014}. In this range, the inverse
Laplace saddle lies at $\beta_*>2\pi$, where the BTZ contribution is
subdominant to thermal AdS in the canonical ensemble. Direct spectral
control supplies information that the dominant canonical free energy
does not determine. The agreement holds to every finite order in an
asymptotic expansion; the estimates do not fix an exponentially small
remainder.

A macroscopic primary gap produces a different pattern for the
explicit schedule and lexicographic atom choice studied in
Section~\ref{sec:entropy-positive-kappa}.
At energies an extensive distance below the primary gap, the count
smoothed with a fixed compactly supported kernel comes entirely from
vacuum descendants. Immediately above the prescribed first level,
the initial discretization produces additional primary levels with
exponentially large multiplicities. Their entropy rate $H$ depends
on the construction. For the compact smooth bump kernel specified in
that section, their descendants give an exponential excess over the
BTZ prediction throughout the interval established in
Proposition~\ref{prop:entropy-kappa-nonbtz}. At higher energies,
Proposition~\ref{prop:entropy-kappa-btz} restores the same BTZ expansion
as in the fixed-$B$ family, provided
$E/c>\vartheta/12$ by a fixed margin, with
$\vartheta=R_{\mathrm{init}}\kappa$. The nonmatching interval and this
sufficient matching range leave an intermediate region unresolved.
These estimates do not determine the precise energy range over which
BTZ matching holds.

\paragraph{Bulk interpretation of the spectral completion.}
The Poincar\'e sum of the Virasoro vacuum combines the perturbative
contributions of thermal AdS$_3$ and its modular images, which fill
the boundary torus by hyperbolic solid tori. These exhaust the smooth,
complete, real Euclidean Einstein metrics whose only end is the
prescribed torus conformal boundary
\cite[Sec.~2.1]{MaloneyWitten2007}.
A bulk origin for our spectral corrections would therefore require
extending this prescription.

In AdS$_3$, conical-defect orbifolds have been used to address negative
spectral densities, while introducing primaries below the black-hole
threshold and retaining a continuum above it
\cite{BenjaminCollierMaloney2020}.
Off-shell Seifert topologies provide another route: the near-extremal
analysis of Maxfield and Turiaci replaces the negative-density
expansion by a positive spectral edge with a nonperturbative shift
\cite{MaxfieldTuriaci2020}. These contributions involve metric
integrals on topologies without the required Einstein saddle.
Extending that analysis to a global prescription requires control
of the mapping class group and the off-shell amplitudes, as emphasized
by subsequent Virasoro-TQFT and surgery studies
\cite{Yan2025OffShell,DeBoerKamesKingPost2025}.
Cusp counterterms offer a prescription reproducing the regularized
vacuum modular sum, while retaining its spectral problems
\cite{StanfordYan2025Cusps}.

For $1/8$-BPS black holes in four-dimensional $\mathcal N=8$
supergravity, supersymmetric localization of the gravitational path
integral, including orbifold saddles and their quantum contributions,
reproduces the exact integer index and, with additional analysis,
relates it to an extremal degeneracy~\cite{IliesiuMurthyTuriaci2022}.

These examples motivate asking whether an extended bulk prescription,
including additional saddles or off-shell contributions, could reproduce some of our
discrete spectra with positive integer multiplicities. In the
fixed-$B$ family, such a prescription would have to recover the
individual levels and multiplicities while retaining the established
smoothed BTZ entropy expansion. It might also constrain the freedom
in the spectral construction. Identifying such bulk contributions
remains open.

\paragraph{Spectral freedom and ensemble interpretations.}
The fixed-$B$ construction gives distinct genus-one partition functions
at the same central charge, each with a discrete spectrum and positive
integer multiplicities. For the same fixed smoothing kernel, their
entropies share the perturbative BTZ expansion through every fixed
finite order in $1/c$. Within the fixed-$B$ families of
Lemma~\ref{lem:entropy-reference-comparison}, both the prescribed
first-level data, with total multiplicity bounded independently of $c$,
and the admissible atomic replacements can vary without changing the
perturbative expansion of the smoothed entropy. Thus the perturbative
entropy expansion does not determine
the individual levels or their multiplicities.

This freedom in microscopic spectra is relevant to statistical
descriptions of AdS$_3$ gravity. Torus-wormhole calculations
\cite{CotlerJensen2020Random,CotlerJensen2020Bootstrap}, ensembles
of CFT data and the statistical conformal bootstrap
\cite{ChandraEtAl2022Ensemble,BelinEtAl2023Approximate,JafferisRozenbergWong2024}
provide complementary approaches. Recent work on sums over
topologies further explores constraints on their statistical
completion~\cite{BelinEtAl2026Topologies}.
The common entropy expansion is compatible with these ensemble
interpretations, as well as with coarse-grained descriptions of
individual CFTs
\cite{DiUbaldoPerlmutter2023RMT}.

An ensemble interpretation would require a probability measure on the
set of genus-one partition functions produced by the fixed-$B$
construction at a common central charge.
The measure is additional physical input: for example, large-gap
CFTs are atypical in the regulated maximum-ignorance ensemble
studied in Ref.~\cite{BelinMaloneySeefeld2025}. That conclusion
depends on the proposed measure and its assumptions and does not
exclude an ensemble restricted to the spectra considered here.
Moreover, a positive average of nonnegative spectral measures is
nonnegative, so the uncorrected signed MWK spectrum cannot be its
exact average.
A natural question is whether a suitable
measure can reproduce the connected spectral correlations associated
with torus wormholes~\cite{BoruchEtAl2025RMT}.

\paragraph{Explicit computations and further questions.}

An important next step is the certified computation of individual
primary energies determined by the repaired moments. The compact atom
selector provides an exact mathematical prescription; evaluating it
requires a terminating procedure that encloses the energies, resolves
all coincident contributions to their integer multiplicities, and
verifies preservation under subsequent recursion. Determining the
lowest such level across all spins is a further question.

Further analytic questions concern the full intermediate-energy profile
at positive $\kappa$, the approach to the finite-energy threshold, and
counts with total spin of order $c$. The fixed-parameter limits proved
here do not settle simultaneous limits in which $\kappa$ or $E/c$
tends to zero.

\paragraph{Towards a full CFT realization.}
A central question is which of the constructed spectra admit OPE
coefficients defining a unitary compact CFT. For a member with a scalar
at its first primary level, the four-point function of a putative
Hermitian primary $\mathcal O$ offers a concrete first test: can the
allowed dimensions and even spins support nonnegative sums of squared
OPE coefficients that satisfy crossing? At each degenerate level,
this correlator sees the sum over the corresponding primary operators.
One can begin with necessary conditions using the gap and spectral
support, then use explicitly computed levels to constrain individual
members. The external dimension here grows with $c$, so these tests
must accommodate heavy operators.

A completion would have to use the same OPE coefficients, including
their multiplicity labels, across mixed four-point functions and
modular-covariant torus one-point functions
\cite{BelinEtAl2023Approximate,JafferisRozenbergWong2024}.
The same data could then be sewn into a genus-two partition function.
Consistency would require agreement between different pants
decompositions and factorization onto the prescribed torus spectrum
and the same lower-genus correlation functions in degeneration limits.
The genus-two modular bootstrap provides a framework for testing
constraints on these structure constants
\cite{ChoCollierYin2017}.

These conditions could select a smaller family of spectra, restrict
the gap or initial data, and determine how much microscopic freedom
survives beyond genus one. A longer-term goal is to choose the discrete
spectrum and compatible OPE data together, using the freedom left by
the moment constraints. Passing a finite set of crossing or sewing
tests would be only a first step. Establishing a full CFT realization
would also require control of the high-energy contributions weighted
by OPE coefficients; the thermal convergence proved here controls
spectral multiplicities.

\section*{Acknowledgements}
\addcontentsline{toc}{section}{Acknowledgements}
We made extensive use of AI in developing the spectral construction and
its proofs, exploring generalizations, and generating the Lean~4 proof
code. AI tools were also used in preparing the manuscript. The authors
directed the research and take responsibility for the content of this paper.
C.C. is supported by NSFC Grant No.~12575075.
WJM is supported by the National Natural Science Foundation of China
No.~12405082 and Shanghai Pujiang Program No.~24PJA118.

\appendix
\section{Canonical completion and kernel estimates}
\label{app:canonical}

We derive the spectral measure and kernel estimates used in
Section~\ref{sec:canonical-completion}. The proof first identifies
the induced measure at each fixed output spin, then obtains the
joint arithmetic bound needed to sum over spins. These estimates
also justify integration against finite signed seed measures,
including measures whose scalar support approaches zero.
We use the canonical Poincar\'e continuation, regular at $s=1/2$,
with the normalization of Ref.~\cite[Secs.~3--4]{KellerMaloney2014}
as an external analytic input. The measure calculation and the joint
bound below are derived with this same continuation.

\subsection{Canonical continuation and the unit-seed measure}
\label{app:canonical-measure}

\paragraph{Dependence on the seed energy.}
For fixed $J'$, write $\gamma\tau=x_\gamma+iy_\gamma$. Subtracting
the seed at zero energy in \eqref{eq:poincare-seed} gives
\begin{align}
  \mathcal P_{\Ehat',J'}(\tau;s)
  ={}&\mathcal P_{0,J'}(\tau;s)\nonumber\\
   &+\sum_{\gamma\in\Gamma_\infty\backslash\Gamma}
      y_\gamma^s e^{2\pi iJ'x_\gamma}
      \bigl(e^{-2\pi\Ehat'y_\gamma}-1\bigr).
  \label{eq:canonical-energy-subtraction}
\end{align}
On a compact subset of $\HH$, the heights $y_\gamma$ are bounded.
For $\Ehat'$ in a compact complex set, the bracket is bounded by
a constant times $y_\gamma$. The sum therefore converges absolutely
and uniformly on compact sets with $\operatorname{Re}s>0$, since
the extra power of $y_\gamma$ puts it in the convergent Eisenstein
range. It is holomorphic in $\Ehat'$ and $s$ there. Continuing the
first term with the stated convention proves the local uniform
bound in Section~\ref{sec:canonical-completion}. The reference
$\mathcal P_{0,J'}$ in this subtraction is an auxiliary continued
function; no zero-energy seed at nonzero spin is added to the spectrum.

\paragraph{Fourier coefficients.}
For $\operatorname{Re}s>1$, separate the identity coset and write
\begin{equation}
  \mathcal P_{\Ehat',J'}(x+iy;s)
   =y^s e^{-2\pi y\Ehat'+2\pi ixJ'}
      +\sum_{J\in\Z} A_{J;\Ehat',J'}(y;s)e^{2\pi iJx}.
  \label{eq:canonical-fourier-expansion}
\end{equation}
A nonidentity left coset has a primitive lower row $(m,d)$, with
$m\ge1$ and $d\in\Z$. Write $d=r+m\ell$, where $r$ is a unit
modulo $m$ and $\ell\in\Z$. If the upper-left entry is $\alpha$,
then $\alpha r\equiv1\pmod m$. With $t=x+d/m$,
\[
  \operatorname{Im}\gamma\tau=\frac{y}{m^2(t^2+y^2)},\qquad
  \operatorname{Re}\gamma\tau
     =\frac{\alpha}{m}-\frac{t}{m^2(t^2+y^2)}.
\]
The sum over $\ell$ unfolds the Fourier integral to the real line;
the remaining residue-class phase is
$\exp[2\pi i(Jr+J'r^{-1})/m]$. Consequently,
\begin{equation}
  A_{J;\Ehat',J'}(y;s)
   =\sum_{m=1}^{\infty}\frac{S(J,J';m)}{m^{2s}}
      y^s\int_{\R}\frac{e^{-2\pi iJt}}{(t^2+y^2)^s}
      \exp\!\left[-\frac{2\pi(\Ehat'y+iJ't)}
                         {m^2(t^2+y^2)}\right]\dd t.
  \label{eq:canonical-fourier-integral}
\end{equation}
The finite sums $S(J,J';m)$ need no regularization. Continuation
acts on the complete Fourier coefficient in the variable $s$.

Let $A^{[0]}$ denote the contribution obtained by replacing the
last exponential in \eqref{eq:canonical-fourier-integral} by $1$.
It is independent of $\Ehat'$. For scalar output,
\begin{equation}
  A^{[0]}_{0;\Ehat',J'}(y;s)
    =y^{1-s}\frac{\sqrt\pi\,\Gamma(s-1/2)}{\Gamma(s)}
       \mathcal Z_{0,J'}(s).
  \label{eq:canonical-scalar-zero-order}
\end{equation}
For $n\ge1$, define $\sigma_\alpha(n)=\sum_{d\mid n}d^\alpha$,
so that $\sigma_0(n)=d(n)$. The Ramanujan-sum identities give
\begin{equation}
  \mathcal Z_{0,J'}(s)=
  \begin{cases}
    \sigma_{1-2s}(|J'|)/\zeta(2s),&J'\ne0,\\
    \zeta(2s-1)/\zeta(2s),&J'=0.
  \end{cases}
  \label{eq:canonical-ramanujan}
\end{equation}
Here $\zeta$ is the Riemann zeta function. At $s=1/2$, the zero
of $1/\zeta(2s)$ cancels the pole of $\Gamma(s-1/2)$.
The remaining coefficient is $2d(|J'|)$ for $J'\ne0$ and
$2\zeta(0)=-1$ for $J'=0$. Hence
\begin{equation}
  A^{[0]}_{0;\Ehat',J'}(y;1/2)=c_{J'}\sqrt y,
  \qquad c_0=-1,\qquad c_{J'}=2d(|J'|)\quad(J'\ne0).
  \label{eq:canonical-threshold-calculation}
\end{equation}
After division by $\sqrt y$, this is the Laplace transform of
the scalar atom $c_{J'}\delta_0$.

\paragraph{The continuous part.}
Denote the cosine-sum term of \eqref{eq:repair-Q-kernel} by
$Q^{(\ge1)}$, so that
\begin{equation}
  Q_{J,J'}(\Ehat,\Ehat')
    =2\mathbf1_{JJ'\ne0}\mathcal Z_{J,J'}(1/2)
       +Q^{(\ge1)}_{J,J'}(\Ehat,\Ehat').
  \label{eq:canonical-kernel-split}
\end{equation}
The bracket in the cosine sum is bounded by
$C\min\{1,\Ehat\Ehat'/m^2\}$ for physical real energies.
Since $|S(J,J';m)|\le m$,
splitting the sum at $m\sim\sqrt{\Ehat\Ehat'}$ proves
\begin{equation}
  |Q^{(\ge1)}_{J,J'}(\Ehat,\Ehat')|
        \le C\sqrt{\Ehat\Ehat'},\qquad
  |Q_{0,J'}(\Ehat,\Ehat')|\le C\Ehat\Ehat'.
  \label{eq:canonical-higher-kernel-bounds}
\end{equation}
When $\Ehat\Ehat'<1$, the unsplit quadratic bound is
$O(\Ehat\Ehat')$ and is stronger than the first inequality.
For scalar output the arithmetic term in
\eqref{eq:canonical-kernel-split} is absent, and summing the
quadratic bound gives the second inequality. On compact complex
energy sets, each full summand is $O(m^{-2})$ uniformly, proving
the holomorphy stated after \eqref{eq:repair-Q-kernel}.

\begin{proposition}[Induced measure of a unit seed]
\label{prop:canonical-unit-measure}
For $\Ehat'\ge|J'|$ and each fixed output spin $J$, the induced
Fourier coefficient has the unique representation
\begin{equation}
  \frac{A_{J;\Ehat',J'}(y;1/2)}{\sqrt y}
    =\int_{|J|}^{\infty}e^{-2\pi y\Ehat}
        Q_{J,J'}(\Ehat,\Ehat')\dd\omega_J(\Ehat)
       +\delta_{J,0}c_{J'}
  \label{eq:canonical-unit-laplace}
\end{equation}
by a signed measure on $[|J|,\infty)$ with finite thermal total
variation at every $y>0$. The density term is defined on
$\Ehat>|J|$ and gives zero weight to the endpoint. The only
induced atom is $c_{J'}\delta_0$ in the scalar row.
\end{proposition}

\begin{proof}
The proposed density defines a real measure: the Kloosterman sums
are real by pairing $r$ with $-r$, and their continued values at a
regular real parameter are real. At scalar output,
\eqref{eq:canonical-higher-kernel-bounds} cancels the $1/\Ehat$
in $\dd\omega_0$. At nonzero spin, the kernel is bounded near
$\Ehat=|J|$, where $\dd\omega_J$ has an integrable square-root
singularity. At infinity, the finite arithmetic constant and
the first bound in \eqref{eq:canonical-higher-kernel-bounds}
give growth at most $O_{J,J',\Ehat'}(1+\sqrt{\Ehat})$.
Thus the measure is locally finite and has finite thermal total
variation at fixed $J$, without using a uniform arithmetic bound.

To identify its Laplace transform, first subtract $A^{[0]}$.
The product of two Gaussian integrals gives
\begin{align}
 &\frac{\sqrt y}{m\sqrt{t^2+y^2}}
     \exp\!\left[-\frac{2\pi(\Ehat'y+iJ't)}{m^2(t^2+y^2)}\right]
       \nonumber\\
 &\qquad=\frac{\sqrt y}{m}\int_{\R^2}
     e^{-\pi(y-it)u^2-\pi(y+it)v^2}
     e^{(2\pi i/m)(\sqrt{\Ehat'+J'}\,u+
                       \sqrt{\Ehat'-J'}\,v)}\dd u\dd v.
  \label{eq:canonical-gaussian-transform}
\end{align}
Let $H_m(t)$ be the difference between the left side and the same
expression with its final exponential replaced by $1$.
On the right, make the change of variables
\[
  \Ehat=\frac{u^2+v^2}{2},\qquad
  \ell=\frac{u^2-v^2}{2},\qquad
  \dd u\dd v=\frac{\dd\Ehat\dd\ell}
                       {2\sqrt{\Ehat^2-\ell^2}}
\]
in each quadrant. Here $\ell$ is a real Fourier variable, and
$\Ehat>|\ell|$. The four sign choices give four times the cosine
product minus $1$. Define
\[
  \Phi_m(\Ehat,\ell)=
    \cos\!\left(\frac{2\pi}{m}
        \sqrt{(\Ehat'+J')(\Ehat+\ell)}\right)
    \cos\!\left(\frac{2\pi}{m}
        \sqrt{(\Ehat'-J')(\Ehat-\ell)}\right).
\]
The Gaussian factor becomes
$e^{-2\pi y\Ehat+2\pi it\ell}$. Fourier inversion therefore gives
\begin{equation}
  \int_{\R}e^{-2\pi iJt}H_m(t)\dd t
   =\frac{2\sqrt y}{m}\int_{|J|}^{\infty}e^{-2\pi y\Ehat}
       [\Phi_m(\Ehat,J)-1]\dd\omega_J(\Ehat).
  \label{eq:canonical-higher-transform}
\end{equation}
To justify inversion, the subtraction makes
$H_m(t)=O(|t|^{-2})$ at infinity. Its inverse Fourier density is
continuous and integrable: the bound
$|\Phi_m-1|\le C\min\{1,\Ehat\Ehat'/m^2\}$ controls the origin,
the spin edges are integrable, and the exponential controls infinity.
Summing \eqref{eq:canonical-higher-transform} with the remaining
factor $S(J,J';m)$ gives the thermal transform of $Q^{(\ge1)}$.
The same bounds justify the sum--integral interchange. This
subtracted Fourier expression continues to $\operatorname{Re}s>0$
by the extra denominator decay, in agreement with
\eqref{eq:canonical-energy-subtraction}.

For the zeroth-order term at $J\ne0$, the continued Fourier
integral instead gives
\begin{equation}
  \frac{A^{[0]}_{J;\Ehat',J'}(y;1/2)}{\sqrt y}
     =2\mathcal Z_{J,J'}(1/2)K_0(2\pi|J|y),\qquad
  K_0(2\pi|J|y)
     =\int_{|J|}^{\infty}e^{-2\pi y\Ehat}\dd\omega_J(\Ehat).
  \label{eq:canonical-bessel-transform}
\end{equation}
The latter identity follows from
$K_0(z)=\int_0^\infty e^{-z\cosh u}\dd u$ by setting
$\Ehat=|J|\cosh u$. When $J'=0$, symmetry and
\eqref{eq:canonical-ramanujan} give
$\mathcal Z_{J,0}(1/2)=0$. For $J=0$, the zeroth-order term
is the separate scalar atom of
\eqref{eq:canonical-threshold-calculation}, rather than a density.
These cases give exactly the indicator in
\eqref{eq:canonical-kernel-split} and prove
\eqref{eq:canonical-unit-laplace}.

Finally, the difference of two measures in the stated class has a
Laplace transform holomorphic in the right half-plane. If it
vanishes on the positive real axis, it vanishes throughout that
half-plane. For each $\beta_0>0$, the finite signed measure
obtained by multiplying the difference by $e^{-\beta_0\Ehat}$
then has zero Fourier transform and is zero. Removing the positive
weight proves uniqueness, excluding any additional edge atoms.
\end{proof}

For the auxiliary seed $\Ehat'=J'=0$, the kernel vanishes.
The induced atom $-\delta_0$ cancels the direct unit atom, giving
$\mathcal P_{0,0}(\tau;1/2)=0$. Actual seeds in $\mathcal M_B$
exclude this origin; the separate threshold term remains essential
when their scalar energies approach zero.

\subsection{The joint arithmetic bound}
\label{app:canonical-arithmetic}

Proposition~\ref{prop:canonical-unit-measure} is a fixed-spin
statement. We now control the continued arithmetic term uniformly
in the output spin, as required by
\eqref{eq:repair-compact-kernel-bound}.

\begin{lemma}[Joint bound for the continued zeta function]
\label{lem:canonical-joint-bound}
There is an absolute constant $C$ such that
\begin{equation}
  |\mathcal Z_{J,J'}(1/2)|
    \le C\sqrt{|JJ'|}\,
       \sigma_{-1}\bigl(\gcd(|J|,|J'|)\bigr),\qquad JJ'\ne0.
  \label{eq:canonical-joint-zeta-bound}
\end{equation}
Consequently \eqref{eq:repair-compact-kernel-bound} holds for all
uncensored input and output pairs with $0<\Ehat'\le B$.
\end{lemma}

\begin{proof}
Consider the single auxiliary function $\mathcal P_{1,1}$.
At a fixed height $y_0>0$, subtract its direct term, its scalar
term $2\sqrt{y_0}$, and its higher-order Fourier series.
The latter converges absolutely before any arithmetic estimate:
\eqref{eq:canonical-higher-kernel-bounds} gives
\[
  \sum_{J\in\Z}\int_{|J|}^{\infty}
       e^{-2\pi y_0\Ehat}\sqrt{\Ehat}\dd\omega_J(\Ehat)
   \le\sum_{J\in\Z}e^{-2\pi y_0|J|}
          \int_0^\infty e^{-2\pi y_0 s}\frac{\dd s}{\sqrt s}
   <\infty.
\]
The remaining periodic function is bounded. Its nonzero
Fourier coefficients are
\[
  2\mathcal Z_{J,1}(1/2)\sqrt{y_0}K_0(2\pi|J|y_0).
\]
The large-argument asymptotic of $K_0$
\cite[Eq.~10.40.2]{NISTDLMF} implies
\begin{equation}
  |\mathcal Z_{J,1}(1/2)|
     \le C_{y_0}\sqrt{|J|}\,e^{2\pi|J|y_0},\qquad J\ne0.
  \label{eq:canonical-fixed-height-bound}
\end{equation}
Given any positive height $y$, first choose $y_0<y$ in this bound.
It follows that the zeroth-order Fourier series converges
absolutely at height $y$. In particular, using a fixed $y_0$ for
large $y$, its nonzero-spin terms decay exponentially in the cusp.
The scalar refinement in \eqref{eq:canonical-higher-kernel-bounds}
then gives
\[
  \mathcal P_{1,1}(x+iy)=2\sqrt y+O(y^{-1/2}),\qquad y\to\infty,
\]
uniformly in $x$.

Introduce the modular comparison function
\begin{equation}
  \Theta(x+iy)=\sqrt y\sum_{p,r\in\Z}
      e^{-\pi y(p^2+r^2)+2\pi iprx}.
  \label{eq:canonical-theta-comparison}
\end{equation}
Translation invariance follows from $pr\in\Z$; inversion follows
from two-dimensional Gaussian Poisson summation. Its cusp term is
$\sqrt y+O(\sqrt y\,e^{-\pi y})$. Thus
$\mathcal P_{1,1}-2\Theta$ is bounded on a fundamental domain
and, by modularity, on all of $\HH$. This comparison is used only
to estimate the arithmetic term.

Now set $k=|J|\ge1$ and evaluate the $J$th coefficient at $y=1/k$.
The zeroth-order contribution is
$2\mathcal Z_{J,1}(1/2)k^{-1/2}K_0(2\pi)$.
The higher-order term is bounded independently of $J$ by
\[
  C\int_1^\infty
       \frac{\sqrt v\,e^{-2\pi v}}{\sqrt{v^2-1}}\dd v,
\]
after setting $\Ehat=kv$. The direct term is bounded as well.
There are $2d(k)$ pairs with $pr=J$, and $p^2+r^2\ge2k$.
Therefore the absolute value of the theta coefficient is at most
$2d(k)k^{-1/2}e^{-2\pi}\le4e^{-2\pi}$.
The global bound on $\mathcal P_{1,1}-2\Theta$ and
$K_0(2\pi)>0$ prove
\begin{equation}
  |\mathcal Z_{J,1}(1/2)|\le C\sqrt{|J|}.
  \label{eq:canonical-single-input-bound}
\end{equation}

The finite Selberg identity
\cite[Eq.~(1.1)]{Xi2023Selberg} holds for signed integer indices:
\[
  S(J,J';m)=\sum_{\ell\mid\gcd(|J|,|J'|,m)}
      \ell\,S(JJ'/\ell^2,1;m/\ell).
\]
Summing first for $\operatorname{Re}s>1$ and then continuing gives
\begin{equation}
  \mathcal Z_{J,J'}(s)
    =\sum_{\ell\mid\gcd(|J|,|J'|)}
        \ell^{1-2s}\mathcal Z_{JJ'/\ell^2,1}(s).
  \label{eq:canonical-selberg-zeta}
\end{equation}
Applying \eqref{eq:canonical-single-input-bound} at $s=1/2$
proves \eqref{eq:canonical-joint-zeta-bound} for either sign of
$JJ'$. Finally,
$\sigma_{-1}(n)\le1+\log n$,
$|J'|\le\Ehat'\le B$, and $|J|\le\Ehat$ give
\eqref{eq:repair-compact-kernel-bound} after adding
\eqref{eq:canonical-higher-kernel-bounds}. If either spin is zero,
the arithmetic term is absent and the higher-order bound suffices.
\end{proof}

The same bound holds for $R$, since its additional term in
\eqref{eq:repair-R-kernel} is at most $12\sqrt{\Ehat\Ehat'}$.
Both kernels are real and symmetric under simultaneous exchange
of the input and output labels: for the arithmetic factor use
$r\mapsto r^{-1}$ in the Kloosterman sum, and for the cosine
terms use their displayed products.

\subsection{Finite signed seeds and thermal convergence}
\label{app:canonical-finite-seeds}

Let $\sigma\in\mathcal M_B$. There are only finitely many input
spins, although the output may involve every integer spin. The
joint bound gives
\begin{equation}
  |(Q\sigma)_J(\Ehat)|
    \le C\log(2+B)\sqrt{\Ehat}
        \sum_{J'}\int\sqrt{\Ehat'}\dd|\sigma_{J'}|(\Ehat')
    \le C\log(2+B)\sqrt{B\Ehat}\,\|\sigma\|_{\TV}.
  \label{eq:canonical-finite-seed-bound}
\end{equation}
The first integral is taken against the total-variation measure
of the input seed. The identical estimate holds with $Q$ replaced
by $R$.

To integrate over output energies, write $j=|J|$ and
$\Ehat=j+s$. For every $\beta>0$,
\begin{align}
  \int_j^\infty e^{-\beta\Ehat}\sqrt{\Ehat}
            \dd\omega_J(\Ehat)
   &=e^{-\beta j}\int_0^\infty
        e^{-\beta s}\frac{\sqrt{j+s}}{\sqrt{s(2j+s)}}\dd s
        \nonumber\\
   &\le e^{-\beta j}\int_0^\infty e^{-\beta s}s^{-1/2}\dd s
    =e^{-\beta j}\sqrt{\frac\pi\beta}.
  \label{eq:canonical-spin-thermal-bound}
\end{align}
This includes $J=0$ and the nonzero-spin edges. Summing over
$J\in\Z$ yields
\begin{equation}
  \sum_J\int_{|J|}^\infty e^{-\beta\Ehat}
           |(Q\sigma)_J(\Ehat)|\dd\omega_J(\Ehat)
   \le C\log(2+B)\sqrt B\,\|\sigma\|_{\TV}
         \sqrt{\frac\pi\beta}
          \left(1+\frac{2}{e^\beta-1}\right).
  \label{eq:canonical-all-spin-thermal-bound}
\end{equation}
The same inequality holds for $R$ after increasing $C$.

These absolute bounds justify integrating
\eqref{eq:canonical-unit-laplace} against $\sigma$ and exchanging
the seed integrals, the output integrals and the spin sum.
The direct term integrates to $\sigma$ itself; the induced
threshold integrates to $t(\sigma)\delta_{(0,0)}$.
This proves \eqref{eq:canonical-spectral-measure} as an equality
of ordinary spectral measures and of their thermal transforms.
The direct measure has finite total variation, while
$|t(\sigma)|\le\max_{|J'|\le B}|c_{J'}|\|\sigma\|_{\TV}$.
Together with \eqref{eq:canonical-all-spin-thermal-bound}, this
gives finite thermal total variation for the complete output.

For scalar input approaching the origin, the estimates remain
uniform in the input energy. In particular, the factor
$\sqrt{\Ehat'}$ in \eqref{eq:canonical-finite-seed-bound} is
bounded on $(0,B]$; no inverse power of the input energy is needed.
The local uniform bound from
\eqref{eq:canonical-energy-subtraction} independently justifies
the integral of the completed modular functions on each compact
subset of $\HH$. Integration therefore preserves modular
invariance and gives the same function as the absolutely
convergent thermal transform above.

Finally, the scalar density added by $R-Q$ is
$12m_g(\sigma)\Ehat^{-1/2}\dd\Ehat$. Its thermal transform is
the modular constant $12m_g(\sigma)/\sqrt2$ by
\eqref{eq:constant-continuum}, and it gives no atom at zero.
This proves the measure formula
\eqref{eq:corrected-completion-measure}. The bounds apply to
every finite correcting seed constructed in Appendix~B.
The additional estimates needed to sum infinitely many repairs
are supplied by Appendix~\ref{app:global-estimates} and used in
Section~\ref{sec:global-limit}.

\section{Positivity, inversion and exact repair}
\label{app:inversion-repair}

This appendix proves Proposition~\ref{prop:controlled-band-inverse},
the finite-measure inverse \eqref{eq:inverse-measure-bound}, and the
anchor bound \eqref{eq:anchor-size-bound}. It then completes the proof
of Proposition~\ref{prop:exact-local-repair} and derives the exterior
estimates used in the recursion. We use the canonical completion and
kernel estimates of Appendix~\ref{app:canonical}. Constants denoted by
$C$ are independent of the cutoff and may increase between estimates;
$C_B$ is finite for each fixed $B>0$. Whenever an exponential bound
$C_B\le e^{CB}$ is asserted, its range is $B\ge2$.

\subsection{Positivity of the quadratic form}
\label{app:band-positivity}

\paragraph{Boundedness on the band.}
For the reference measures in \eqref{eq:band-Hilbert-operator},
\begin{equation}
  \sum_{|J|<B}\int_{|J|}^B\Ehat\dd\omega_J
    =B+2\sum_{1\le J<B}\sqrt{B^2-J^2}
    \le CB^2,\qquad B\ge2.
  \label{eq:inverse-band-integral}
\end{equation}
The scalar integral is $B$ and is finite for every positive cutoff.
The kernel bound \eqref{eq:repair-compact-kernel-bound}, also valid
for $R$, therefore gives
\begin{equation}
  \begin{aligned}
  \|R_B\|_{\mathrm{HS}}^2
    &:=\sum_{|J|,|J'|<B}\int_{|J|}^B\int_{|J'|}^B
       |R_{J,J'}(\Ehat,\Ehat')|^2
                     \dd\omega_{J'}(\Ehat')\dd\omega_J(\Ehat),\\
  \|R_B\|_{\mathrm{HS}}&\le CB^2\log(2+B),\qquad
  \|P_B\|\le CB^3\qquad(B\ge2).
  \end{aligned}
  \label{eq:inverse-Hilbert-Schmidt}
\end{equation}
For $0<B<2$ the same integrals are finite, without the need for
these particular powers of $B$. The kernel is real and symmetric,
so $R_B$ is compact and self-adjoint. In particular, its quadratic
form is continuous on $\mathcal H_B$.

\paragraph{A positive spatial kernel.}
Let $\Gamma=\mathrm{PSL}_2(\Z)$, $z=x+iy$ and $w=x'+iy'$.
Consider the point-pair kernel
\begin{equation}
  k_s(u)=\frac14(1+u)^{-s},\qquad
  u(z,w)=\frac{|z-w|^2}{4yy'},\qquad
  A_s(z,w)=\sum_{\gamma\in\Gamma}k_s(u(z,\gamma w)),
  \quad\operatorname{Re}s>1.
  \label{eq:inverse-spatial-kernel}
\end{equation}
We use the spatial pre-trace theorem with hyperbolic measure
$\dd x\dd y/y^2$, an orthonormal cuspidal basis, and continuous
spectral measure $\dd r/(4\pi)$; see
\cite[Section~2.3]{HuangXu2015}. This theorem and its locally uniform
spectral convergence for rapidly decreasing transforms are the
external spectral-theory input in this subsection. Our point-pair
variable includes the factor $4$ in \eqref{eq:inverse-spatial-kernel},
so we compute the transform in this convention.

Integrating first in $x$ at $z=i$ gives
\begin{align*}
  h_s(r)
  &=4^{s-1}\frac{\sqrt\pi\,\Gamma(s-1/2)}{\Gamma(s)}
    \int_0^\infty y^{s-3/2+ir}(1+y)^{1-2s}\dd y\\
  &=\frac{\pi\,\Gamma(s-1/2+ir)\Gamma(s-1/2-ir)}{\Gamma(s)^2}.
\end{align*}
The second equality is the beta integral followed by gamma
duplication. For $\operatorname{Re}s>1$, this transform is even,
holomorphic in a strip wider than $|\operatorname{Im}r|\le1/2$,
and rapidly decreasing there, as required by the pre-trace theorem.
Write $\operatorname{Eis}(z,s)=\mathcal P_{0,0}(z;s)$
for the Eisenstein series. The pre-trace formula is
\begin{equation}
  \begin{aligned}
  A_s(z,w)={}&\frac{h_s(i/2)}{\operatorname{vol}(\Gamma\backslash\HH)}
     +\sum_\ell h_s(r_\ell)U_\ell(z)\overline{U_\ell(w)}\\
   &+\frac1{4\pi}\int_\R h_s(r)
       \operatorname{Eis}(z,1/2+ir)
       \overline{\operatorname{Eis}(w,1/2+ir)}\dd r,
  \end{aligned}
  \label{eq:inverse-pretrace}
\end{equation}
where $U_\ell$ has Laplace eigenvalue $1/4+r_\ell^2$.

For completeness, no cuspidal or residual eigenvalue in the
complementary range contributes to this formula. To see the
cuspidal assertion directly, let $\mathcal F$ be the standard
fundamental domain and put
$\mathcal S=\{|x|\le1/2,\ y\ge\sqrt3/2\}$.
Inversion maps $\mathcal S\setminus\mathcal F$ injectively into
$\mathcal F$. Conformal invariance of the Dirichlet integral
therefore bounds the energy on $\mathcal S$ by twice that on
$\mathcal F$. A cusp form $U$ has zero horizontal mean, so the
periodic Poincare inequality gives
\[
  2\lambda\|U\|_{L^2(\mathcal F)}^2
    \ge\int_{\mathcal S}|\partial_xU|^2\dd x\dd y
    \ge4\pi^2\int_{\mathcal S}|U|^2\dd x\dd y
    \ge3\pi^2\|U\|_{L^2(\mathcal F)}^2.
\]
Thus $\lambda\ge3\pi^2/2>1/4$. The scattering coefficient
$\varphi(s)=\zeta^*(2s-1)/\zeta^*(2s)$, with
$\zeta^*(s)=\pi^{-s/2}\Gamma(s/2)\zeta(s)$, has no pole for
$1/2<s<1$; the only residual term is the constant function.

Continuing \eqref{eq:inverse-pretrace} meromorphically in $s$,
the transform and the constant contribution at $s=1/2$ are
\begin{equation}
  h_{1/2}(r)=\frac{\pi}{r\sinh(\pi r)}>0\quad(r\in\R\setminus\{0\}),
  \qquad
  \frac{h_{1/2}(i/2)}{\pi/3}=-6.
  \label{eq:inverse-positive-transform}
\end{equation}
The continuous integral has a finite limit at $r=0$.
Indeed, $\operatorname{Eis}(z,1/2)=0$ and hence
$\operatorname{Eis}(z,1/2+ir)=O(r)$ locally uniformly in $z$.
For $s=1/2+\varepsilon$ the product of its two factors cancels
the possible bound $C/(\varepsilon^2+r^2)$ from $h_s(r)$.
At large $|r|$, Stirling's formula gives exponential decay,
which dominates the local spectral bounds in the pre-trace theorem.
The continued formula consequently shows that $A_{1/2}(z,w)+6$
is a positive kernel: testing it against finitely many spatial
points, or compactly supported spatial functions, gives a sum
and an integral of absolute squares with nonnegative weights.

\paragraph{Fourier--Laplace testing and the scalar terms.}
We transfer this spatial positivity to the energy kernel using
ordinary convergent integrals. Periodization in $x'$ followed by
Poisson summation gives, initially for $\operatorname{Re}s>1$,
\begin{equation}
  \begin{aligned}
  A_s(z,w)
    &=C_s(y')^s\sum_{J'\in\Z}e^{-2\pi iJ'x'}
      \int_{|J'|}^\infty
        (\Ehat'^2-J'^2)^{s-1}e^{-2\pi y'\Ehat'}
        \mathcal P_{\Ehat',J'}(z;s)\dd\Ehat',\\
  C_s&=\frac{2^{2s-1}\pi^{2s}}{\Gamma(s)^2},\qquad C_{1/2}=1.
  \end{aligned}
  \label{eq:inverse-pair-Poisson}
\end{equation}
The zero horizontal Fourier coefficient before the coset sum is
$4^{s-1}\sqrt\pi\,\Gamma(s-1/2)(yy')^s
(y+y')^{1-2s}/\Gamma(s)$; the gamma integral on the right gives
exactly this coefficient. The nonzero coefficients follow from the
same Fourier integral using the $K$-Bessel representation.

Fix the input Fourier index before continuing
\eqref{eq:inverse-pair-Poisson}. When $J'\ne0$, the lower-end
factor has exponent $s-1$ and remains integrable near $s=1/2$.
The exponential controls large $\Ehat'$: the seed subtraction
\eqref{eq:canonical-energy-subtraction} bounds the continued
Poincare function by $C(1+\Ehat')$ on compact spatial sets.
When $J'=0$, split off $\operatorname{Eis}(z,s)$ instead.
The difference is $O(\Ehat')$ at zero, and the part split off
integrates to
\[
  \operatorname{Eis}(z,s)\Gamma(2s-1)(2\pi y')^{1-2s}.
\]
Its limit is $\tfrac12\partial_s\operatorname{Eis}(z,1/2)$.
Thus the scalar input Fourier coefficient of
$A_{1/2}(z,w)/\sqrt{y'}$ is
\begin{equation}
  \frac12\partial_s\operatorname{Eis}(z,1/2)
    +\int_0^\infty e^{-2\pi y'\Ehat'}
       \mathcal P_{\Ehat',0}(z)\frac{\dd\Ehat'}{\Ehat'}.
  \label{eq:inverse-scalar-axis}
\end{equation}
The first term must be retained until the input variable is tested.

Choose a finite set of spins and finite exponential sums
\begin{equation}
  p_J(\Ehat)=\sum_\ell\alpha_{J\ell}e^{-2\pi y_\ell\Ehat},
  \qquad y_\ell>0,\qquad
  p_0(0)=\sum_\ell\alpha_{0\ell}=0.
  \label{eq:inverse-exponential-tests}
\end{equation}
Define the Fourier matrix of the positive spatial kernel by
\[
  \mathcal B_{J,J'}(y,y')=\frac1{\sqrt{yy'}}
    \int_0^1\!\int_0^1 e^{-2\pi iJx}e^{2\pi iJ'x'}
       [A_{1/2}(x+iy,x'+iy')+6]\dd x\dd x'.
\]
Testing \eqref{eq:inverse-pair-Poisson} first in the input variable
and then using \eqref{eq:canonical-unit-laplace} in the output
variable gives
\begin{equation}
  \begin{aligned}
  0\le{}&\sum_{J,J',\ell,m}
     \overline{\alpha_{J\ell}}\alpha_{J'm}
       \mathcal B_{J,J'}(y_\ell,y_m)\\
   ={}&\sum_J\int_{|J|}^\infty|p_J(\Ehat)|^2\dd\omega_J(\Ehat)\\
    &+\sum_{J,J'}\int_{|J|}^\infty\int_{|J'|}^\infty
       \overline{p_J(\Ehat)}R_{J,J'}(\Ehat,\Ehat')p_{J'}(\Ehat')
                  \dd\omega_{J'}(\Ehat')\dd\omega_J(\Ehat).
  \end{aligned}
  \label{eq:inverse-tested-identity}
\end{equation}
Here the input condition $p_0(0)=0$ removes the first term in
\eqref{eq:inverse-scalar-axis}. The same condition on the output
test removes the scalar threshold in the unit-seed measure.
The direct seed gives the single integral in
\eqref{eq:inverse-tested-identity}, and its induced continuum
gives the double integral with kernel $Q$.
Finally, the spatial constant $6$ gives the rank-one term with
coefficient $12$, since
\[
  \int_0^\infty\sqrt{\Ehat}\,e^{-2\pi y\Ehat}\dd\omega_0
     =\frac1{\sqrt{2y}}.
\]
This accounts for all terms in $R=Q+12g\otimes g$.

Every integral in this argument is absolutely convergent.
For a fixed finite spin set $S$, Appendix~\ref{app:canonical-arithmetic}
gives $|R_{J,J'}(\Ehat,\Ehat')|\le C_S\sqrt{\Ehat\Ehat'}$
on the full physical energy ranges. The higher-order estimate is
global, while the finitely many arithmetic constants are absorbed
using $\Ehat\Ehat'\ge|JJ'|\ge1$ when they occur.
The scalar test vanishes linearly at zero, and all tests decay
exponentially at infinity. Thus the tested identity uses thermal
integration of the canonical measure, without taking termwise
inverse transforms of the spatial spectral expansion.

\paragraph{Passage to the band Hilbert space.}
For fixed $S$, the quadratic form on the right of
\eqref{eq:inverse-tested-identity} is continuous in the norm
\begin{equation}
  \|f\|_{X_S}
   =\left(\sum_{J\in S}\int_{|J|}^\infty|f_J|^2\dd\omega_J\right)^{1/2}
      +\sum_{J\in S}\int_{|J|}^\infty\sqrt{\Ehat}|f_J|\dd\omega_J.
  \label{eq:inverse-test-norm}
\end{equation}
Take a smooth function compactly supported in each open band row
and extend it by zero to the full energy range. In the scalar row,
divide it by $e^{-\Ehat}(1-e^{-\Ehat})$ and regard the quotient
as a continuous function of $e^{-\Ehat}\in[0,1]$, extended by
zero at both ends. Uniform polynomial approximation and
multiplication by $e^{-\Ehat}(1-e^{-\Ehat})$ give exponential
sums satisfying $p_0(0)=0$. Their errors are bounded by
$\epsilon e^{-\Ehat}(1-e^{-\Ehat})$, which tends to zero in
both terms of \eqref{eq:inverse-test-norm}. For a nonzero-spin
row, divide by $e^{-\Ehat}$ instead. The resulting error
$\epsilon e^{-\Ehat}$ is integrable in the same norms, since
the nonzero-spin edge of $\dd\omega_J$ is integrable.

Consequently \eqref{eq:inverse-tested-identity} implies
$\langle f,P_Bf\rangle\ge0$ on this smooth dense core.
The approximating exponential sums need not be supported in the
band; convergence in $X_S$ is taken before the limiting compact
test is restricted. Boundedness of $P_B$ then extends the result
to every $f\in\mathcal H_B$, for every $B>0$.
The vectors of cuspidal and continuous spectral factors in
\eqref{eq:inverse-tested-identity} are Cauchy under this approximation:
their squared distance is the same quadratic form evaluated on
the difference of the tests. Taking their Hilbert-space limit
gives the sum-and-integral-of-squares representation of the band
quadratic form used in Section~\ref{sec:band-inverse}.

\paragraph{Zero modular integral.}
The same constant correction also removes the modular average of
the completed seed. Let $\mathcal F$ be the standard fundamental
domain and $\dd\mu=\dd x\dd y/y^2$, so that
$\operatorname{vol}(\mathcal F)=\pi/3$.
For $\Ehat'>0$ and $\operatorname{Re}s>1$, unfolding gives
\[
  \int_{\mathcal F}\mathcal P_{\Ehat',J'}(\tau;s)\dd\mu
    =\delta_{J',0}\int_0^\infty y^{s-2}e^{-2\pi\Ehat'y}\dd y
    =\delta_{J',0}\Gamma(s-1)(2\pi\Ehat')^{1-s}.
\]
Near $s=1/2$, the continued cusp expansion is
$O(y^{1-\operatorname{Re}s})$ apart from the exponentially decaying
direct seed. The modular integral therefore converges locally uniformly
in $s$ there, and meromorphic continuation yields
\[
  \frac3\pi\int_{\mathcal F}
       \mathcal P_{\Ehat',J'}(\tau)\dd\mu
    =-\frac{12}{\sqrt2}\,\delta_{J',0}\sqrt{\Ehat'}.
\]
The compact-input thermal bounds of Appendix~\ref{app:canonical}
give $|\mathcal P_{\Ehat',J'}(\tau)|\le C_B\sqrt y$ on
$\mathcal F$, uniformly for seeds in $\mathcal B_B$.
Since $\int_{\mathcal F}\sqrt y\,\dd\mu<\infty$, we may integrate
this identity against any $\sigma\in\mathcal M_B$.
The added constant in \eqref{eq:corrected-completion} then gives
\[
  \frac3\pi\int_{\mathcal F}\mathscr C[\sigma](\tau)\dd\mu
    =-\frac{12}{\sqrt2}\,m_g(\sigma)
       +\frac{12}{\sqrt2}\,m_g(\sigma)=0.
\]

\subsection{Strict positivity and the Hilbert-space inverse}
\label{app:strict-band-inverse}

The preceding nonnegativity also holds on larger bands. We use it
to control an analytic response on $[2B,3B]$, and then propagate
that control to the original band. The following elementary
estimate keeps the dependence on $B$ quantitative.

\begin{lemma}[Propagation from a separated interval]
\label{lem:inverse-analytic-propagation}
If $F$ is holomorphic on $|z|<1000$ and $|F(z)|\le M$ there, then
\begin{equation}
  \sup_{0\le t\le1}|F(t)|
    \le14M^{2/3}
       \left(\int_{\sqrt2}^{\sqrt3}|F(t)|^2\dd t\right)^{1/6}.
  \label{eq:inverse-analytic-propagation}
\end{equation}
\end{lemma}
\begin{proof}
Put $I=[\sqrt2,\sqrt3]$ and $\delta=\|F\|_{L^2(I)}$.
The affine map from $I$ to $[-1,1]$ sends $[0,1]$ into
$[-10,10]$. The Legendre recurrence bounds the degree-$n$
Legendre polynomial there by $21^n$. Expanding a degree-$N$
polynomial in the orthonormal Legendre basis on $I$ and applying
Cauchy--Schwarz yields
\[
  \|p\|_{L^\infty([0,1])}\le2\cdot42^N\|p\|_{L^2(I)}.
\]
The Taylor polynomial $p_N$ of $F$ at zero has error at most
$2M500^{-(N+1)}$ on $|z|\le2$, by Cauchy's inequality.
Combining the two bounds gives
\[
  \|F\|_{L^\infty([0,1])}
     \le2\cdot42^N\delta+M(42/500)^N.
\]
For $0<\delta<M$, choose
$N=\lfloor\log(M/\delta)/\log500\rfloor$. With
$\vartheta=\log(500/42)/\log500>1/3$, the right side is at most
$(2+500/42)M(\delta/M)^\vartheta
\le14M^{2/3}\delta^{1/3}$.
If $\delta=0$, use the identity theorem; if $\delta\ge M$, the
claimed bound follows immediately.
\end{proof}

\paragraph{Exponential lower bound.}
Let $B\ge2$, $\|f\|_{\mathcal H_B}=1$, and
$\delta_f=\langle f,P_Bf\rangle\ge0$.
For each $|J|<B$ define the entire function
\begin{equation}
  F_J(z)=\sum_{|J'|<B}\int_{|J'|}^B
     Q_{J,J'}(Bz^2,\Ehat')f_{J'}(\Ehat')\dd\omega_{J'}
       +12\delta_{J,0}\sqrt B\,z\,\langle g,f\rangle_{\mathcal H_B}.
  \label{eq:inverse-analytic-response}
\end{equation}
The scalar function satisfies $F_0(0)=0$. Set $H_0(z)=F_0(z)/z$
with its removable value at zero, and $H_J=F_J$ for $J\ne0$.
On any fixed complex disk, the cosine power series in
\eqref{eq:repair-Q-kernel} give
$|Q_{J,J'}(Bz^2,\Ehat')|\le CB\Ehat'e^{CB}$.
Indeed, each full bracket is bounded by
$CB\Ehat'm^{-2}e^{CB/m}$; the arithmetic term has the weaker
but sufficient bound $C|JJ'|\le CB\Ehat'$.
Cauchy--Schwarz and \eqref{eq:inverse-band-integral} imply
\[
  \sum_{|J'|<B}\int_{|J'|}^B\Ehat'|f_{J'}|\dd\omega_{J'}
     \le CB^{3/2},\qquad
  |\langle g,f\rangle_{\mathcal H_B}|\le\sqrt B.
\]
Cauchy's inequality on a slightly larger disk controls the
division by $z$ in $H_0$. Absorbing polynomial factors of $B$,
we obtain $\sup_{|z|<1000}|H_J(z)|\le e^{C_1B}$ uniformly in
the old spin rows.

Extend $f$ by zero to $\mathcal H_{3B}$. Positivity and
\eqref{eq:inverse-Hilbert-Schmidt} give
\begin{equation}
  \|P_{3B}f\|_{\mathcal H_{3B}}^2
     \le\|P_{3B}\|\langle f,P_{3B}f\rangle
     \le CB^3\delta_f.
  \label{eq:inverse-exterior-L2}
\end{equation}
On $2B<\Ehat<3B$ in every old row the direct function vanishes,
so this is a bound on $F_J(\sqrt{\Ehat/B})$.
Under $\Ehat=Bt^2$, one has
$\dd\omega_J\ge2\dd t/\sqrt3$ for $t\in[\sqrt2,\sqrt3]$.
It follows that
\[
  \sum_{|J|<B}\|H_J\|_{L^2([\sqrt2,\sqrt3])}^2
       \le CB^3\delta_f.
\]
Apply Lemma~\ref{lem:inverse-analytic-propagation} to these
functions. For nonzero spin the reference measure of the low
row is at most $\log(2B)$. For scalar spin, use instead
\[
  \int_0^B|F_0(\sqrt{\Ehat/B})|^2\frac{\dd\Ehat}{\Ehat}
      =2\int_0^1t|H_0(t)|^2\dd t.
\]
There are at most $2B+1$ rows, and hence
\begin{equation}
  \|R_Bf\|_{\mathcal H_B}\le e^{C_2B}\delta_f^{1/6},\qquad
  \|P_Bf\|_{\mathcal H_B}\le CB^{3/2}\delta_f^{1/2}.
  \label{eq:inverse-propagated-bound}
\end{equation}
Since $f=P_Bf-R_Bf$ has unit norm, these inequalities force
$\delta_f\ge e^{-CB}$ after increasing $C$.
The spectral theorem for bounded self-adjoint operators now
proves both assertions of \eqref{eq:band-coercivity} for $B\ge2$.

\paragraph{Small cutoffs and an inverse series.}
For $0<B<2$, let $E_B:\mathcal H_B\to\mathcal H_2$ be extension
by zero. The kernels agree on their common band, so
\[
  P_B=E_B^*P_2E_B,\qquad
  \langle f,P_Bf\rangle
      =\langle E_Bf,P_2E_Bf\rangle
      \ge\|P_2^{-1}\|^{-1}\|f\|_{\mathcal H_B}^2.
\]
This proves invertibility of the actual $P_B$ and completes
Proposition~\ref{prop:controlled-band-inverse}.
Only the quadratic form is compressed in this argument.

For every $B>0$ an explicit admissible scale is
\begin{equation}
  \lambda_B=2+\|R_B\|_{\mathrm{HS}}>
            1+\|R_B\|\ge\|P_B\|.
  \label{eq:inverse-series-scale}
\end{equation}
If $c_B>0$ denotes the established lower bound for $P_B$, its
spectrum lies in $[c_B,\|P_B\|]$, and therefore
$\|I-P_B/\lambda_B\|\le1-c_B/\lambda_B<1$.
The geometric series gives \eqref{eq:band-inverse-series} in
operator norm on $\mathcal H_B$.

\subsection{The inverse on finite measures}
\label{app:finite-measure-inverse}

We now pass from Hilbert-space functions to finite signed measures,
including atoms at band and spin edges. This requires an integrability
argument in addition to the Hilbert inverse. For $B\ge2$, the kernel
bound gives
\begin{equation}
  \begin{aligned}
  \|(R\nu)_{\le B}\|_{\mathcal H_B}
    &\le CB^{3/2}\log(2+B)\|\nu\|_{\TV},\\
  \|\mathscr R_B\nu\|_{\TV}
    &\le CB^2\log(2+B)\|\nu\|_{\TV},\\
  \|(R_Bf)\dd\omega\|_{\TV}
    &\le CB^{5/2}\log(2+B)\|f\|_{\mathcal H_B}.
  \end{aligned}
  \label{eq:inverse-measure-mapping}
\end{equation}
To verify these estimates, use
$\sum_{J'}\int\sqrt{\Ehat'}\dd|\nu_{J'}|
\le\sqrt B\|\nu\|_{\TV}$ and
\eqref{eq:inverse-band-integral}. The remaining output integral
satisfies
\begin{equation}
  \sum_{|J|<B}\int_{|J|}^B\sqrt{\Ehat}\dd\omega_J
       \le CB^{3/2}\qquad(B\ge2).
  \label{eq:inverse-band-L1-integral}
\end{equation}
Indeed, its scalar contribution is $2\sqrt B$, and setting
$\Ehat=|J|+s$ bounds each other row by
$\int_0^{B-|J|}s^{-1/2}\dd s\le2\sqrt B$.
There are $O(B)$ rows. For $0<B<2$, extending the input and
enlarging only these estimating integrals to the band $2$ gives
finite bounds of the same mapping types.

Given $\nu\in\mathcal M_B$, define
$u=P_B^{-1}(R\nu)_{\le B}\in\mathcal H_B$.
Its Hilbert-space equation yields
\begin{equation}
  u\dd\omega=\mathscr R_B\nu-(R_Bu)\dd\omega,\qquad
  \|u\dd\omega\|_{\TV}\le C_B\|\nu\|_{\TV}.
  \label{eq:inverse-integrable-density}
\end{equation}
Both terms on the right have finite total variation by
\eqref{eq:inverse-measure-mapping}. For $B\ge2$, the exponential
Hilbert inverse bound absorbs all its polynomial factors, so
$C_B\le e^{CB}$. This proves integrability of $u\dd\omega$
even in the scalar row, where square integrability of $u$ alone
would not suffice.

Set $w=\nu-u\dd\omega$. Substitution gives
$w+\mathscr R_Bw=\nu$, proving that
$\mathcal I_B[\nu]=w$ is a right inverse on $\mathcal M_B$
with the bound \eqref{eq:inverse-measure-bound}.
If $w+\mathscr R_Bw=0$, the first mapping estimate implies
$w=f\dd\omega$ for some $f\in\mathcal H_B$. Its equation is
$P_Bf=0$, so $w=0$. This proves uniqueness and hence the
two-sided inverse identity on finite measures.

In each row the subtracted measure has an integrable density
$u_J(\Ehat)/\sqrt{\Ehat^2-J^2}$ with respect to $\dd\Ehat$.
Thus the inverse preserves the entire part of $\nu$ supported
on sets of zero Lebesgue measure, including all its atoms.
This also covers atoms at $\Ehat=B$ and at $\Ehat=|J|>0$.
If $|J|=B$ is integral, that row has no open Hilbert interval;
its input atom still contributes to $R\nu$ in the other rows
and remains in the direct measure $w$.
The origin atom is excluded from $\mathcal M_B$ and is handled
by the anchor below.

For later estimates, the threshold functional is bounded on
$\mathcal M_B$:
\begin{equation}
  |t(\nu)|\le\max_{|J|\le B}|c_J|\,\|\nu\|_{\TV},\qquad
  \max_{|J|\le B}|c_J|\le4\sqrt B\quad(B\ge2).
  \label{eq:inverse-threshold-functional}
\end{equation}
The second bound follows from $d(n)\le2\sqrt n$ and $c_0=-1$.
Only boundedness on the finite-measure space is used here;
$t$ need not be a bounded functional on all of $\mathcal H_B$.

\subsection{The threshold anchor}
\label{app:threshold-anchor}

\paragraph{Analyticity of the residual threshold.}
For a complex variable $z$, let
\begin{equation}
  \eta_B(z)=
    [R_{\cdot,0}(\cdot,Bz^2)\dd\omega]_{\le B},\qquad
  G_B(z)=-1-t\bigl(\mathcal I_B[\eta_B(z)]\bigr).
  \label{eq:inverse-anchor-analytic-function}
\end{equation}
In the rank-one term we take $\sqrt{Bz^2}=\sqrt B\,z$.
The spaces and linear maps are extended to complex measures in
this formula. For real $z>0$, it agrees with
$G_B(z)=\zeta_B(Bz^2,0)$ from \eqref{eq:anchor-response}.

The map $\eta_B$ is entire in the total-variation norm.
The scalar-input arithmetic term is zero. On each fixed complex
disk the higher-order kernel is bounded by $C_{B}\Ehat$,
and the rank-one term by $C_B\sqrt{\Ehat}$.
These bounds are integrable on every output row, including
$\dd\omega_0=\dd\Ehat/\Ehat$. The cosine series and its
derivatives converge uniformly on smaller disks under these
integrable bounds. Hence they define an entire finite-measure
map. Its value at zero is the zero measure.
Boundedness of $\mathcal I_B$ and $t$ then proves that $G_B$
is entire and
\begin{equation}
  G_B(0)=-1,\qquad
  \sup_{|z|<1000}|G_B(z)|\le e^{CB}\quad(B\ge2).
  \label{eq:inverse-anchor-disk-bound}
\end{equation}
For the second estimate, the same cosine bounds used in
\eqref{eq:inverse-analytic-response} give
$|Q_{J,0}(\Ehat,Bz^2)|\le CB\Ehat e^{CB}$.
Its low-band total variation is polynomial in $B$ times $e^{CB}$;
the rank-one term obeys the same bound. Apply
\eqref{eq:inverse-measure-bound} and
\eqref{eq:inverse-threshold-functional} to obtain
\eqref{eq:inverse-anchor-disk-bound}.
At $z=0$ the inverse is applied only to the zero measure, not
to an origin atom.

\paragraph{Normalization on the actual remote interval.}
Under $\Ehat'=Bt^2$, the normalization in
\eqref{eq:anchor-normalization-integral} is
\begin{equation}
  N_B=2\int_{\sqrt2}^{\sqrt3}G_B(t)^2\frac{\dd t}{t}>0,
  \qquad N_B\ge e^{-CB}\quad(B\ge2).
  \label{eq:inverse-anchor-lower-bound}
\end{equation}
The strict inequality for every $B>0$ follows from the identity
theorem and $G_B(0)=-1$, since $G_B$ is real on the positive
real axis. The quantitative lower bound follows by applying
Lemma~\ref{lem:inverse-analytic-propagation} to $G_B$ and using
\eqref{eq:inverse-anchor-disk-bound}. This uses the interval
$[2B,3B]$ of the actual cutoff, also when $B<2$.

Define $\psi_B$, $v_B$ and $\theta_B$ by
\eqref{eq:anchor-seed-construction}. The normalized seed satisfies
\begin{equation}
  \|\psi_B\|_{\TV}
       \le\left(\frac{\log(3/2)}{N_B}\right)^{1/2},\qquad
  \|\rho_B^\psi\|_\infty
       \le\frac{\sup_{\sqrt2\le t\le\sqrt3}|G_B(t)|}{N_B}.
  \label{eq:inverse-anchor-seed-bound}
\end{equation}
The first inequality is Cauchy--Schwarz. Kernel integration
against this bounded-support seed gives a finite measure
$\mathscr R_B\psi_B$. The inverse bound therefore controls
$v_B=\mathcal I_B[\mathscr R_B\psi_B]$.
Equations \eqref{eq:inverse-anchor-disk-bound}--\eqref{eq:inverse-anchor-seed-bound}
bound both measures and $\rho_B^\psi$ by $e^{CB}$ for $B\ge2$.
They also give finite constants for every fixed $B>0$.
Since $\theta_B$ has energies at most $3B$,
\[
  \sum_{J'}\int\sqrt{\Ehat'}\dd|\theta_{B,J'}|
       \le\sqrt{3B}\|\theta_B\|_{\TV}.
\]
This proves all terms of \eqref{eq:anchor-size-bound}, with
the stated fixed-$B$ and exponential large-$B$ dependence.

\paragraph{The exact anchor identities.}
The inverse equation for $v_B$ gives
$v_B+\mathscr R_Bv_B=\mathscr R_B\psi_B$ and hence
\eqref{eq:anchor-band-cancellation}.
To compute the threshold, the integrable kernel bounds allow
integration of the finite measures in \eqref{eq:anchor-response}
against $\psi_B$. Applying the bounded maps $\mathcal I_B$ and
$t$ to this integral proves \eqref{eq:anchor-response-identity}.
The chosen density of $\psi_B$ then gives
\[
  t(\theta_B)=\frac1{N_B}\int_{2B}^{3B}
       \zeta_B(\Ehat',0)^2\frac{\dd\Ehat'}{\Ehat'}=1.
\]
Both $\psi_B$ and $\mathscr R_B\psi_B$ have integrable energy
densities, and the finite-measure inverse preserves this property.
Consequently $\theta_B$ has no atoms. Its completion
$\mathcal A_B=\mathscr C[\theta_B]$ has a unit scalar threshold,
zero positive-energy measure through $B$, and only energy
densities above $B$. In particular, the density above $B$
includes the direct remote seed as well as the induced response.

\subsection{Complete repair and exterior estimates}
\label{app:complete-repair-estimates}

\paragraph{Exact band data and thermal convergence.}
For $\nu\in\mathcal M_B$ let $w=\mathcal I_B[\nu]$ and
$\sigma_B[\nu]=w-t(w)\theta_B$.
The inverse and anchor identities give
\[
  [\sigma_B[\nu]+(R\sigma_B[\nu])\dd\omega]_{\le B}=\nu,
  \qquad t(\sigma_B[\nu])=0.
\]
The finite-seed completion is modular invariant by
Appendix~\ref{app:canonical-finite-seeds} and has only uncensored
support. The inverse subtracts an energy density and the anchor
seed has an energy density, so the entire singular part of the
complete repair is that of $\nu$. In particular, the repair
has no atoms above $B$. This proves
Proposition~\ref{prop:exact-local-repair}.

For the exterior bound, \eqref{eq:inverse-threshold-functional}
and \eqref{eq:anchor-size-bound} imply
\begin{equation}
  \sum_{J'}\int\sqrt{\Ehat'}\dd|\sigma_{B,J'}[\nu]|
       +|t(w)|\|\rho_B^\psi\|_\infty
       \le C_B\|\nu\|_{\TV}.
  \label{eq:inverse-complete-seed-bound}
\end{equation}
Apply \eqref{eq:repair-compact-kernel-bound} with input cutoff
$3B$ to the induced part of
\eqref{eq:complete-repair-exterior}. The direct term is bounded
by the second term of \eqref{eq:inverse-complete-seed-bound}.
Since $1\le B^{-1/2}\sqrt{\Ehat}$ for $\Ehat>B$, both satisfy
\eqref{eq:general-exterior-bound}, with $C_B\le e^{CB}$ for
$B\ge2$.
The all-spin integral from Appendix~\ref{app:canonical-finite-seeds}
then gives the more explicit thermal estimate
\begin{equation}
  \sum_J\int e^{-\beta\Ehat}
      \dd|\mu_{\Repair_B[\nu],J}|(\Ehat)
    \le\|\nu\|_{\TV}
       \left[1+C_B\sqrt{\frac\pi\beta}
         \left(1+\frac2{e^\beta-1}\right)\right],\qquad\beta>0.
  \label{eq:inverse-complete-thermal-bound}
\end{equation}
This proves \eqref{eq:repair-thermal-convergence} for each repair.

\paragraph{Cancellation in the full correcting measure.}
Suppose that $\nu$ is supported in one spin row $J'$ on $(L,V]$,
where $0<L<V<B$ and $|J'|\le L$, and satisfies
\eqref{eq:repair-zero-moments}. Its zeroth moment gives
$t(\nu)=0$, and its first moment removes the scalar rank-one
term in $R\nu$.
Set $v_\nu=\mathcal I_B[\mathscr R_B\nu]$.
The inverse equation gives $w=\nu-v_\nu$, and the full
correcting seed is \eqref{eq:moment-repair-decomposition}.
Above $B$ the resulting density is therefore
\begin{equation}
  \begin{aligned}
  r_{B,J}[\nu](\Ehat)={}&(R\nu)_J(\Ehat)-(Rv_\nu)_J(\Ehat)\\
     &+t(v_\nu)(R\theta_B)_J(\Ehat)
       +t(v_\nu)\delta_{J,0}\rho_B^\psi(\Ehat).
  \end{aligned}
  \label{eq:inverse-moment-full-response}
\end{equation}
The bounds already proved imply
\[
  \|v_\nu\|_{\TV}+|t(v_\nu)|
       \le C_B\|\mathscr R_B\nu\|_{\TV}.
\]
The last three terms in \eqref{eq:inverse-moment-full-response}
are thus bounded by
$C_B\sqrt{\Ehat}\|\mathscr R_B\nu\|_{\TV}$, including the
direct remote density. This proves
\eqref{eq:moment-exterior-bound} with $C_B\le e^{CB}$ for
$B\ge2$.

\paragraph{Polynomial approximation.}
Use the full moment interval
\[
  [u_-,u_+]=[\sqrt{L-|J'|},\sqrt{V-|J'|}],\qquad
  d=u_+-u_->0.
\]
For $\varpi\ge2$, its Bernstein ellipse has
boundary
\begin{equation}
  z=\frac{u_-+u_+}{2}+\frac d4(\xi+\xi^{-1}),\qquad
  |\xi|=\varpi.
  \label{eq:inverse-Bernstein-ellipse}
\end{equation}
Write $H_\varpi=d(\varpi-\varpi^{-1})/4$ for its height and
$W_\varpi=\max(|z|^2+2|J'|)$ on this ellipse.
For complex $z$ and real $A\ge0$,
\begin{equation}
  |\operatorname{Im}\sqrt{z^2+A}|\le|\operatorname{Im}z|,
  \qquad |\sqrt{z^2+A}|^2\le|z|^2+A.
  \label{eq:inverse-square-root-bounds}
\end{equation}
The second inequality is the triangle inequality. For the first,
write $z=p+it$ and use the formula for the squared imaginary
part of a square root; the required comparison after squaring
has nonnegative difference $4At^2$.
The sign of the root is immaterial.

After substituting $\Ehat'=|J'|+z^2$, the two input factors
$\Ehat'\pm J'$ are $z^2$ and $z^2+2|J'|$.
For a physical output $\Ehat\ge|J|$, the sum of the imaginary
parts of the two cosine arguments at denominator $m$ is bounded
in absolute value by $4\pi H_\varpi\sqrt{\Ehat}/m$; here
$\sqrt{\Ehat+J}+\sqrt{\Ehat-J}\le2\sqrt{\Ehat}$.
Their squared moduli sum to at most
$CW_\varpi\Ehat/m^2$.
Bounding the intact product-minus-one bracket and using
$|S(J,J';m)|/m\le1$, its sum is bounded by
\begin{equation}
  C\sum_{m\ge1}\min\left(1,\frac{W_\varpi\Ehat}{m^2}\right)
       e^{4\pi H_\varpi\sqrt{\Ehat}/m}
    \le C\sqrt{W_\varpi\Ehat}\,
       e^{4\pi H_\varpi\sqrt{\Ehat}}.
  \label{eq:inverse-complex-kernel-bound}
\end{equation}
The arithmetic constant cancels by the zeroth moment. The
rank-one term is either zero or linear in $u$, so it is removed
by the same moment conditions. Only the higher-order kernel
bounded in \eqref{eq:inverse-complex-kernel-bound} requires
approximation.

For a function bounded by $M$ on a Bernstein ellipse, Cauchy's
estimate after the substitution in
\eqref{eq:inverse-Bernstein-ellipse} bounds its degree-$r$
Chebyshev coefficient by $2M\varpi^{-r}$. Truncating at degree
$k$ therefore gives a uniform remainder at most
$2M\varpi^{-k}/(\varpi-1)$ on $[u_-,u_+]$.
Subtracting this polynomial inside the integral against $\nu$
as in \eqref{eq:kernel-moment-approximation} yields
\begin{equation}
  |(R\nu)_J(\Ehat)|
     \le C\|\nu\|_{\TV}\sqrt{W_\varpi}\,\varpi^{-k}
             \sqrt{\Ehat}\,e^{4\pi H_\varpi\sqrt{\Ehat}}.
  \label{eq:inverse-raw-moment-bound}
\end{equation}
This estimate is uniform in the output spin and applies both
inside and outside the band. It concerns the entire original
cell $(L,V]$; its closed hull is used only for polynomial
approximation. No exclusion of a left portion of the signed
input measure is made.

Integrate \eqref{eq:inverse-raw-moment-bound} over the low band
using \eqref{eq:inverse-band-L1-integral}, and then apply
\eqref{eq:moment-exterior-bound}. Absorbing the finite band
integral into $C_B$ gives
\begin{equation}
  |r_{B,J}[\nu](\Ehat)|
   \le C_B\|\nu\|_{\TV}\sqrt{W_\varpi}\,\varpi^{-k}
      e^{4\pi H_\varpi\sqrt B}
      \sqrt{\Ehat}\,e^{4\pi H_\varpi\sqrt{\Ehat}},
  \quad\Ehat>B,\quad\Ehat>|J|.
  \label{eq:inverse-complete-moment-bound}
\end{equation}
Here $C_B$ is independent of the cell, $k$, $\varpi$, and the
output variables, and obeys $C_B\le e^{CB}$ for $B\ge2$.
Thus the inverse and anchor costs and the ellipse growth are
kept separate in this bound.

\paragraph{The initial and tail estimates used in Appendix D.}
For an initial cell with $V\le\tau_0a$, choose
$\varpi=1+\varepsilon_{\mathrm{ell}}\sqrt a/d$ and
$\tau_0\le\varepsilon_{\mathrm{ell}}^2$.
Then $\varpi\ge2$,
$H_\varpi\le\varepsilon_{\mathrm{ell}}\sqrt a/2$, and
$W_\varpi\le C_{\varepsilon_{\mathrm{ell}}}a$.
With $\gamma=2\pi\varepsilon_{\mathrm{ell}}$,
\eqref{eq:inverse-complete-moment-bound} gives, for $B\ge2$,
\begin{equation}
  |r_{B,J}[\nu](\Ehat)|
    \le C\|\nu\|_{\TV}\sqrt a\,\varpi^{-k}
       e^{CB+\gamma\sqrt{aB}}
       \sqrt{\Ehat}\,e^{\gamma\sqrt{a\Ehat}}.
  \label{eq:inverse-initial-moment-bound}
\end{equation}
Here and in this specialization $C$ may depend on the fixed
$\varepsilon_{\mathrm{ell}}$.
For a tail cell with $d\le1$, choose $\varpi=2$.
Then $H_\varpi\le3/8$ and
$W_\varpi\le C(V+1)\le CB$ for $B\ge2$ and $V<B$.
Absorbing $e^{C\sqrt B}$ into $e^{CB}$ gives
\begin{equation}
  |r_{B,J}[\nu](\Ehat)|
    \le C\|\nu\|_{\TV}\sqrt B\,2^{-k}e^{CB}
           \sqrt{\Ehat}\,e^{C_1\sqrt{\Ehat}}.
  \label{eq:inverse-tail-moment-bound}
\end{equation}
These are the local estimates used in
\eqref{eq:global-long-completion} and
\eqref{eq:global-short-completion}. They hold for any finite
signed cell measure with the indicated vanishing moments.
Appendix~\ref{app:global-initial-errors} selects the initial
parameters so that \eqref{eq:inverse-initial-moment-bound}
fits the available error bound; Appendix~\ref{app:global-tail}
uses \eqref{eq:inverse-tail-moment-bound} for every admissible
tail step and sums the resulting errors.

\section{Integer moment matching}
\label{app:integer-moments}

This appendix proves Proposition~\ref{prop:construction-integer-matching}
and justifies the compact selector in
\eqref{eq:construction-node-set}. The polynomial estimates below also
supply the inputs used in
Appendices~\ref{app:global-initial-moments} and~\ref{app:global-tail}
to verify the moment condition for every cell. Throughout, $\Pi_k$
denotes the real polynomials of degree at most $k$, with $k\ge1$.
All moments belong to the original signed cell measure. Restricting
the allowed atomic positions leaves that functional unchanged.

\subsection{Polynomial inequalities}
\label{app:moment-polynomials}

Let $X=[\alpha,\alpha+d]$ with $d>0$. Write
$\|p\|_{\infty,X}=\max_X|p|$,
$\operatorname{osc}_Xp=\max_Xp-\min_Xp$, and
$\operatorname{Var}_Xp=\int_X|p'(u)|\dd u$.
Total variation is unchanged by an increasing affine
reparametrization of the interval. The constants $C$ below are
absolute, may increase between estimates, and are independent
of $k$, $d$ and the location of $X$.

\paragraph{Derivative and integral bounds.}
We will use the explicit derivative estimate
\begin{equation}
  \|p'\|_{\infty,X}
     \le\frac{D_k}{d}\operatorname{osc}_Xp
     \le\frac{D_k}{d}\operatorname{Var}_Xp,
  \qquad D_k=4(k+1)^3.
  \label{eq:moment-derivative-variation}
\end{equation}
To prove it, subtract $\min_Xp$ and map $X$ affinely to $[-1,1]$.
The resulting polynomial takes values between zero and
$\operatorname{osc}_Xp$. In its Chebyshev expansion every
nonconstant coefficient has magnitude at most
$2\operatorname{osc}_Xp$. Since
$T_r(\cos\theta)=\cos(r\theta)$ gives
$\|T_r'\|_{\infty,[-1,1]}=r^2$, rescaling the derivative back
to $X$ yields
\[
  \|p'\|_{\infty,X}
    \le\frac4d\operatorname{osc}_Xp\sum_{r=1}^k r^2
    \le\frac{4(k+1)^3}{d}\operatorname{osc}_Xp.
\]
The second inequality in \eqref{eq:moment-derivative-variation}
follows by integrating $|p'|$ between a minimum and a maximum.
Constants have zero derivative and satisfy the same statement.

If $p\ge0$ on $X$ and $S=\max_Xp$, the derivative bound gives
a one-sided interval of length $d/(2D_k)$ next to a maximum
on which $p\ge S/2$. Such a direction is always available,
since at least one endpoint is at distance $d/2$ from that
maximum. Integrating on this shorter interval gives
\begin{equation}
  \|p'\|_{\infty,X}\le\frac{D_k}{d}S,
  \qquad
  S\le\frac{4D_k}{d}\int_Xp(u)\dd u.
  \label{eq:moment-polynomial-integral}
\end{equation}
A nonconstant polynomial of degree at most $k$ has at most
$k$ monotonicity intervals, each contributing at most $S$ to
its variation. Therefore
\begin{equation}
  \operatorname{Var}_Xp
     \le kS
     \le\frac{C(k+1)^4}{d}\int_Xp(u)\dd u,
  \qquad p\ge0\text{ on }X.
  \label{eq:moment-integral-variation}
\end{equation}
For a constant polynomial the left side is zero.

\paragraph{A weighted bound at a spin opening.}
Suppose $\alpha\ge0$. The estimate needed when the reference
density vanishes quadratically at $u=0$ is
\begin{equation}
  \int_Xu^2p(u)\dd u
      \ge\frac{d^2}{C(k+1)^6}\int_Xp(u)\dd u,
  \qquad p\ge0\text{ on }X.
  \label{eq:moment-weighted-integral}
\end{equation}
Indeed, remove the leftmost segment of length
$\delta=d/(16D_k)$ only for this estimate.
Its contribution to $\int_Xp$ is at most
$\delta S\le\tfrac14\int_Xp$ by
\eqref{eq:moment-polynomial-integral}.
On the remaining interval $u\ge\alpha+\delta\ge\delta$,
so the weighted integral is at least
$\tfrac34\delta^2\int_Xp$.
This proves \eqref{eq:moment-weighted-integral}, uniformly
also when $\alpha=0$.

\paragraph{Extrapolation from the upper half.}
Let $X_+=[\alpha+d/2,\alpha+d]$ and let
$S_+=\max_{X_+}p$ for $p\ge0$ on $X$. Then
\begin{equation}
  \begin{aligned}
  \max_Xp&\le6^kS_+,\\
  \operatorname{Var}_Xp&\le(k+1)6^kS_+,\\
  \int_{X_+}p(u)\dd u&\ge\frac{d}{C(k+1)^3}S_+.
  \end{aligned}
  \label{eq:moment-upper-half-bounds}
\end{equation}
For the first inequality, map $X_+$ to $[-1,1]$; the full
interval maps to $[-3,1]$. If $S_+=0$, the polynomial vanishes
identically; otherwise normalize by $S_+$. If a degree-$k$ polynomial $q$ has
$|q|\le1$ on $[-1,1]$, then
$|q(x)|\le|T_k(x)|$ for $|x|>1$.
One can see this by contradiction: a violation at $x_0>1$
would make $T_k-\lambda q$, with
$\lambda=T_k(x_0)/q(x_0)$ and $|\lambda|<1$, alternate in
sign at the $k+1$ extrema of $T_k$ in $[-1,1]$ and also
vanish at $x_0$. It would have at least $k+1$ distinct zeros.
Reflection treats $x_0<-1$.
Thus the extrapolation factor is bounded by
$T_k(3)\le(3+\sqrt8)^k<6^k$.
The variation bound follows by summing over the monotonicity
intervals, and the last inequality is
\eqref{eq:moment-polynomial-integral} on the interval of
length $d/2$.
These are precisely the estimates used in
\eqref{eq:global-upper-half-polynomials};
\eqref{eq:moment-weighted-integral} and
\eqref{eq:moment-integral-variation} supply the two tail-cell
inequalities in Appendix~\ref{app:global-tail}.

\subsection{Positive representation and unit-weight designs}
\label{app:moment-designs}

\paragraph{A positive measure representing signed moments.}
The first step uses only finitely many moments. In particular,
it imposes no pointwise positivity condition on the original
cell density.

\begin{lemma}[Positive representation on an interval]
\label{lem:moment-positive-representation}
Let $X$ be a compact interval of positive length and let
$\Lambda:\Pi_k\to\R$ be linear. If $\Lambda(1)=M>0$ and
$\Lambda(p)\ge0$ for every $p\in\Pi_k$ nonnegative on $X$,
then a positive probability measure $\sigma$ on $X$ satisfies
\begin{equation}
  \int_Xp\dd\sigma=\frac{\Lambda(p)}M,
  \qquad p\in\Pi_k.
  \label{eq:moment-positive-representation}
\end{equation}
It can be chosen to have at most $k+1$ atoms.
\end{lemma}
\begin{proof}
Let $\Gamma(s)=(s,s^2,\ldots,s^k)$ and let $\mathcal K$ be its
convex hull for $s\in X$. A convex combination of more than
$k+1$ points can be reduced using an affine dependence, varying
the weights until one becomes zero while keeping all weights
nonnegative. Repetition leaves at most $k+1$ points. Hence
$\mathcal K$ is the continuous image of the compact product
of $X^{k+1}$ and the probability simplex, and is compact.

Put $y=(\Lambda(s),\ldots,\Lambda(s^k))/M$. If $y\notin\mathcal K$,
choose the closest point $x_*\in\mathcal K$ and set $v=y-x_*$.
Minimality of the distance gives $v\cdot(x-x_*)\le0$ for all
$x\in\mathcal K$. Consequently the polynomial
$p(s)=v\cdot(x_*-\Gamma(s))$ is nonnegative on $X$, whereas
$\Lambda(p)/M=v\cdot(x_*-y)=-|v|^2<0$.
This contradiction proves $y\in\mathcal K$.
Its convex weights define $\sigma$, and matching the constant
and the $k$ coordinate monomials proves
\eqref{eq:moment-positive-representation}.
\end{proof}

\paragraph{Moving the allowed positions away from the left endpoint.}
We now start with a functional on $[0,1]$ and keep it fixed
while restricting possible representing points to $[h,1]$.
The following estimate is the quantitative step needed for that
restriction.

\begin{lemma}[Positivity on the smaller interval]
\label{lem:moment-inward-positivity}
Suppose $\Lambda(1)=M>0$ and
\begin{equation}
  \Lambda(p)\ge\beta\operatorname{Var}_{[0,1]}p
  \quad(p\in\Pi_k,\ p\ge0\text{ on }[0,1]),
  \qquad \beta>\frac12.
  \label{eq:moment-full-variation-bound}
\end{equation}
For every $0<h<1$ with $\beta_h:=\beta-MD_kh>1/2$,
\begin{equation}
  \Lambda(p)\ge\beta_h\operatorname{Var}_{[0,1]}p
       \ge\beta_h\operatorname{Var}_{[h,1]}p
  \quad(p\in\Pi_k,\ p\ge0\text{ on }[h,1]).
  \label{eq:moment-inward-variation-bound}
\end{equation}
Moreover, $\Lambda(p)>0$ for every nonzero such polynomial.
\end{lemma}
\begin{proof}
For $p\ge0$ on $[h,1]$, set $m=\min_{[0,1]}p$.
If $m<0$, a minimum lies in $[0,h)$, and comparison with
$p(h)\ge0$ gives
\[
  m\ge-h\|p'\|_{\infty,[0,1]}
      \ge-hD_k\operatorname{Var}_{[0,1]}p.
\]
This lower bound also holds if $m\ge0$.
The polynomial $p-m$ is nonnegative on the whole original
interval, so
\[
  \Lambda(p)=\Lambda(p-m)+Mm
      \ge(\beta-MD_kh)\operatorname{Var}_{[0,1]}p.
\]
Restricting the variation integral gives the second inequality
in \eqref{eq:moment-inward-variation-bound}.
A nonconstant polynomial has strictly positive variation.
A nonzero nonnegative constant has value $\Lambda(p)=Mp>0$.
This also proves the last assertion.
\end{proof}

Lemma~\ref{lem:moment-positive-representation} now supplies a
positive representing measure on $[h,1]$ for the unchanged
functional. In the spectral application, this measure represents
the moments of all of $(L,V]$, including the portion with
$0<s<h$.

\paragraph{The prescribed number of unit weights.}
We use the interval specialization of Kane's design theorem
\cite[Prop.~20 and the preceding formula for $K_\gamma$]{Kane2012}.
For a positive probability measure $\sigma$ on a compact
interval $X$, define
\begin{equation}
  K_X=\sup_{\substack{p\in\Pi_k,\ p\ge0\text{ on }X\\p\ne0}}
       \frac{\operatorname{Var}_Xp}{\int_Xp\dd\sigma}.
  \label{eq:moment-Kane-ratio}
\end{equation}
When the denominators are positive and $K_X$ is finite, the
theorem gives, for each integer $N>K_X/2$, points
$s_1,\ldots,s_N\in X$ with
\begin{equation}
  \frac1N\sum_{\ell=1}^N p(s_\ell)=\int_Xp\dd\sigma,
  \qquad p\in\Pi_k.
  \label{eq:moment-Kane-design}
\end{equation}
Repeated points are permitted. This is the external design
theorem used here. To match its notation, take $W=\Pi_k$ and
let $\gamma$ traverse $X$ monotonically. The path variation
is then $\operatorname{Var}_X$, and the mean-zero subspace
has dimension $k>0$. The measure representing the moments may
be atomic, as allowed in Kane's definition of a design problem.

\begin{lemma}[Unit weights from a quantitative moment bound]
\label{lem:moment-unit-design}
Let $X$ be a compact interval of positive length and let
$\Lambda:\Pi_k\to\R$ be linear. Suppose
\begin{equation}
  \Lambda(1)=M\in\Z_{>0},\qquad
  \Lambda(p)\ge\beta_X\operatorname{Var}_Xp
     \quad(p\ge0\text{ on }X),\qquad \beta_X>\frac12.
  \label{eq:moment-interval-design-condition}
\end{equation}
Then exactly $M$ unit-weight points in $X$ represent $\Lambda$
on $\Pi_k$.
\end{lemma}
\begin{proof}
The variation bound makes $\Lambda$ strictly positive on
nonconstant nonnegative polynomials; positive constants have
positive value because $M>0$. Let $\sigma$ be the probability
measure from Lemma~\ref{lem:moment-positive-representation}.
Its denominators in \eqref{eq:moment-Kane-ratio} are positive,
and
\begin{equation}
  K_X=\sup_{\substack{p\in\Pi_k,\ p\ge0\text{ on }X\\p\ne0}}
        \frac{M\operatorname{Var}_Xp}{\Lambda(p)}
      \le\frac M{\beta_X}<2M.
  \label{eq:moment-prescribed-count}
\end{equation}
Thus $N=M$ is allowed in \eqref{eq:moment-Kane-design}.
Multiplying that identity by $M$ gives
$\sum_{\ell=1}^M p(s_\ell)=\Lambda(p)$.
\end{proof}

\begin{proof}[Proof of Proposition~\ref{prop:construction-integer-matching}]
Apply Lemma~\ref{lem:moment-inward-positivity} with $\beta=1$
and the specified $h=[64M(k+1)^3]^{-1}$. Then
\begin{equation}
  MD_kh=\frac1{16},\qquad \beta_h=\frac{15}{16}>\frac12.
  \label{eq:moment-chosen-inward-bound}
\end{equation}
Lemma~\ref{lem:moment-unit-design} on $X=[h,1]$ now gives
the $M$ unit-weight points and every identity in
\eqref{eq:construction-unit-moments}.
\end{proof}

\subsection{A definite selector on a half-open cell}
\label{app:moment-selector}

\paragraph{Compactness and successive minimization.}
For the original cell functional in
\eqref{eq:construction-cell-functional}, the preceding proposition
makes the sorted feasible set $\mathcal X$ of
\eqref{eq:construction-node-set} nonempty. It is a closed
subset of $[h,1]^M$: ordering and all its moment equations
are closed conditions. Thus $\mathcal X$ is compact.
Minimize its first coordinate and retain the minimizer set;
then minimize the second coordinate on that set, and continue.
At each of the finitely many stages the remaining set is
nonempty and compact. After $M$ stages every coordinate is
fixed, giving a unique vector. This is the lexicographic
minimum used in Section~\ref{sec:construction-moments}.

Let $d=u_+-u_->0$. The energy map in
\eqref{eq:construction-cell-atoms} is increasing on the allowed
interval, and every selected point satisfies
\begin{equation}
  L+d h(2u_-+d h)\le\Ehat_\ell\le V.
  \label{eq:moment-selected-support}
\end{equation}
In particular $L<\Ehat_\ell$, also at a spin opening where
$u_-=0$. Repeated coordinates give repeated energies, whose
weights add to positive integer multiplicities in that spin
row. A point at $V$ is allowed and belongs to the present
cell $(L,V]$; it is excluded from the next cell in that row.

\paragraph{The full cell moments and their residual.}
Since $u=u_-+ds$, each $u^r$ with $0\le r\le k$ is a
polynomial in $s$ of degree at most $k$. The selected atomic
measure therefore obeys
\begin{equation}
  \int u^r\dd Q=\int_{(L,V]}u^r\dd\mu^{\mathrm{cell}},
  \qquad 0\le r\le k,
  \qquad u=\sqrt{\Ehat-|j|}.
  \label{eq:moment-full-cell-identities}
\end{equation}
The zeroth identity gives $Q((L,V])=M$.
For $\nu=Q-\mu^{\mathrm{cell}}$, positivity of $Q$ and
$M=\mu^{\mathrm{cell}}((L,V])>0$ also give
\begin{equation}
  \|\nu\|_{\TV}\le M+\|\mu^{\mathrm{cell}}\|_{\TV}
                \le2\|\mu^{\mathrm{cell}}\|_{\TV}.
  \label{eq:moment-cell-total-variation}
\end{equation}
These are \eqref{eq:construction-residual-properties}, including
the cancellations needed in
Appendix~\ref{app:complete-repair-estimates}.
The residual still includes the continuous measure on the
entire original cell, so its polynomial approximation is made
on the full interval $[u_-,u_+]$. Restricting the new atomic
positions introduces no inverse power of $h$ into those
exterior estimates.

\paragraph{Admissible choices.}
More generally, under \eqref{eq:moment-full-variation-bound},
any $h$ satisfying
\begin{equation}
  0<h<1,\qquad \beta-MD_kh>\frac12
  \label{eq:moment-admissible-inset}
\end{equation}
gives a nonempty compact feasible set with the same mass and
moments. Every vector in that set has the local properties
\eqref{eq:moment-selected-support}--\eqref{eq:moment-cell-total-variation}.
The displayed $h$ and lexicographic rule fix the particular
member studied in this paper. An alternative admissible
selection can change individual energies and multiplicities,
and therefore the later cell functionals through its exterior
response. The construction estimates apply to those choices
when the processing order, endpoint choices, moment degrees and
repair cutoffs satisfy the conditions of
Section~\ref{sec:construction-schedule}.

\paragraph{Endpoint multiplicity for the lexicographic rule.}
We record the local consequence used in the later endpoint
counting arguments. For $X=[h,1]$, write
$\operatorname{ev}_h(p)=p(h)$. The maximal real weight that can
be assigned to the leftmost allowed position while retaining
positive finite moments is
\begin{equation}
  \lambda=\max\left\{w\in[0,M]:
       \Lambda(p)-w p(h)\ge0\text{ for every }p\in\Pi_k
                               \text{ nonnegative on }X\right\}.
  \label{eq:moment-real-endpoint-weight}
\end{equation}

\begin{lemma}[Real endpoint weight and integer multiplicity]
\label{lem:moment-endpoint-multiplicity}
Suppose \eqref{eq:moment-interval-design-condition} holds on
$X=[h,1]$ with $\beta_X>1$. Let $m_h$ be the number of copies
of $h$ in the lexicographically least feasible vector of
$M$ unit-weight points. Then $\lambda>0$ and
\begin{equation}
  \left\lfloor(1-\beta_X^{-1})\lambda\right\rfloor
       \le m_h\le\lambda,
  \qquad
  0\le1-\frac{m_h}{\lambda}
       \le\beta_X^{-1}+\lambda^{-1}.
  \label{eq:moment-integer-endpoint-bound}
\end{equation}
\end{lemma}
\begin{proof}
The admissible weights in \eqref{eq:moment-real-endpoint-weight}
form a nonempty closed subset of $[0,M]$, so the maximum is
attained. Testing \eqref{eq:moment-interval-design-condition}
with $(s-h)/(1-h)$ and its complement gives $M\ge2\beta_X$.
For a polynomial $p\ge0$ on $X$, put $m=\min_Xp$.
Then
\[
  \Lambda(p)\ge Mm+\beta_X\operatorname{Var}_Xp
      \ge\beta_X p(h),
\]
because $p(h)\le m+\operatorname{Var}_Xp$.
Thus $\lambda\ge\beta_X>0$.

Let $q=\lfloor(1-\beta_X^{-1})\lambda\rfloor$. The identity
\begin{equation}
  \begin{aligned}
  \Lambda-q\operatorname{ev}_h
    ={}&\beta_X^{-1}\Lambda
       +(1-\beta_X^{-1})(\Lambda-\lambda\operatorname{ev}_h)\\
       &+\bigl[(1-\beta_X^{-1})\lambda-q\bigr]
                                             \operatorname{ev}_h
  \end{aligned}
  \label{eq:moment-endpoint-subtraction}
\end{equation}
is a sum of functionals nonnegative on the same polynomial
cone. Its first term is at least $\operatorname{Var}_Xp$.
The residual total weight $M-q$ is a positive integer since
$q<\lambda\le M$. Lemma~\ref{lem:moment-unit-design} therefore
represents this residual by exactly $M-q$ unit-weight points
in $X$. Adding $q$ copies of $h$ gives a feasible vector for
the original functional. The lexicographic minimum must
contain at least that many copies of $h$.
Conversely, subtracting its $m_h$ copies leaves a positive
atomic measure, so $m_h\le\lambda$ by definition.
The relative estimate in \eqref{eq:moment-integer-endpoint-bound}
follows from the floor error being less than one.
\end{proof}

An elementary lower bound follows directly from subtraction:
for $p\ge0$ on $X$ and $m=\min_Xp$,
\[
  (\Lambda-r\operatorname{ev}_h)(p)
    \ge(M-r)m+(\beta_X-r)\operatorname{Var}_Xp.
\]
Thus any integer $r$ with $r<M$ and $\beta_X-r>1/2$ can be
placed at $h$ before applying Lemma~\ref{lem:moment-unit-design}
to the remainder. In particular, $r=\lfloor\beta_X/2\rfloor$
is permitted for $\beta_X>1$.
These statements concern the fixed lexicographic selector;
the capacity asymptotics and the exclusion of other nearby
points are separate estimates in
Appendix~\ref{app:entropy-positive-kappa}.

\section{Uniform estimates for the recursion}
\label{app:global-estimates}

Appendices~\ref{app:global-vacuum}--\ref{app:global-tail} prove
Lemma~\ref{lem:global-step-estimates} for the initialization of
Section~\ref{sec:construction-initial} and the explicit schedule of
Section~\ref{sec:construction-schedule}.
The estimates use only the support, total weight and matched moments
of each replacement, so every choice in $\mathcal X$ is allowed.
The local inputs are the canonical kernel from
Appendix~\ref{app:canonical}, the finite-measure inverse, anchor and
complete repair estimates from Appendix~\ref{app:inversion-repair},
and the polynomial and matching results from
Appendix~\ref{app:integer-moments}.
We apply these inputs to choose the global parameters and bound the
errors uniformly over the successive cells.
All spectral densities are taken with respect to $\dd\omega_J$.
Constants $c,C>0$ may change between estimates.
A subscript $B$ permits dependence on the fixed initial cutoff but
never on $a$ or the prescribed first-level data. We display the
dependence on their total-weight bound $M_*$ separately.
The growing-cutoff estimates use the separate exponential bounds
of Appendix~\ref{app:inversion-repair}; no uniformity as $B$ tends
to zero is required.
Appendix~\ref{app:global-thermal-comparison} supplies the thermal
bound used to pass to the limit in Section~\ref{sec:global-limit}.
Only afterward does Appendix~\ref{app:global-near-extremal} use
that established positive limit to study near-extremal windows.

\subsection{The vacuum reference and its remainder}
\label{app:global-vacuum}

The four vacuum seeds are $(-a,0)$, $(1-a,1)$, $(1-a,-1)$ and
$(2-a,0)$, with coefficients $1,-1,-1,1$.
The seed-energy continuation in
\eqref{eq:canonical-energy-subtraction} also applies to these
negative energies. The cosine factors in the kernel are entire
functions of their squared arguments and become hyperbolic
cosines here. Their bounds below give locally uniform thermal
integrability: on a compact set of seed energies the same
product-minus-one estimate gives at most a polynomial in
$\Ehat$ times $e^{C\sqrt{\Ehat}}$. The spectral formula
therefore continues to these vacuum seeds.
In the $m=1$ summand
of \eqref{eq:repair-Q-kernel}, the Kloosterman sum equals one.
Writing $A_\pm=\cosh(2\pi\sqrt{a(\Ehat\pm J)})$ and
$C_\pm=\cosh(2\pi\sqrt{(a-2)(\Ehat\pm J)})$ for this computation,
the four terms combine to
\[
  2\bigl[(A_+A_--1)-(C_+A_--1)-(A_+C_--1)+(C_+C_--1)\bigr]
     =2(A_+-C_+)(A_--C_-).
\]
This proves \eqref{eq:global-leading-density}. The superscript $(1)$
refers to the arithmetic denominator $m=1$.

For $a\ge100$ and $\Ehat\ge|J|$, the reference satisfies
\begin{equation}
  c_*(\Ehat^2-J^2)e^{8\sqrt{a\Ehat}}
    \le\rho^{(1)}_{a,J}(\Ehat)
    \le2e^{4\pi\sqrt{a\Ehat}},
  \qquad c_*=\frac{2\pi^4}{625}.
  \label{eq:global-leading-bounds}
\end{equation}
To retain the zero at the spin threshold, use
\[
  \begin{gathered}
    D_a(z)=2\sinh(A\sqrt z)\sinh(C'\sqrt z),\qquad AC'=2\pi^2,\\
    A=\pi(\sqrt a+\sqrt{a-2}),\qquad
    C'=\pi(\sqrt a-\sqrt{a-2}).
  \end{gathered}
\]
The elementary inequalities $\sinh x\ge(x/100)e^{19x/20}$ and
$\sinh x\ge x$ for $x\ge0$ give
$D_a(z)\ge(\pi^2/25)z e^{(19/20)A\sqrt z}$.
Together with
$\sqrt{2\Ehat}\le\sqrt{\Ehat+J}+\sqrt{\Ehat-J}\le2\sqrt{\Ehat}$,
this proves the lower bound in \eqref{eq:global-leading-bounds};
indeed $(19/20)\pi(1+\sqrt{0.98})\sqrt2>8$.
The upper bound follows from $D_a(z)\le e^{2\pi\sqrt{az}}$.

Let $\rho^{\mathrm{rem}}_{a,J}$ be the MWK continuous density after
subtracting $\rho^{(1)}_{a,J}$. It obeys
\begin{equation}
  |\rho^{\mathrm{rem}}_{a,J}(\Ehat)|
     \le Ca\Ehat e^{2\pi\sqrt{a\Ehat}},
  \qquad \Ehat\ge\max\{1,|J|\}.
  \label{eq:global-mwk-remainder-bound}
\end{equation}
For any vacuum seed, the two nonnegative hyperbolic arguments
$u_m,v_m$ at denominator $m$ satisfy
$u_m+v_m\le4\pi\sqrt{a\Ehat}/m$ and
$u_m^2+v_m^2\le8\pi^2a\Ehat/m^2$. Thus
\[
  |\cosh u_m\cosh v_m-1|
       \le\frac{Ca\Ehat}{m^2}e^{4\pi\sqrt{a\Ehat}/m}.
\]
Keeping the subtraction of one inside the bracket provides the
summable $m^{-2}$ factor. Since $|S(J,J';m)|/m\le1$, summing over
$m\ge2$ and the four vacuum seeds gives
$Ca\Ehat e^{2\pi\sqrt{a\Ehat}}$.
The separate arithmetic term in \eqref{eq:repair-Q-kernel} contributes
at most $C|J|$, which is absorbed on the stated domain.
For scalar output that term vanishes. Including $m=1$ in the same
bracket estimate therefore also gives
\begin{equation}
  |\rho^{\mathrm{MWK}}_{a,0}(\Ehat)|
     \le Ca\Ehat e^{4\pi\sqrt{a\Ehat}},\qquad \Ehat>0.
  \label{eq:global-scalar-vacuum-bound}
\end{equation}
In particular, the scalar continuous measure is integrable at zero.
This estimate is needed when the actual clearing cutoff is below one.

\subsection{The complete initialization response}
\label{app:global-initial-reference}

Denote by $\rho^{\mathrm{resp}}_{a,B,J}$ the change in the continuous
density produced by the initialization of
Section~\ref{sec:construction-initial}, suppressing its dependence on
the prescribed first-level data. Above the initial cutoff,
\begin{equation}
  \rho_{0,J}(\Ehat)-\rho^{(1)}_{a,J}(\Ehat)
    =\rho^{\mathrm{rem}}_{a,J}(\Ehat)
     +\rho^{\mathrm{resp}}_{a,B,J}(\Ehat),
  \qquad \Ehat>B,\quad\Ehat>|J|.
  \label{eq:global-initial-decomposition}
\end{equation}
The response includes the clearing repair, the prescribed first level,
and all initialization anchors, including their direct density on
$[2B,3B]$.

To make this last contribution explicit, set
$(\alpha,r)=(\nu_{\mathrm{first}},0)$ when $\Ehat_*>0$, and
$(\alpha,r)=(0,d_0)$ when $\Ehat_*=0$. Thus $\alpha$ contains the
prescribed positive-energy atoms and $r$ is the prescribed threshold
multiplicity. They satisfy
\begin{equation}
  \|\alpha\|_{\TV}\le M_*,\qquad
  \sum_J\int\sqrt{\Ehat}\,\dd|\alpha_J|\le M_*\sqrt B,
  \qquad 0\le r\le M_*.
  \label{eq:global-first-level-norms}
\end{equation}
Set
\[
  \nu^{\mathrm{init}}=-[\mu_{\mathrm{MWK}}^{>0}]_{\le B}+\alpha,
  \qquad w^{\mathrm{init}}=\mathcal I_B[\nu^{\mathrm{init}}],
  \qquad s^{\mathrm{init}}=6+r-t(w^{\mathrm{init}}).
\]
The full response is
\begin{equation}
  \rho^{\mathrm{resp}}_{a,B,J}
    =s^{\mathrm{init}}\delta_{J,0}\rho_B^\psi
      +\bigl(R[w^{\mathrm{init}}+s^{\mathrm{init}}\theta_B]\bigr)_J.
  \label{eq:global-response-formula}
\end{equation}
The coefficient six compensates the MWK threshold once, at
initialization. It does not recur in the later zero-threshold repairs.

The vacuum estimates, the compact-input bound
\eqref{eq:repair-compact-kernel-bound}, and the inverse and anchor
estimates proved in Appendices~\ref{app:finite-measure-inverse}
and~\ref{app:threshold-anchor} imply
\begin{equation}
  |\rho^{\mathrm{resp}}_{a,B,J}(\Ehat)|\le
  \begin{cases}
    C_B(1+M_*)a\sqrt{\Ehat}\,e^{4\pi\sqrt{aB}},
       & B>0\text{ fixed},\\[2pt]
    C(1+M_*)a b^4\sqrt{\Ehat}\,e^{4\pi\sqrt{ab}+C_{\mathrm{ref}}b},
       & B=b=\kappa a,\quad 2\le b<a/100,
  \end{cases}
  \label{eq:global-initial-response-bound}
\end{equation}
for $\Ehat>B$ and $\Ehat>|J|$.
For fixed $B$, integrating the vacuum density through $B$ gives
$\|[\mu_{\mathrm{MWK}}^{>0}]_{\le B}\|_{\TV}
\le C_Ba e^{4\pi\sqrt{aB}}$. Equation~\eqref{eq:global-scalar-vacuum-bound}
controls the scalar integral down to zero, and the nonzero-spin edges
are integrable. The bounds in \eqref{eq:global-first-level-norms}
are uniform in the first-level energies, spins and multiplicities.
Applying the fixed-band inverse and anchor bounds to
\eqref{eq:global-response-formula} proves the first line of
\eqref{eq:global-initial-response-bound}. The direct term is covered
because $\Ehat\ge2B$ on its support. No inverse power of $\Ehat_*$,
$\Ehat_*-|J|$ or $B-\Ehat_*$ enters. The case $\Ehat_*=0$ uses $r=d_0$ instead of
inverting an origin atom. At a nonzero spin edge or at the upper
band boundary, the measure inverse retains the prescribed atom,
as shown in Appendix~\ref{app:finite-measure-inverse}.

For growing $b$, the same integration gives the loose bound
$Ca b^{5/2}e^{4\pi\sqrt{ab}}$ for the low-band input.
The bounds \eqref{eq:inverse-measure-bound},
\eqref{eq:inverse-threshold-functional} and
\eqref{eq:anchor-size-bound} have costs at most $e^{Cb}$ for
$b\ge2$, after absorbing polynomial factors.
All input energies in \eqref{eq:global-response-formula} are at most
$3b$. Applying the kernel bound and including the direct term yields
the second line, with $C,C_{\mathrm{ref}}$ independent of
$a,b,\kappa,M_*$ and the prescribed spins.
Polynomial factors in $b$ can be absorbed by increasing
$C_{\mathrm{ref}}$; the fixed-$B$ constants are not used here.
The factor $1+M_*$ is kept outside the exponential, so the same
$C_{\mathrm{ref}}$ as for one scalar primary suffices.

Dividing \eqref{eq:global-initial-decomposition} by $H_a$ now gives,
at fixed $B$,
\[
  \frac{|\rho_{0,J}-\rho^{(1)}_{a,J}|}{H_a}
    \le Ca\Ehat e^{-(7-2\pi)\sqrt{a\Ehat}}
       +C_B(1+M_*)a\sqrt{\Ehat}\,
                   e^{4\pi\sqrt{aB}-7\sqrt{a\Ehat}}.
\]
Both terms are eventually decreasing for $\Ehat\ge T$.
With $T=5B^\sharp$,
$7\sqrt T-4\pi\sqrt B\ge(7\sqrt5-4\pi)\sqrt{B^\sharp}>0$.
Their supremum tends to zero uniformly in all allowed first-level
data with total weight at most $M_*$, and in output spin.
For the growing cutoff, fix $R_0>5$ such that
\begin{equation}
  7\sqrt{R_0}>4\pi+C_{\mathrm{ref}}/10+1.
  \label{eq:global-R0-choice}
\end{equation}
Since $b/a<1/100$, the exponent of the response divided by $H_a$
at $T=R_0b$ is bounded above by
$(4\pi+C_{\mathrm{ref}}/10-7\sqrt{R_0})\sqrt{ab}$.
For each fixed $\kappa$ and $M_*$, this overcomes all prefactors.
The vacuum remainder is smaller because $7>2\pi$.
This proves \eqref{eq:global-initial-reference} in both regimes.

We also need the total variation of the initialized continuum through
a finite energy $W$. The same estimates give
\begin{equation}
  \begin{aligned}
  \sum_J\int_{\max\{B,|J|\}}^W
             |\rho_{0,J}(\Ehat)|\dd\omega_J
  &\le C_B(1+M_*)aW^2e^{4\pi\sqrt{aW}},
       && B\text{ fixed},\quad W\ge B^\sharp,\\
  \sum_J\int_{\max\{b,|J|\}}^W
             |\rho_{0,J}(\Ehat)|\dd\omega_J
  &\le C(1+M_*)W^2\left[a e^{4\pi\sqrt{aW}}
              +a b^4e^{4\pi\sqrt{ab}+C_{\mathrm{ref}}b}\right],
       && W\ge b\ge2.
  \end{aligned}
  \label{eq:global-initial-variation}
\end{equation}
A row with lower limit above $W$ contributes zero.
To obtain these bounds, sum the elementary integrals of
$\Ehat\dd\omega_J$ and $\sqrt{\Ehat}\dd\omega_J$ over
$|J|\le W$; they are at most $CW^2$ and $CW^{3/2}$, respectively.
For a fixed cutoff below one, treat the scalar vacuum integral near
zero by \eqref{eq:global-scalar-vacuum-bound} and absorb the finite
response integral starting at $B>0$ into $C_B$.
These bounds include the direct anchor density, but exclude the
prescribed atoms from the continuous cell functionals.

\subsection{The moment inequality in the initial cells}
\label{app:global-initial-moments}

For an initial row, put $L=\max\{B,|J|\}$, choose
$V\in[U,U+1]$, and use $\xi=\sqrt{\Ehat-|J|}$.
Write $I^\xi=[\sqrt{L-|J|},\sqrt{V-|J|}]$ and let $d$ be its
length. A polynomial $p$ in the normalized coordinate of
\eqref{eq:construction-cell-functional} becomes a polynomial
$\widetilde p$ of the same degree on $I^\xi$; total variation is
unchanged by this affine reparametrization.
For $\widetilde p\ge0$ of degree at most $k$, let $S_+$ be its
maximum on the upper half of $I^\xi$. The polynomial estimates
\eqref{eq:moment-upper-half-bounds} give
\begin{equation}
  \sup_{I^\xi}\widetilde p\le C6^kS_+,\qquad
  \operatorname{Var}_{I^\xi}\widetilde p\le C(k+1)6^kS_+,\qquad
  \int_{\text{upper half}}\widetilde p\dd\xi
        \ge\frac{cd}{(k+1)^3}S_+.
  \label{eq:global-upper-half-polynomials}
\end{equation}
Here the factor $6^k$ is the Chebyshev extrapolation bound from a
half interval; the integral estimate follows by controlling the
polynomial derivative near a maximum.

In either regime, $|J|\le T\le U/16$ and $B\le U/16$.
On the upper half of $I^\xi$, the energy is at least $V/4>T+1$,
the length $d$ is bounded below by $c\sqrt V$, and
$\Ehat-|J|\ge cV$. Since
\[
  \frac{\dd\mu_0^\xi}{\dd\xi}
     =\frac{2\rho_{0,J}(|J|+\xi^2)}{\sqrt{\xi^2+2|J|}},
\]
\eqref{eq:global-leading-bounds} and
\eqref{eq:global-initial-reference} give the lower bound
$cV^{-1/2}e^{4\sqrt{aV}}$ for this ordinary density there.
The continuous density in every initial row is positive above
$T+1$ for sufficiently large $a$. Let $N_-$ be the combined
negative variation of these rows below $T+1$.
Equations~\eqref{eq:global-upper-half-polynomials} then yield
\begin{equation}
  \Lambda_{J,V}(p)\ge S_+
     \left[\frac{ce^{4\sqrt{aV}}}{(k+1)^3}-CN_-6^k\right].
  \label{eq:global-initial-moment-comparison}
\end{equation}
All functionals here use the continuous measure of $F_0$.

For fixed $B$, \eqref{eq:global-initial-variation} bounds
$N_-\le C_B(1+M_*)a e^{4\pi\sqrt{a(T+1)}}$.
With $k=\lfloor\sqrt{aU}\rfloor$, the strict inequality
\[
  4-\log6-4\pi\sqrt{\frac{T+1}{U}}>1
\]
shows that the negative term in
\eqref{eq:global-initial-moment-comparison} is eventually smaller
than half the positive term. The estimate is uniform in the
admissible first-level data with total weight at most $M_*$, $V$
and the initial row.

For $B=b=\kappa a$, \eqref{eq:global-initial-variation} at $W=T+1$
gives
$\log(2+N_-)\le C_N\sqrt{ab}+O_\kappa(\log a)+\log(1+M_*)$,
where $C_N$ is fixed after $R_0$ and $C_{\mathrm{ref}}$.
Choose $R_{\mathrm{init}}\ge16R_0$ with
$C_N/\sqrt{R_{\mathrm{init}}}\le1$.
Then the logarithm of $N_-6^k$ is at most
$(1+\log6)\sqrt{aU}+O_\kappa(\log a)+\log(1+M_*)$, whereas the positive
exponent is at least $4\sqrt{aU}$. The same domination follows.
In both regimes, therefore,
\begin{equation}
  \Lambda_{J,V}(p)\ge\beta_a\operatorname{Var}_{[0,1]}p,
  \qquad
  \beta_a\ge\frac{c\,e^{4\sqrt{aU}-k\log6}}{(k+1)^4}
                 \longrightarrow\infty.
  \label{eq:global-initial-reserve}
\end{equation}
The constant may depend on $B$ in the fixed-cutoff case.
This proves the variation inequality for every allowed endpoint.
Applying \eqref{eq:global-initial-moment-comparison} to $p=1$
at $V=U$ also gives $m(U)>0$. On $[U,U+1]$, the density is
positive almost everywhere and its integral exceeds one for large
$a$, by the same leading lower bound. These are the three endpoint
conditions required in Section~\ref{sec:construction-moments}.
The integer endpoint rule then supplies positive integer total weight.
Once $\beta_a\ge1$, Proposition~\ref{prop:construction-integer-matching},
proved in Appendix~\ref{app:moment-designs}, supplies the integer
replacement of the entire signed cell. The smaller interval $[h,1]$
restricts only the new atomic positions.

\subsection{Complete exterior errors from the initial replacements}
\label{app:global-initial-errors}

For a matched cell, \eqref{eq:moment-full-cell-identities} gives
zero moments through degree $k$ for
$\nu=Q-\mu^{\mathrm{cell}}$, and
\eqref{eq:moment-cell-total-variation} gives
\begin{equation}
  \|\nu\|_{\TV}\le2\|\mu^{\mathrm{cell}}\|_{\TV}.
  \label{eq:global-signed-residual-bound}
\end{equation}
The completion estimates use the whole interval $[L,V]$ containing
the residual. Restricting the nodes to the inset of
Proposition~\ref{prop:construction-integer-matching} introduces
no inverse power of that inset.

We apply the initial-cell specialization
\eqref{eq:inverse-initial-moment-bound} from
Appendix~\ref{app:complete-repair-estimates}.
Fix $\varepsilon_{\mathrm{ell}}>0$, put
$\gamma=2\pi\varepsilon_{\mathrm{ell}}$, and take
$0<\tau_0\le\min\{\varepsilon_{\mathrm{ell}}^2,1/10\}$ with
$2\gamma<7$. For example,
$\varepsilon_{\mathrm{ell}}=1/8$ and $\tau_0=1/64$ meet these
conditions. For an initial cell with
$0<L<V\le\tau_0a$ and $|J'|\le L$, set
\[
  d_\xi=\sqrt{V-|J'|}-\sqrt{L-|J'|},\qquad
  \varpi=1+\frac{\varepsilon_{\mathrm{ell}}\sqrt a}{d_\xi}\ge2.
\]
For a repair cutoff $C\ge2$ with $2V+2\le C\le2V+4$, the full
exterior satisfies
\begin{equation}
  |r_{C,J}[\nu](\Ehat)|
     \le C_*\|\nu\|_{\TV}\sqrt a\,\varpi^{-k}
         e^{C_*C+\gamma\sqrt{aC}}
         \sqrt{\Ehat}\,e^{\gamma\sqrt{a\Ehat}}.
  \label{eq:global-long-completion}
\end{equation}
Here and in the short-cell estimate below, $\Ehat>C$ and
$\Ehat>|J|$; the constants are uniform in output spin.
The letter $C$ in $r_{C,J}$ denotes the cutoff, while $C_*$
denotes an estimate constant.

The hypotheses of \eqref{eq:inverse-initial-moment-bound} are
satisfied: $V<C$, $C\ge2$ and $V\le\tau_0a$. Its Bernstein
ellipse has height at most $\varepsilon_{\mathrm{ell}}\sqrt a/2$,
and matching the moments on the full square-root interval supplies
the factor $\varpi^{-k}$. The bound is for the complete response
\eqref{eq:inverse-moment-full-response}, which here reads, with
$v_\nu=\mathcal I_C[\mathscr R_C\nu]$,
\begin{equation}
  r_{C,J}[\nu]
     =(R\nu)_J-(Rv_\nu)_J
       +t(v_\nu)(R\theta_C)_J
       +t(v_\nu)\delta_{J,0}\rho_C^\psi.
  \label{eq:global-complete-moment-response}
\end{equation}
The final term is the direct anchor density. All four terms are
included in \eqref{eq:global-long-completion}. The estimate remains
valid when $0<L<2$: the kernel is approximated on the original
square-root interval, and no inverse power of $L$ is introduced.
The recursive cutoff $C\ge2$ is distinct from the possibly small
initial clearing cutoff $B$.

For the initial cells, $C=B_n=2U+4$ and
$\varpi\ge\varepsilon_{\mathrm{ell}}\sqrt{a/(U+1)}$.
Dividing \eqref{eq:global-long-completion} by $H_a$ and taking the
supremum above $B_n$ gives
\begin{equation}
  \varepsilon_n
    \le C_*\|\nu_n\|_{\TV}\sqrt{aB_n}\,e^{C_*B_n}
       \left(\varepsilon_{\mathrm{ell}}
                      \sqrt{\frac a{U+1}}\right)^{-k_n},
  \qquad n<N_{\mathrm{init}}.
  \label{eq:global-initial-step-error}
\end{equation}
For sufficiently large $a$, the remaining function of $\Ehat$ is
decreasing above $B_n$. Its value there contains the factor
$e^{-(7-2\gamma)\sqrt{aB_n}}\le1$, which was discarded.

Let $\mathfrak E_{\mathrm{init}}
=\sum_{n<N_{\mathrm{init}}}\varepsilon_n$.
At fixed $B$, the sufficient condition
\[
  a\ge\max\{\tau_0^{-1},4\varepsilon_{\mathrm{ell}}^{-2}\}(U+1)
\]
makes the approximation factor at least two and ensures the
initial-cell hypotheses. With the example parameters above,
$a\ge256(U+1)$ suffices.
Equations~\eqref{eq:global-initial-variation},
\eqref{eq:global-signed-residual-bound} and
\eqref{eq:global-initial-step-error}, with fixed $U$ and a fixed
number of initial rows, imply
\begin{equation}
  \log\mathfrak E_{\mathrm{init}}
    \le-\tfrac12\sqrt{aU}\log a
                 +O_B(\sqrt a+\log a)+\log(1+M_*)\longrightarrow-\infty.
  \label{eq:global-fixed-initial-error}
\end{equation}
This is uniform in $\delta\in[0,B)$ and all admissible first-level
data with total weight at most $M_*$.

For the growing cutoff put $\vartheta=R_{\mathrm{init}}\kappa$,
so $U=\vartheta a$. The continuous source bound, the $e^{CB_n}$
completion cost and the $O(a)$ initial rows instead give
\begin{equation}
  \begin{aligned}
  \log\mathfrak E_{\mathrm{init}}
  &\le C_0\sqrt{a(U+1)}+C_1U
     -\lfloor\sqrt{aU}\rfloor
       \log\!\left(\varepsilon_{\mathrm{ell}}
                         \sqrt{\frac a{U+1}}\right)
       +O_\kappa(\log a)+\log(1+M_*)\\
  &\le a\sqrt\vartheta
       \left[C_0+C_1\sqrt\vartheta
          -\log\frac{\varepsilon_{\mathrm{ell}}}{\sqrt\vartheta}\right]
       +O_\vartheta(\log a)+\log(1+M_*).
  \end{aligned}
  \label{eq:global-growing-initial-error}
\end{equation}
The constants $C_0,C_1$ are fixed independently of $\kappa,a,M_*$
and the prescribed spins.
A sufficient choice is
\begin{equation}
  \kappa_0=\min\left\{
    \frac1{100},\frac{\tau_0}{R_{\mathrm{init}}},
    \frac{\varepsilon_{\mathrm{ell}}^2}{4R_{\mathrm{init}}},
    \frac{\varepsilon_{\mathrm{ell}}^2}{R_{\mathrm{init}}}
                      e^{-2(C_0+C_1+1)}\right\}>0.
  \label{eq:global-kappa-choice}
\end{equation}
For fixed $0<\kappa<\kappa_0$ and sufficiently large $a$, one has
$U+1\le\tau_0a$, the approximation factor is at least two, and
the bracket in \eqref{eq:global-growing-initial-error} is less than
$-1$. Thus $\mathfrak E_{\mathrm{init}}\le1/8$ eventually in both
regimes. This completes the estimates in part~(i) of
Lemma~\ref{lem:global-step-estimates}. Equation~\eqref{eq:global-kappa-choice}
is a sufficient analytic range, not an optimal endpoint.
The bounded first-level weight changes only the required lower bound
on $a$. In particular, the same $R_0$, $R_{\mathrm{init}}$ and
$\kappa_0$ suffice. At fixed $B,M_*$ one common lower bound works
for every allowed first level; for a fixed positive $\kappa$ the
prescribed spins also require $\kappa a\ge\max_{J\in S}|J|$.

\subsection{Uniform estimates for a conditional tail step}
\label{app:global-tail}

Assume a finite prefix through $F_n$, with
$n\ge N_{\mathrm{init}}$, follows the prescribed schedule and
satisfies \eqref{eq:global-invariant}.
At the end of the initial stage, its processed endpoints obey
\[
  \max\{T,|J|\}\le f_{N_{\mathrm{init}},J}
        \le\max\{U+1,|J|\}.
\]
Indeed, the initial rows have endpoints in $[U,U+1]$, and
the remaining rows have $f_{N_{\mathrm{init}},J}=|J|>T$.
Every later endpoint can only increase. For the next cell write
$L=L_n\ge T\ge10$, $J'=j_n$,
$V\in[L+1/2,L+1]$ and $k=\lceil A_{\mathrm{tail}}(a+L)\rceil$.
The condition $T\ge10$ is automatic at fixed $B$ and holds
eventually at each fixed positive $\kappa$.

The least-endpoint rule also makes the selected $L_n$
nondecreasing during the tail stage. Hence every earlier tail
endpoint is at most $L_n+1$, while every initial endpoint is
at most $U+1$. The chosen cutoff therefore satisfies
\[
  \max\{B,V_0,\ldots,V_n\}
     \le\max\{U+1,L_n+1\}
     <2+\max\{U+2,L_n\}=B_n.
\]
Thus it protects all previously processed cells, as required in
\eqref{eq:construction-protecting-cutoff}. These geometric bounds
hold for every finite prefix of the prescribed schedule.

In the coordinate $\xi=\sqrt{\Ehat-|J'|}$, the interval length
satisfies $1/(4\sqrt{L+1})\le d\le1$. The reference lower bound
retains the factor $\Ehat^2-J'^2=\xi^2(\xi^2+2|J'|)$.
Consequently the current ordinary density obeys
\begin{equation}
  \frac{\dd\mu_n^\xi}{\dd\xi}
    \ge c\sqrt L\,\xi^2e^{8\sqrt{aL}}
          -CL^{-1/2}e^{7\sqrt{a(L+1)}}.
  \label{eq:global-tail-coordinate-bound}
\end{equation}
This includes a cell starting at $L=|J'|$.

For a nonnegative polynomial $\widetilde p$ of degree at most $k$
on this interval, \eqref{eq:moment-weighted-integral} and
\eqref{eq:moment-integral-variation} give
\[
  \int\xi^2\widetilde p\dd\xi
       \ge\frac{cd^2}{(k+1)^6}\int\widetilde p\dd\xi,
  \qquad
  \operatorname{Var}\widetilde p
       \le\frac{C(k+1)^4}{d}\int\widetilde p\dd\xi.
\]
The first bound keeps a definite fraction of the integral away
from the left edge, even when that edge is zero.
Inserting these inequalities into
\eqref{eq:global-tail-coordinate-bound} gives
\[
  \Lambda_n(p)\ge\frac1{\sqrt{L+1}}
      \left[\frac{ce^{8\sqrt{aL}}}{(k+1)^6}
                  -Ce^{7\sqrt{a(L+1)}}\right]
                         \int\widetilde p\dd\xi.
\]
For $L\ge10$, the difference of the two exponents is at least
$\sqrt{aL}/2$. The function
$\sqrt{aL}/2-6\log(a+L)$ is increasing in $L$ for $a>12$,
and tends to infinity at $L=10$ as $a\to\infty$.
After fixing $A_{\mathrm{tail}}$, the positive term therefore
dominates uniformly over every $L\ge T$.
Using the variation bound and the lower bound on $d$ yields
\begin{equation}
  \Lambda_n(p)\ge
  \frac{ce^{8\sqrt{aL_n}}}{(L_n+1)(k_n+1)^{10}}
       \operatorname{Var}_{[0,1]}p,
  \qquad p\in\Pi_{k_n},\quad p\ge0\text{ on }[0,1].
  \label{eq:global-tail-reserve}
\end{equation}
To show that the coefficient of $\operatorname{Var}_{[0,1]}p$ in
\eqref{eq:global-tail-reserve} is at least one uniformly over the tail
cells, use $k+1\le C_A(a+L)$.
Even the stronger denominator $(L+1)[C_A(a+L)]^{11}$ gives a
smooth lower bound whose logarithmic derivative is
\[
  \frac{4\sqrt a}{\sqrt L}-\frac1{L+1}-\frac{11}{a+L}
       \ge\frac{4\sqrt{aL}-12}{L}>0
\]
for large $a$. Its value at $L=10$ tends to infinity.
Thus one lower bound on $a$ makes this coefficient exceed one for
all tail cells, including when $T$ is fixed.

The positive coefficient of $\int\widetilde p\dd\xi$ above also
applies to $p=1$ at the candidate endpoint $V=L+1/2$.
It proves positive total weight there before the integer endpoint
is selected.
On the movable interval $[L+1/2,L+1]$ one has
$\Ehat-|J'|\ge1/2$. The leading term then dominates $H_a$
pointwise, and its integral tends uniformly to infinity.
The density is therefore positive almost everywhere there and the
weight increment exceeds one. This proves the endpoint conditions.
After that selection the total weight is a positive integer, so
\eqref{eq:global-tail-reserve} and
Proposition~\ref{prop:construction-integer-matching} give the
required replacement.

It remains to control the repair after the endpoint and matching
nodes have been chosen. From the invariant, the upper vacuum bound
and \eqref{eq:global-signed-residual-bound},
$\log(2+\|\nu_n\|_{\TV})\le C_*(a+L_n)$.
For sufficiently large $a$ in either regime, $U+2\le a$, and
Table~\ref{tab:construction-cell-parameters} gives
$B_n\le C_*(a+L_n)$.
For a cell of width at most one, the Bernstein ellipse with fixed
parameter two has height at most $3/8$ in the $\xi$ coordinate.
The complete tail bound \eqref{eq:inverse-tail-moment-bound}
therefore applies with cutoff $B_n$, giving
\begin{equation}
  |r_{B_n,J}[\nu_n](\Ehat)|
    \le C_*\|\nu_n\|_{\TV}\sqrt{B_n}\,2^{-k_n}
           e^{C_*B_n}\sqrt{\Ehat}\,e^{C_2\sqrt{\Ehat}}.
  \label{eq:global-short-completion}
\end{equation}
Its maximum relative to $H_a$ is at most
\begin{equation}
  \varepsilon_n\le
    C_*\|\nu_n\|_{\TV}B_ne^{C_*B_n}2^{-k_n}.
  \label{eq:global-tail-error-prefactor}
\end{equation}
The positive logarithmic costs are bounded by $C_*(a+L_n)$,
whereas the moment factor contributes
$-A_{\mathrm{tail}}(a+L_n)\log2$.
Choose $A_{\mathrm{tail}}$ large enough to leave room for
$4\log(2+L_n)+\log64$. This proves
\eqref{eq:global-tail-error} uniformly for every $L_n\ge T$.
Only the invariant, the cell width and the cutoff bound entered
these costs. Thus $A_{\mathrm{tail}}$ is independent of the finite
history, the first-level data, $B$, $\kappa$ and $a$.
After fixing it, one sufficiently large $a$ gives both the moment
inequality and this exterior bound for the entire range $L\ge T$.
Part~(ii) of Lemma~\ref{lem:global-step-estimates} follows.

The order of choices is now explicit. Fix the local analytic and
polynomial constants and the ellipse parameters; choose $R_0$ by
\eqref{eq:global-R0-choice}, then $R_{\mathrm{init}}$ using $C_N$,
and then the universal $A_{\mathrm{tail}}$.
For growing cutoffs choose $0<\kappa<\kappa_0$ from
\eqref{eq:global-kappa-choice}, then take $a$ sufficiently large,
including $b\ge2$, $T\ge10$, $U+2\le a$ and the initial-cell
condition $U+1\le\tau_0a$. The fixed spins additionally require
$\kappa a\ge\max_{J\in S}|J|$.
For fixed $B>0$, take the maximum of the finitely many lower bounds
on $a$, uniformly in $\delta\in[0,B)$ and in all allowed first-level
data of total multiplicity at most $M_*$.
All these bounds hold on a real half-line in $a$, and the tail
estimates already cover the full unbounded range of $L$. Thus a single
lower bound suffices for every subsequent step, without restricting
$a$ to a sequence. The fixed-cutoff threshold may depend on $B,M_*$;
the growing-cutoff threshold may depend on the fixed $\kappa,M_*$
and spin set. No threshold uniform as $\kappa\to0$ is asserted.

The admissible variations described in
Section~\ref{sec:construction-schedule} obey the same estimates.
A bounded offset in the initial degree changes the exponential
comparison only by bounded multiples and lower-order logarithmic
terms. For tail degrees between fixed admissible multiples of $a+L$,
the upper multiple bounds the degree-dependent factor in the denominator
of \eqref{eq:global-tail-reserve}, and the lower multiple controls the
exterior decay. The endpoint estimates
hold throughout each allowed window, and the residual bound uses
only exact total weight and matched moments. Thus the same argument
covers these choices after increasing the eventual lower bound on
$a$ when necessary.

\subsection{The thermal comparison measure}
\label{app:global-thermal-comparison}

For fixed $a$ and $\beta>0$, the upper bound in
\eqref{eq:global-leading-bounds} gives a finite constant
$C_{a,\beta}$ such that
\[
  e^{-\beta\Ehat}
       \bigl(\rho^{(1)}_{a,J}(\Ehat)+H_a(\Ehat)\bigr)
         \le C_{a,\beta}e^{-\beta\Ehat/2}.
\]
For $J\ne0$, writing $x=\Ehat-|J|$ bounds the remaining integral by
\begin{equation}
  \int_{|J|}^\infty e^{-\beta\Ehat/2}\dd\omega_J(\Ehat)
    \le\frac{e^{-\beta|J|/2}}{\sqrt{2|J|}}
            \int_0^\infty e^{-\beta x/2}x^{-1/2}\dd x
    =\sqrt{\frac\pi{\beta|J|}}e^{-\beta|J|/2}.
  \label{eq:global-spin-thermal-integral}
\end{equation}
This is summable over nonzero spins. The scalar integral is
$\int_T^\infty e^{-\beta\Ehat/2}\dd\Ehat/\Ehat<\infty$
because $T>0$. These bounds show that the positive measure with spin
components $(\rho^{(1)}_{a,J}+H_a)\dd\omega_J$, restricted to
$\Ehat\ge T$, has finite thermal norm after summing over all spins.
This is the comparison measure used in
Section~\ref{sec:global-limit}. The argument there then sums
the direct cell residuals over disjoint cells and the complete
exterior errors using \eqref{eq:global-initial-budget} and
\eqref{eq:global-tail-budget}. No uniform thermal norm as $a$ grows
is asserted.

\subsection{Near-extremal limits and modular ambiguities}
\label{app:global-near-extremal}

The convergence proved in Section~\ref{sec:global-limit} holds at
each fixed positive temperature. Here we compare it with the
large-spin, near-extremal limit used to diagnose the negative
vacuum density~\cite{BenjaminEtAl2019,AldayBae2020}.
Fix one admissible member of either gap family, with $a$ sufficiently
large and held fixed throughout this subsection. The constants below
may depend on this member; no uniform limit as $a\to\infty$ is
asserted. We use integrated spectral measures so that the continuous
finite stages and the final atomic spectrum can be compared directly.

For positive odd $J$, fix $0<\lambda<1/(4\pi^2)$ and set
\begin{equation}
  \begin{gathered}
    T_J=2\pi\sqrt{\frac{aJ}{2}},\qquad
    \epsilon_J=\lambda e^{-T_J},\qquad
    W_J=(J,J+\epsilon_J],\\
    D_J=\sqrt{2\lambda}\left(1-\frac{4\pi^2\lambda}{3}\right)
                   \frac{e^{T_J/2}}{\sqrt J}>0.
  \end{gathered}
  \label{eq:global-edge-windows}
\end{equation}
The interval $W_J$ is an interval of shifted energy $\Ehat$ in
the spin-$J$ row.

\begin{proposition}[Finite stages and the near-extremal limit]
\label{prop:global-edge-nonuniformity}
For each fixed finite stage $N$, the nonvacuum measure satisfies
\begin{equation}
  \frac{\mu_{N,J}(W_J)}{D_J}\longrightarrow-1
  \qquad(J\to\infty,\ J\text{ odd}).
  \label{eq:global-edge-finite-stage}
\end{equation}
The nonnegative stabilized measure $\mu_\infty$ consequently obeys
\begin{equation}
  \liminf_{\substack{J\to\infty\\J\text{ odd}}}
     \frac{(\mu_\infty-\mu_N)_J(W_J)}{D_J}\ge1
  \qquad\text{for every fixed }N,
  \label{eq:global-edge-tail}
\end{equation}
although $\|\mu_\infty-\mu_N\|_\beta\to0$ for every fixed
$\beta>0$.
\end{proposition}

\begin{proof}
Write $x=\Ehat-J$ and recall $D_a$ from
\eqref{eq:global-leading-density}. The $m=1$ vacuum numerator is
$2D_a(2J+x)D_a(x)$. At fixed $a>2$,
\[
  D_a(x)=4\pi^2x+O_a(x^2),\qquad
  D_a(2J+x)=\tfrac12e^{2T_J}(1+o_a(1)),
\]
where the second estimate is uniform for $0\le x\le1$.
For $m=2$, the Kloosterman sum is
$S(J,J';2)=(-1)^{J+J'}$. Combining the four vacuum seeds gives
\begin{equation}
  \begin{aligned}
    \rho^{(2)}_{a,J}(\Ehat)
      =(-1)^J\bigl\{&[C_a(\Ehat+J)+C_{a-2}(\Ehat+J)]\\
                    &\times[C_a(\Ehat-J)+C_{a-2}(\Ehat-J)]-4\bigr\},
    \qquad C_b(z)=\cosh(\pi\sqrt{bz}).
  \end{aligned}
  \label{eq:global-edge-second-image}
\end{equation}
For odd $J$, this is $-e^{T_J}(1+o_a(1))$ uniformly on
$0<x\le\epsilon_J$. The bracket estimate used in
Appendix~\ref{app:global-vacuum}, now summed over $m\ge3$,
bounds the remaining numerator, including the separate arithmetic
term, by
\begin{equation}
  |\rho^{(\ge3)}_{a,J}(J+x)|
    \le C_aJ\exp\!\left(\frac{4\pi}{3}\sqrt{a(J+1)}\right)
         +O_a(J)=o_a(e^{T_J}),\qquad 0<x\le1.
  \label{eq:global-edge-remainder}
\end{equation}
The last equality follows from $4\pi/3<\pi\sqrt2$ and
is uniform down to the edge. Dividing by
$\sqrt{x(2J+x)}$ and integrating the first two terms gives
\[
  \begin{aligned}
    \mu_{\mathrm{MWK},J}(W_J)
      &=\frac{8\pi^2e^{2T_J}\epsilon_J^{3/2}}{3\sqrt{2J}}
         -\frac{2e^{T_J}\sqrt{\epsilon_J}}{\sqrt{2J}}
         +o_a(D_J)\\
      &=-D_J(1+o_a(1)).
  \end{aligned}
\]
Here $\mu_{\mathrm{MWK},J}$ denotes the nonvacuum measure of the
vacuum modular completion. This reproduces its odd-spin negative
weight at the exponentially small window scale.

The difference between a fixed $\mu_N$ and the vacuum completion
is a finite sum of completions with bounded seed support, anchors
and scalar constants. For $J$ larger than all their seed cutoffs,
only induced continuous terms contribute on $W_J$.
For each such seed $\sigma$, the kernel bound
\eqref{eq:repair-compact-kernel-bound}, with its finite weighted
input norm from \eqref{eq:inverse-complete-seed-bound} and
\eqref{eq:anchor-size-bound}, gives
$|(R\sigma)_J(J+x)|\le K_\sigma\sqrt{J+x}$.
Its ordinary energy density is therefore at most
$K_\sigma/\sqrt x$ in absolute value, after increasing the
constant. It follows that
\begin{equation}
  |(\mu_N-\mu_{\mathrm{MWK}})_J(W_J)|
      \le K_N\sqrt{\epsilon_J},\qquad
  \frac{K_N\sqrt{\epsilon_J}}{D_J}
      =O_{a,N}(\sqrt J\,e^{-T_J})\longrightarrow0.
  \label{eq:global-edge-finite-repairs}
\end{equation}
This proves \eqref{eq:global-edge-finite-stage}.
For each fixed $J$, finite stabilization
\eqref{eq:global-finite-stabilization} identifies the total change
on $W_J$ with $(\mu_\infty-\mu_N)_J(W_J)$.
The independently established positivity of $\mu_\infty$ then
gives \eqref{eq:global-edge-tail}. Thermal convergence is
\eqref{eq:global-complete-thermal-sum}.
\end{proof}

For comparison, continuity of the same window evaluation in the
thermal norm gives only
\begin{equation}
  \frac{|\sigma_J(W_J)|}{D_J}
       \le\frac{e^{\beta(J+\epsilon_J)}}{D_J}\|\sigma\|_\beta.
  \label{eq:global-edge-thermal-comparison}
\end{equation}
Its prefactor grows like $\sqrt J\,e^{\beta J-T_J/2}$.
Thus thermal convergence supplies no bound tending to zero
uniformly over these normalized windows. Also, when the recursion
processes the first cell in a large-spin row, $W_J$ lies in the
band of exact repair. The small exterior estimate applies only
above $B_n$ and does not bound the direct cell replacement there.

\paragraph{Growth of the correcting seeds.}
Let $n(J)$ be the first step processing row $J>T$. Its cell is
$(J,J+w_J]$, where $1/2\le w_J\le1$, and its total integer
weight is $M_J$. Immediately before that step, the invariant
\eqref{eq:global-invariant} gives
\[
  \rho_{n(J),J}(J+x)
    =2D_a(2J+x)D_a(x)+O_a(e^{7\sqrt{a(J+x)}}),
    \qquad0<x\le1.
\]
Integration over the cell therefore yields
\begin{equation}
  \begin{aligned}
    M_J&=\frac{e^{2T_J}}{\sqrt{2J}}
          \left[\int_0^{w_J}\frac{D_a(x)}{\sqrt x}\dd x
                        +o_a(1)\right],\\
    \log M_J&=2\pi\sqrt{2aJ}-\tfrac12\log J+O_a(1).
  \end{aligned}
  \label{eq:global-edge-seed-growth}
\end{equation}
The integral in brackets is bounded above and below by positive
$a$-dependent constants for $w_J\in[1/2,1]$.
The error is relatively small because $7<2\pi\sqrt2$.
Here $M_J$ is the sum of multiplicities, with coincident energies
counted repeatedly; it is not the number of distinct levels.

Write $\sigma_n=\sigma_{B_n}[\nu_n]$ for the correcting seed
in \eqref{eq:moment-repair-decomposition}. The removed cell
measure, the inverse response $v_{\nu_n}$ and the anchor seed
$\theta_{B_n}$ have energy densities and no atoms, by
Appendices~\ref{app:finite-measure-inverse} and
\ref{app:threshold-anchor}. Consequently
the atomic part of $\sigma_n$ is exactly $Q_n$, and
\begin{equation}
  \sum_n |\sigma_n|\bigl(\{(\Ehat',J'):
                 J'=J,\ J<\Ehat'\le J+1\}\bigr)\ge M_J.
  \label{eq:global-edge-absolute-seeds}
\end{equation}
Any finite partial sum containing step $n(J)$ also has atomic
variation at least $M_J$ in this band: other seed atoms have
nonnegative weights, and density measures cannot cancel them.
Thus the natural seed family has no polynomial total-variation
growth bound. This conclusion concerns that family and its finite
partial sums; it requires neither convergence of the raw seed sum
nor a classification of alternative seed representations.

Alday and Bae assume at-most-polynomial seed-density growth when
extending the fixed-seed analysis to uncensored ambiguities, with
an additional qualification concerning twist accumulation
\cite[Secs.~4.1 and 4.3]{AldayBae2020}.
Equations~\eqref{eq:global-edge-seed-growth}--\eqref{eq:global-edge-absolute-seeds}
exhibit the growth of the present seeds, while
\eqref{eq:global-edge-tail} proves the failure of a uniformly
small remainder on the relevant shrinking windows. These are
compatible with summability of the complete spectral changes
\eqref{eq:global-complete-step}. Absolute convergence of those
completed changes does not justify an ungrouped integral over
the absolute correcting-seed family.

\paragraph{Threshold accumulation for the specified selector.}
For the lexicographic prescription of
Section~\ref{sec:construction-moments}, the first cell in row $J$
has normalized coordinate $s=\sqrt{(\Ehat-J)/w_J}$ and degree
$k_J=\lceil A_{\mathrm{tail}}(a+J)\rceil$.
Equation~\eqref{eq:global-tail-reserve} supplies a coefficient
\[
  b_J=\frac{c_*'e^{8\sqrt{aJ}}}{(J+1)(k_J+1)^{10}}
        \longrightarrow\infty
\]
such that $\Lambda_J(p)\ge b_J\operatorname{Var}_{[0,1]}p$
for every nonnegative $p\in\Pi_{k_J}$; $c_*'>0$ is the
constant supplied by that estimate. Let
$h_J=[64M_J(k_J+1)^3]^{-1}$ be the minimum allowed coordinate.
Lemma~\ref{lem:moment-inward-positivity} transfers this bound to
polynomials nonnegative on $[h_J,1]$, with variation coefficient
$b_J-1/16$. For sufficiently large $J$, the integer
$r_J=\lfloor b_J/2\rfloor$ is positive and
$b_J-1/16-r_J>1/2$. Testing the original bound with $p=s$ and
$p=1-s$ gives $M_J\ge2b_J$, so $M_J-r_J$ is also a positive
integer. The endpoint subtraction argument in
Appendix~\ref{app:moment-selector} therefore gives a feasible
design with $r_J$ copies of $h_J$ and the remaining points in
$[h_J,1]$. Lexicographic minimization forces at least that many
copies of $h_J$ in the selected vector.
The least shifted energy in the completed spin-$J$ row consequently
satisfies
\begin{equation}
  x_{\min,J}:=\min_{\text{spin }J}(\Ehat-J)
     =\frac{w_J}{4096M_J^2(k_J+1)^6}
     \asymp_a J^{-5}e^{-4\pi\sqrt{2aJ}}.
  \label{eq:global-edge-lex-accumulation}
\end{equation}
Later cells in this row have higher energies, and protection preserves
this first cell. Since $d_J\ge r_J$ and $k_J$ grows linearly
with $J$ at fixed $a$, the multiplicity obeys
\begin{equation}
  d_J\ge c_aJ^{-11}e^{8\sqrt{aJ}}
  \qquad(J\text{ sufficiently large}).
  \label{eq:global-edge-lex-multiplicity}
\end{equation}
These last two estimates concern the specified selector and fixed
$a$. They give neither the complete near-extremal counting function
nor the multiplicity of every constructed level. They do show
explicitly how chiral weights approach $(c-1)/24$ while total
dimensions diverge, consistently with local finiteness and the
twist-accumulation result of Ref.~\cite{PalQiaoRychkov2022}.

\section{Spectral asymptotics}
\label{app:spectral-asymptotics}

We prove the counting statements of Section~\ref{sec:entropy} from
the construction estimates in Appendices~\ref{app:integer-moments}
and~\ref{app:global-estimates}. The first two subsections treat
fixed-$B$ smooth counts and their reference expansion; the final subsection
treats the fixed positive-$\kappa$ regimes.
Each limit uses the parameter domain stated in its subsection.

\subsection{Transfer to the reference and the spin integral}
\label{app:entropy-transfer}

We prove the uniform comparison in
Lemma~\ref{lem:entropy-reference-comparison} and justify the replacement
of integer primary spins by continuous spin in
\eqref{eq:entropy-reference-inversion}. Fix $B>0$, $M_*\ge1$ and
$0<\lambda_-\le\lambda\le\lambda_+<\infty$, and put
$E=\lambda c$, $a=(c-1)/12$. Until the final paragraph of this
subsection, $\phi$ is the fixed nonnegative normalized kernel in
$C_c^\infty(\R)$, with $\supp\phi\subset[-R,R]$.
Use the schedule in Table~\ref{tab:construction-cell-parameters}
with fixed admissible auxiliary parameters. All constants below are
uniform in $\delta\in[0,B)$, the admissible first-level data with
total weight at most $M_*$, and every choice of atomic positions
in $\mathcal X$ at each step. They may depend on the fixed auxiliary
parameters, $B$, $M_*$, $\phi$, the energy-ratio interval and the
fixed order of the estimate.

\paragraph{The construction inputs.}
We first collect the properties supplied by Table~\ref{tab:construction-cell-parameters}
and Section~\ref{sec:global-control}. At fixed $B$, the scales $T<U$
in \eqref{eq:construction-initial-scales} are independent of $a$.
After the initial stage, write $F_J=f_{N_{\mathrm{init}},J}$ for
the processed endpoint in row $J$. The initial rows $|J|\le T$
end in $[U,U+1]$, while the remaining rows retain $F_J=|J|$.
The endpoint, width and degree rules therefore give
\begin{equation}
  \begin{gathered}
    \max\{T,|J|\}\le F_J\le\max\{U+1,|J|\},\\
    L_n\ge\max\{T,|j_n|\},\qquad
    \tfrac12\le V_n-L_n\le1,\qquad
    k_n\ge A_{\mathrm{tail}}a
       \quad(n\ge N_{\mathrm{init}}).
  \end{gathered}
  \label{eq:entropy-tail-contract}
\end{equation}
Each row is partitioned by its consecutive tail cells above $F_J$.
The positive packet $Q_n$ has the same total weight and the same
moments through degree $k_n$ as the entire current cell measure
$\mu_n^{\mathrm{cell}}$. The moments use the coordinate in
\eqref{eq:entropy-cell-coordinate}, while every packet is supported
inside its cell. The density estimate \eqref{eq:global-invariant}
places the current density between $\rho^{(1)}_{a,J}-H_a$ and
$\rho^{(1)}_{a,J}+H_a$ when each tail cell is processed. No derivative
bound for that density is used. Section~\ref{sec:global-limit}
proves convergence and finite-band stabilization for this construction.

The prescribed first level and all initial packets have energies
at most $U+1$. Let $M_{\mathrm{init}}$ denote their combined
multiplicity. Then
\begin{equation}
  M_{\mathrm{init}}\le\exp(C_{\mathrm{init}}\sqrt a),
  \label{eq:entropy-initial-mass}
\end{equation}
with $C_{\mathrm{init}}$ depending only on $B$, $M_*$ and the fixed
auxiliary parameters. To see this, the common initial repair cutoff
$2U+4$ ensures that each initial cell retains the continuous measure
of $F_0$ until it is processed. Matching the constant moment bounds
each packet's weight by the total variation of that cell. Summing
these weights, applying \eqref{eq:global-initial-variation} at
$W=U+1$, and adding $d_{\mathrm{first}}\le M_*$ gives
\eqref{eq:entropy-initial-mass}. The polynomial factor in $a$ is
absorbed by increasing $C_{\mathrm{init}}$ and the lower bound on $a$.

\paragraph{Reference growth and integrated spin thresholds.}
For this proof write
\begin{equation}
  S_a(Y)=4\pi\sqrt{aY},\qquad
  g_a(\Ehat,J)=\frac{\rho^{(1)}_{a,J}(\Ehat)}
                         {\sqrt{\Ehat^2-J^2}},\qquad \Ehat>|J|.
  \label{eq:entropy-reference-growth-functions}
\end{equation}
Thus $g_a$ is the ordinary reference density with respect to
$\dd\Ehat$; the scalar denominator is $\Ehat$.
For $z/a$ in a fixed positive compact interval, the exact factorization
\begin{equation}
  \begin{gathered}
    D_a(z)=\tfrac12 e^{A_z}
       (1-e^{C_z-A_z})(1-e^{-A_z-C_z}),\\
    A_z=2\pi\sqrt{az},\qquad C_z=2\pi\sqrt{(a-2)z},\qquad
    A_z-C_z=\frac{4\pi\sqrt z}{\sqrt a+\sqrt{a-2}}
  \end{gathered}
  \label{eq:entropy-reference-factorization}
\end{equation}
shows that $D_a(z)$ is bounded above and below by fixed positive
multiples of $e^{A_z}$. Globally, the upper bound
$D_a(z)\le e^{A_z}$ and
$\sqrt{1+x}+\sqrt{1-x}\le2-x^2/4$ give
\begin{equation}
  \rho^{(1)}_{a,J}(\Ehat)
    \le2\exp\!\left(S_a(\Ehat)
            -\frac{\pi\sqrt a}{2\Ehat^{3/2}}J^2\right),
       \qquad |J|\le\Ehat.
  \label{eq:entropy-reference-gaussian-bound}
\end{equation}

The measure at a nonzero-spin threshold is integrable. For $Z\ge T$,
\begin{equation}
  \begin{aligned}
    \sum_{J\in\Z}\int_{\max\{T,|J|\}}^Z\dd\omega_J
      &\le\log(Z/T)
          +2\sum_{1\le J<Z}\operatorname{arcosh}(Z/J)\\
      &\le\log(Z/T)+\pi Z.
  \end{aligned}
  \label{eq:entropy-integrated-spin-bound}
\end{equation}
A row whose lower limit exceeds $Z$ contributes zero. The last
inequality follows because $\operatorname{arcosh}(Z/x)$ decreases
in $x$ and its integral on $(0,Z)$ is $\pi Z/2$.
For continuous spin the corresponding identity is
$\int_{-\Ehat}^{\Ehat}(\Ehat^2-J^2)^{-1/2}\dd J=\pi$.
These integrated bounds will control the spin edges; no pointwise
bound on the sum of their singular denominators is needed.

Let $P_{\mathrm{all}}(X)=\sum_J\mu_{\infty,J}([|J|,X])$ be the
cumulative number of nonvacuum primaries. Any tail packet containing
an atom below $X$ lies wholly below $X+1$. Its weight is at most the
integral of $\rho^{(1)}_{a,J}+H_a$ over its entire cell. Disjointness
of the cells in each row therefore gives
\begin{equation}
  P_{\mathrm{all}}(X)\le M_{\mathrm{init}}
       +\sum_J\int_{\max\{T,|J|\}}^{X+1}
                (\rho^{(1)}_{a,J}+H_a)\dd\omega_J.
  \label{eq:entropy-primary-cumulative-envelope}
\end{equation}
The common envelope makes this valid even though the cells were
processed at different times.

For $X/c$ in a fixed positive compact interval, this implies the
sharper bound
\begin{equation}
  P_{\mathrm{all}}(X)\le Cc^{-1/2}e^{S_a(X)}.
  \label{eq:entropy-primary-cumulative-bound}
\end{equation}
To see the power of $c$, split the energy integral at $X/2$.
The lower part has an order-$c$ deficit in $S_a$. In the upper
part and $|J|\le\Ehat/2$, the denominator is of order $c$ and
\eqref{eq:entropy-reference-gaussian-bound} gives
$\sum_J e^{-\gamma J^2/c}\le C\sqrt c$.
The shift by one changes the bound only by a constant
factor in this extensive-energy range. The energy integral has bounded effective width since
$S_a'(\Ehat)$ has a positive lower bound on this interval.
The remaining spin region has an order-$c$ action deficit, and
\eqref{eq:entropy-integrated-spin-bound} bounds its integrated
measure by a polynomial in $c$. Finally $7<4\pi$ suppresses the
$H_a$ contribution, and \eqref{eq:entropy-initial-mass} controls the
initial weight. The same proof gives
\eqref{eq:entropy-primary-cumulative-bound} for the integer-spin
and continuous-spin reference measures. Their contribution below
$T$ is $\exp(O(\sqrt c))$, directly from $D_a$; it is finite at
zero and at every spin threshold.

\paragraph{Uniform polynomial approximation on a cell.}
Write $\alpha=\sqrt{L_n-|j_n|}$ and
$d=\sqrt{V_n-|j_n|}-\alpha$. In the normalized moment coordinate,
$\Ehat_n(z)=|j_n|+(\alpha+dz)^2$ for $0\le z\le1$. Then
\begin{equation}
  \begin{gathered}
    2\alpha d+d^2=V_n-L_n\le1,\\
    0\le\Ehat_n'(z)\le2,\qquad
    0\le\Ehat_n''(z)\le2,\qquad
    \Ehat_n^{(m)}(z)=0\quad(m\ge3).
  \end{gathered}
  \label{eq:entropy-cell-derivatives}
\end{equation}
For a parent center $Y$, let $h_n(z)=\phi(\Ehat_n(z)-Y)$.
Every fixed derivative of $h_n$ is uniformly bounded. The even
periodic function $h_n((1+\cos\theta)/2)$ has the same property.
Repeated integration by parts bounds its cosine coefficients by
$C_m\ell^{-m}$. Truncation after degree $k_n$ is a polynomial in
$z$, and summing the remaining coefficients proves
\begin{equation}
  \inf_{p\in\Pi_{k_n}}\|h_n-p\|_{\infty,[0,1]}
       \le C_P k_n^{-P}\le C_Pc^{-P}
       \quad\text{for every fixed }P\ge1.
  \label{eq:entropy-uniform-polynomial-error}
\end{equation}
For example, $P+2$ integrations by parts suffice. The constants
require only finitely many derivatives of the fixed kernel and the
fixed constant $A_{\mathrm{tail}}$. They are independent of $c$,
the cell, spin and parent center.

\paragraph{Positivity where the approximation is used.}
Suppose $Y/c$ lies in a positive compact interval. The cells that
meet $Y+\supp\phi$ lie in $[Y-R-1,Y+R+1]$. In the central rows
$|J|\le Y/2$, both $(\Ehat+J)/a$ and $(\Ehat-J)/a$ stay in
positive compact intervals. Equation~\eqref{eq:entropy-reference-factorization}
and the reference exponent give
\begin{equation}
  H_a(\Ehat)\le Ce^{-\eta c}\rho^{(1)}_{a,J}(\Ehat)
  \quad\text{on these cells},\qquad \eta>0.
  \label{eq:entropy-central-positivity}
\end{equation}
In particular, the current cell measure is positive for sufficiently
large $c$, so $\|\mu_n^{\mathrm{cell}}\|_{\TV}=M_n$ there.
Combining \eqref{eq:entropy-smooth-moment-bound} and
\eqref{eq:entropy-uniform-polynomial-error} bounds the replacement
error on each such cell by $C_PM_nc^{-P}$.

The number of relevant cells in a row is at most
$4R+6$, independently of $c$, the spin and the
parent center. On their union,
\eqref{eq:entropy-reference-gaussian-bound} bounds the ordinary
reference density by
$Cc^{-1}e^{S_a(Y)}e^{-\gamma J^2/c}$.
Summing their masses over central rows therefore gives at most
$Cc^{-1/2}e^{S_a(Y)}$. Replacing the current measures on those cells
by the reference costs only a relatively exponentially small error
by \eqref{eq:entropy-central-positivity}.
In the other rows, bound each packet by its whole-cell weight in
\eqref{eq:entropy-primary-cumulative-envelope}; the Gaussian deficit
and \eqref{eq:entropy-integrated-spin-bound} again give a relatively
exponentially small contribution. This last argument allows signed
current densities at spin openings.

To state the resulting primary comparison, set
\begin{equation}
  \begin{aligned}
    P_\phi^{\mathrm{at}}(Y)
       &=\sum_J\int_{[|J|,\infty)}\phi(\Ehat-Y)
                                      \dd\mu_{\infty,J}(\Ehat),\\
    P_\phi^{\mathbb Z}(Y)
       &=\sum_J\int_{|J|}^\infty\phi(\Ehat-Y)
                                      g_a(\Ehat,J)\dd\Ehat.
  \end{aligned}
  \label{eq:entropy-primary-smoothed-counts}
\end{equation}
Because a nonzero nonnegative continuous kernel is bounded below on
some interval, retaining that interval and $|J|\le\sqrt c$ also
gives a lower bound for the reference. Thus, uniformly on the stated
parent-energy range,
\begin{equation}
  \begin{gathered}
    C^{-1}c^{-1/2}e^{S_a(Y)}\le P_\phi^{\mathbb Z}(Y)
          \le Cc^{-1/2}e^{S_a(Y)},\\
    |P_\phi^{\mathrm{at}}(Y)-P_\phi^{\mathbb Z}(Y)|
          \le C_Pc^{-P-1/2}e^{S_a(Y)}.
  \end{gathered}
  \label{eq:entropy-primary-smooth-comparison}
\end{equation}
This supplies the scale needed to turn the cell errors into relative
errors. It does not assume that the kernel is positive throughout
its support.

\paragraph{Summing all descendant levels.}
Put $Y_0=E+1/12$ and $Y_N=Y_0-N$. For $0\le N\le Y_0/2$,
$Y_N/c$ remains in one positive compact interval. Concavity gives
\begin{equation}
  S_a(Y_0-N)\le S_a(Y_0)-\beta_0N,
  \qquad 0<\beta_0\le\inf S_a'(Y_0),
  \label{eq:entropy-descendant-action-bound}
\end{equation}
where $\beta_0$ is fixed uniformly in $c$ and $\lambda$.
Using \eqref{eq:entropy-descendant-product}, the sum of the primary
errors at these levels is bounded by
\begin{equation}
  \begin{aligned}
    \sum_{N\le Y_0/2}p_2(N)
       |P_\phi^{\mathrm{at}}(Y_N)-P_\phi^{\mathbb Z}(Y_N)|
      &\le C_Pc^{-P-1/2}e^{S_a(Y_0)}
                  \sum_{N\ge0}p_2(N)e^{-\beta_0N}\\
      &=C_Pc^{-P-1/2}e^{S_a(Y_0)}G(\beta_0).
  \end{aligned}
  \label{eq:entropy-summed-descendant-error}
\end{equation}
The multiplicities are summed with their suppression factor; the
number of allowed descendant levels introduces no growing factor.

For $N>Y_0/2$, every contributing nonvacuum parent has energy at most
$Y_0/2+R$. The cumulative bound
\eqref{eq:entropy-primary-cumulative-bound} controls all such parents.
Only $N\le Y_0+R$ can contribute. For any integer $M\ge0$,
\begin{equation}
  \sum_{N=0}^M p_2(N)
     \le \exp\!\left(\beta M+\frac{\pi^2}{3\beta}\right),
       \qquad \beta>0.
  \label{eq:entropy-descendant-cumulative-bound}
\end{equation}
Indeed, $\log G(\beta)=2\sum_{m\ge1}[m(e^{m\beta}-1)]^{-1}
\le\pi^2/(3\beta)$. Taking $\beta$ of order $c^{-1/2}$ when
$M=O(c)$ bounds the total multiplicity by $e^{O(\sqrt c)}$.
It cannot compensate the order-$c$ deficit between $S_a(Y_0/2+R)$
and $S_a(Y_0)$. The same bound applies to the reference tail.
The descendants of the initial packets contribute at most
$e^{O(\sqrt c)}$ by \eqref{eq:entropy-initial-mass}; the vacuum
module has the same bound since its relevant levels are $O(c)$.

The $N=0$ reference term supplies the lower bound, while the same
split gives the upper bound
$N_\phi^{\mathrm{ref}}(E)\asymp c^{-1/2}e^{S_a(Y_0)}$.
Equations~\eqref{eq:entropy-summed-descendant-error} and
\eqref{eq:entropy-descendant-cumulative-bound} therefore prove
Lemma~\ref{lem:entropy-reference-comparison}, including its
uniformity over the admissible first-level data and atomic choices
at fixed auxiliary parameters.
At each fixed $c$, the primary band through $Y_0+R$ stabilizes after
finitely many steps, so these estimates apply to the limiting
spectrum before the large-$c$ limit is taken.

\paragraph{Replacing the integer-spin sum.}
Only the explicit reference is differentiated in this step. Extend
$g_a$ to real $J$ and choose a fixed smooth cutoff $\chi$ equal to
one on $[-1/4,1/4]$ and zero outside $(-1/2,1/2)$. For an extensive
parent center $Y$, define
\begin{equation}
  F_Y(J)=\chi(J/Y)\int_{\R}\phi(s)g_a(Y+s,J)\dd s.
  \label{eq:entropy-poisson-test}
\end{equation}
On its support, factorization
\eqref{eq:entropy-reference-factorization} writes $g_a$ as
$e^{\Phi_a}$ times an amplitude whose $m$th spin derivative is
$O_m(c^{-1-m})$, where
\[
  \Phi_a(\Ehat,J)
       =2\pi\sqrt a\bigl(\sqrt{\Ehat+J}+\sqrt{\Ehat-J}\bigr).
\]
Here $\partial_J\Phi_a=O(|J|/c)$ and
$\partial_J^m\Phi_a=O_m(c^{1-m})$ for $m\ge2$.
Together with the Gaussian deficit in
\eqref{eq:entropy-reference-gaussian-bound}, the product rule gives
\begin{equation}
  \|\partial_J^m F_Y\|_{L^1(\R)}
       \le C_{m,\phi}e^{S_a(Y)}c^{-(m+1)/2}.
  \label{eq:entropy-poisson-derivative-bound}
\end{equation}
More explicitly, each differentiated exponential is bounded by
$c^{-m/2}$ times a fixed polynomial in $|J|/\sqrt c$, multiplying
$e^{S_a(Y)-\gamma J^2/c}$. Its spin integral contributes
$O(\sqrt c)$, while the ordinary density supplies $c^{-1}$.
Derivatives of $\chi$ are supported where the action deficit is
already of order $c$ and satisfy the same bound.

With Fourier convention
$\widehat F(\ell)=\int F(J)e^{-2\pi i\ell J}\dd J$,
Poisson summation and $m>1$ integrations by parts give
\begin{equation}
  \begin{aligned}
    \left|\sum_{J\in\Z}F_Y(J)-\int_{\R}F_Y(J)\dd J\right|
      &\le\sum_{\ell\ne0}|\widehat F_Y(\ell)|\\
      &\le\frac{2\zeta(m)}{(2\pi)^m}
                                  \|\partial_J^mF_Y\|_{L^1(\R)}.
  \end{aligned}
  \label{eq:entropy-poisson-error}
\end{equation}
Restoring the removed spin region on both sides costs only a
relatively exponentially small term, using
\eqref{eq:entropy-integrated-spin-bound} and its continuous-spin
counterpart. Taking $m\ge2P$ gives relative $O_P(c^{-P})$.
The estimates are uniform for $Y=Y_N$, $N\le Y_0/2$; the descendant
argument above controls the other levels.

Denote by $N_\phi^{\mathbb R}$ the reference count
\eqref{eq:entropy-reference-count} with the primary-spin sum replaced
by the real-spin integral. We have proved
\begin{equation}
  \frac{N_\phi^{\mathrm{ref}}(\lambda c)}
       {N_\phi^{\mathbb R}(\lambda c)}
        =1+O_P(c^{-P})\qquad\text{for every fixed }P\ge1.
  \label{eq:entropy-spin-integral-comparison}
\end{equation}
This is a comparison of smoothed counts, not a relative-error bound
for thermal transforms along a complex inversion contour.

\paragraph{The additional tails for a fixed Gaussian.}
A Gaussian kernel of fixed positive width is covered by the same
argument after controlling its noncompact tails. First,
\eqref{eq:entropy-primary-cumulative-envelope},
\eqref{eq:entropy-integrated-spin-bound} and
\eqref{eq:entropy-descendant-cumulative-bound} give the crude global
bound
\begin{equation}
  N_{\mathrm{nv}}(E_s+1/12\le X)
    \le C(1+a+X)^d
       \exp\!\left(S_a(X+1)+C\sqrt X+C\sqrt a\right),
       \qquad X\ge T,
  \label{eq:entropy-gaussian-global-envelope}
\end{equation}
for a fixed $d$. Here $N_{\mathrm{nv}}$ includes all descendants
of nonvacuum primaries. The vacuum contribution has the separate
bound $e^{C\sqrt{a+X+1}}$. These estimates follow by bounding all
parent energies by $X$ and all descendant multiplicities through
$X$; the reference obeys the same bounds. They establish convergence
of the Gaussian-weighted count even at arbitrarily high energies.

A sharper local bound is needed near the extensive center. Repeating
the descendant split with
\eqref{eq:entropy-primary-cumulative-bound} gives a full-state
cumulative bound $Cc^{-1/2}e^{S_a(Y)}$ when $Y/c$ is in a positive
compact interval. Dividing the bound for each unit energy interval
near $E+s$ by the reference scale at $E$, concavity yields
$Ce^{S_a'(Y_0)s}$, up to a uniform constant for the unit-width shift.
Choose a fixed small $\epsilon>0$. For
$\log c<|s|<\epsilon c$, multiplication by the Gaussian and summation
over unit intervals is consequently
$O(e^{-\gamma(\log c)^2})$ relative to that scale.
For $|s|\ge\epsilon c$, the global bound
\eqref{eq:entropy-gaussian-global-envelope} suffices: the negative
quadratic Gaussian exponent dominates its growth, beginning with
an order-$c^2$ loss at the near boundary. Physical energies are
bounded below by $-c/12$, so the negative tail is controlled as well.

It remains to compare the counts for $|s|\le\log c$. Multiply the
Gaussian by a smooth cutoff equal to one there and zero for
$|s|\ge2\log c$. Every fixed derivative of this truncated kernel
is uniformly bounded. The relevant cells now span $O(\log c)$
energy units, and their summed masses are at most a fixed power of
$c$ times the reference scale, by the local bound just proved.
Choosing the approximation order in
\eqref{eq:entropy-uniform-polynomial-error} larger by that fixed
power gives any prescribed relative $O_P(c^{-P})$ error.
The same adjustment in \eqref{eq:entropy-poisson-derivative-bound}
controls the spin integral. The descendant estimates remain uniform
on the enlarged compact energy range. Restoring the Gaussian tails
therefore proves both comparisons for each fixed Gaussian width;
no uniform limit as the width tends to zero or infinity is asserted.

\subsection{The exact reference transform and entropy coefficients}
\label{app:entropy-reference-expansion}

We now identify $N_\phi^{\mathbb R}$ with the BTZ integral and prove
its expansion with a uniform remainder at every fixed order. The
assumptions and uniform parameter range are those of
Appendix~\ref{app:entropy-transfer}. This part concerns the explicit
reference; it does not require inversion of the constructed atomic
partition function.

\paragraph{The primary integral.}
For real $\beta>0$, set $u=\Ehat+J$ and $v=\Ehat-J$. The Jacobian
is $1/2$, while $\sqrt{\Ehat^2-J^2}=\sqrt{uv}$. The factor two
in \eqref{eq:global-leading-density} thus gives
\begin{equation}
  Z_{\mathrm{ref,prim}}(\beta)
    =\left[\int_0^\infty u^{-1/2}D_a(u)e^{-\beta u/2}\dd u\right]^2.
  \label{eq:entropy-primary-transform-square}
\end{equation}
Putting $u=x^2$ and completing the square in the two exponential
terms of the hyperbolic cosine yields
\[
  \int_0^\infty u^{-1/2}\cosh(2\pi\sqrt{\alpha u})
                                 e^{-\beta u/2}\dd u
      =\sqrt{\frac{2\pi}{\beta}}e^{2\pi^2\alpha/\beta},
       \qquad \alpha\ge0.
\]
Subtract the terms with $\alpha=a$ and $a-2$ to obtain
\begin{equation}
  \int_0^\infty u^{-1/2}D_a(u)e^{-\beta u/2}\dd u
     =\sqrt{\frac{2\pi}{\beta}}e^{2\pi^2a/\beta}
                       (1-e^{-4\pi^2/\beta}).
  \label{eq:entropy-primary-gaussian-integral}
\end{equation}
Squaring proves \eqref{eq:entropy-continuous-reference-transform}.
These integrals converge locally uniformly for $\Re\beta>0$:
a fixed positive real part supplies exponential decay in $u$, whereas
$D_a(u)$ grows only exponentially in $\sqrt u$. The resulting
holomorphic functions therefore satisfy the same identities
throughout that half-plane.

\paragraph{Descendants and the physical energy shift.}
The physical energy of a level-$N$ descendant is
$\Ehat-1/12+N$. Its transform supplies exactly
$e^{\beta/12}G(\beta)$, with no vacuum module added to the reference.
For completeness, if $\tau=i\beta/(2\pi)$, the eta product gives
$\eta(\tau)^{-2}=e^{\beta/12}G(\beta)$.
Using $\eta(-1/\tau)=\sqrt{-i\tau}\,\eta(\tau)$ gives
\eqref{eq:entropy-eta-identity}. Consequently
\begin{equation}
  \begin{aligned}
    Z_{\mathrm{ref}}^{\mathbb R}(\beta)
      &=e^{4\pi^2a/\beta+\beta'/12}
                   (1-e^{-\beta'})^2G(\beta')\\
      &=e^{\pi^2c/(3\beta)}V(\beta'),
         \qquad \beta'=4\pi^2/\beta.
  \end{aligned}
  \label{eq:entropy-exact-thermal-normalization}
\end{equation}
Here $12a+1=c$, and the two factors $1-e^{-\beta'}$ remove the
level-one factors of the nonvacuum descendant product. This is the
exact identity with the perturbative BTZ contribution used in
\eqref{eq:entropy-reference-btz}.

The change of variable $s=E_s-E$ in the smoothing integral gives
$\int e^{-\beta E}\phi(E_s-E)\dd E=e^{-\beta E_s}K_\phi(\beta)$.
The sign in $K_\phi$ is therefore positive, as in
\eqref{eq:entropy-kernel-transform}. Thermal integrability justifies
interchanging the positive real-$\beta$ integrals and sums. The
resulting identity extends to $\Re\beta>0$ by absolute convergence.
Bilateral Laplace inversion then gives
\begin{equation}
  N_\phi^{\mathbb R}(E)
    =\int_{\gamma-i\infty}^{\gamma+i\infty}
       \frac{\dd\beta}{2\pi i}\,
       e^{\beta E}K_\phi(\beta)Z_{\mathrm{ref}}^{\mathbb R}(\beta)
    =N_\phi^{\mathrm{BTZ}}(E),\qquad \gamma>0.
  \label{eq:entropy-continuous-reference-inversion}
\end{equation}
The smoothed count is continuous, and the absolute integrability
on the inversion line established next justifies this pointwise
inversion. Together with \eqref{eq:entropy-spin-integral-comparison}
and Lemma~\ref{lem:entropy-reference-comparison}, this proves
\eqref{eq:entropy-smooth-matching}.

\paragraph{Control of the whole inversion line.}
We may choose $\gamma=\beta_*$ from
\eqref{eq:entropy-saddle-parameters}. On $\beta=\beta_*+it$, the
original-channel product and the primary integral give
\begin{equation}
  \begin{gathered}
    |G(\beta)|\le G(\beta_*),\qquad
    |1-e^{-4\pi^2/\beta}|\le\frac{4\pi^2}{|\beta|},\qquad
    |K_\phi(\beta)|\le K_\phi(\beta_*),\\
    |e^{\beta E}K_\phi(\beta)Z_{\mathrm{ref}}^{\mathbb R}(\beta)|
       \le C(1+|t|)^{-3}
           \exp\!\left(\beta_*E+
                   \frac{4\pi^2a\beta_*}{\beta_*^2+t^2}\right).
  \end{gathered}
  \label{eq:entropy-global-inversion-bound}
\end{equation}
For the second inequality, write $1-e^{-z}=z\int_0^1e^{-sz}\dd s$
and use $\Re z\ge0$. The factor $|\beta|^{-3}$ comes from the
primary prefactor and these two null-state factors. The same
argument on any fixed positive vertical line proves the absolute
integrability required in \eqref{eq:entropy-continuous-reference-inversion}.

For each fixed $\epsilon>0$, the part $|t|\ge\epsilon$ loses
an order-$c$ amount of the exponent relative to $t=0$.
Indeed, the loss in the $a$-dependent term is
$4\pi^2a t^2/[\beta_*(\beta_*^2+t^2)]$. The difference between its
value at $t=0$ and the physical-$c$ saddle action is only the bounded
shift $\pi^2/(3\beta_*)$.
Since $\beta_*$ ranges over a positive compact interval,
\eqref{eq:entropy-global-inversion-bound} makes the integral outside
$|t|<\epsilon$ exponentially smaller than the saddle contribution,
uniformly in $\lambda$. There are no other contributions of the
same exponential order on the inversion line.

\paragraph{The local expansion and its remainder.}
Define the phase only for this calculation by
\begin{equation}
  f_\lambda(\beta)=\lambda\beta+\frac{\pi^2}{3\beta},\qquad
  f_\lambda'(\beta_*)=0,\qquad
  f_\lambda''(\beta_*)=\frac{r^3}{12\pi}>0.
  \label{eq:entropy-saddle-phase}
\end{equation}
Near $\beta_*$ the amplitude is the analytic function $A_\phi$
of \eqref{eq:entropy-smooth-amplitude}, with $A_\phi(\beta_*)>0$.
Its derivatives of every fixed order are uniformly bounded on a
common neighborhood of the compact set of saddles. Along the
vertical line,
\[
  \Re\bigl[f_\lambda(\beta_*+it)-f_\lambda(\beta_*)\bigr]
     =-\frac{\pi^2t^2}{3\beta_*(\beta_*^2+t^2)}.
\]
The region $c^{-2/5}\le|t|<\epsilon$ is therefore suppressed by
$e^{-\gamma c^{1/5}}$. In the remaining region put
$x=\sqrt{c f_\lambda''(\beta_*)}\,t$ and Taylor-expand the amplitude
and the phase beyond its quadratic term. For each fixed requested
order, sufficiently many Taylor terms bound the remainder by the
corresponding inverse power of $c$ times an integrable polynomial
in $x$ multiplying $e^{-\gamma x^2}$. Odd terms integrate to zero
on the symmetric interval. Extending that interval to the real line
adds an exponentially small error. This proves an expansion in
integer powers of $c^{-1}$ with a uniform remainder at each finite
order, rather than only a formal saddle expansion.

The resulting differential operators can be written explicitly.
Let $\mathfrak D=\beta\partial_\beta$ and define
\begin{equation}
  \begin{aligned}
    \mathcal L_0 A&=A,\\
    (\mathcal L_m A)(\beta)
       &=\frac{(-1)^m}{m!}
          \left(\frac{3\beta}{16\pi^2}\right)^m
          \left\{\prod_{j=1}^m
            \bigl[4(\mathfrak D+1)^2-(2j-1)^2\bigr]A\right\}(\beta),
          \qquad m\ge1.
  \end{aligned}
  \label{eq:entropy-saddle-operators}
\end{equation}
The displayed power of $\beta$ multiplies the result after the
differential operators have acted. For every fixed $P\ge0$,
\begin{equation}
  N_\phi^{\mathrm{BTZ}}(\lambda c)
    =\frac{\sqrt6\,e^{\pi cr/3}}{\sqrt c\,r^{3/2}}
       \left[\sum_{m=0}^P
            \frac{(\mathcal L_m A_\phi)(\beta_*)}{c^m}
             +O_P(c^{-P-1})\right].
  \label{eq:entropy-reference-all-orders}
\end{equation}

To verify \eqref{eq:entropy-saddle-operators}, the coefficient at
order $m$ is a linear differential functional of $A$ of order at
most $2m$, by the preceding Gaussian expansion. It is enough to
evaluate it on $A(\beta)=\beta^p$, $0\le p\le2m$.
The positively oriented circle $|\beta|=\beta_*$ has the same local
saddle and gives the exact integral
\[
  \frac{1}{2\pi i}\oint_{|\beta|=\beta_*}
      \beta^p e^{c f_\lambda(\beta)}\dd\beta
     =\beta_*^{p+1}I_{p+1}(z),\qquad
       z=\frac{2\pi^2c}{3\beta_*},
\]
where, for integer $\nu$,
$I_\nu(z)=(2\pi)^{-1}\int_{-\pi}^{\pi}
 e^{z\cos\theta+i\nu\theta}\dd\theta$.
Two integrations by parts give
$z^2 I_\nu''+zI_\nu'-(z^2+\nu^2)I_\nu=0$.
Substituting $e^zz^{-1/2}$ times an inverse-$z$ series into this
identity gives the coefficient ratio
$-[4\nu^2-(2m-1)^2]/(8m)$ at order $m$.
The leading Gaussian integral is $e^z/\sqrt{2\pi z}$, so the
relative coefficient on the monomial is
\[
  \frac{(-1)^m}{m!}
      \left(\frac{3\beta_*}{16\pi^2}\right)^m
       \prod_{j=1}^m[4(p+1)^2-(2j-1)^2].
\]
Since $\mathfrak D\beta^p=p\beta^p$, this agrees with
\eqref{eq:entropy-saddle-operators}. The circle integral is used
only to evaluate the local differential coefficients. The global
remainder for the actual amplitude follows from
\eqref{eq:entropy-global-inversion-bound} and the vertical-contour
argument above.

\paragraph{Passing directly to the entropy coefficients.}
Equation~\eqref{eq:entropy-reference-all-orders} gives the leading
prefactor in \eqref{eq:entropy-smooth-count-expansion}; its logarithm
gives \eqref{eq:entropy-smooth-constant}. For the higher coefficients
one can work directly with $C_{m,\phi}$, without introducing a second
set of scalar coefficients for the count. Equivalently, as a formal
power series in $t$,
\begin{equation}
  \log\!\left[
    \frac{1}{A_\phi(\beta_*)}
      \sum_{m\ge0}t^m(\mathcal L_m A_\phi)(\beta_*)\right]
       =\sum_{m\ge1}C_{m,\phi}(\lambda)t^m.
  \label{eq:entropy-formal-logarithm}
\end{equation}
For any fixed coefficient only finitely many terms are needed.
Differentiating this identity with respect to $t$ and multiplying
by the series inside the logarithm gives the recursion
\begin{equation}
  C_{m,\phi}
    =\frac{(\mathcal L_m A_\phi)(\beta_*)}{A_\phi(\beta_*)}
       -\frac{1}{mA_\phi(\beta_*)}
          \sum_{j=1}^{m-1}j C_{j,\phi}
                   (\mathcal L_{m-j}A_\phi)(\beta_*).
  \label{eq:entropy-coefficient-recursion}
\end{equation}
In particular,
\[
  (\mathcal L_1 A)(\beta)
     =-\frac{3\beta}{16\pi^2}
                  [3A(\beta)+12\beta A'(\beta)+4\beta^2A''(\beta)],
\]
which gives \eqref{eq:entropy-smooth-first-coefficient} with the
signs appropriate to the vertical contour. The positivity of
$A_\phi(\beta_*)$ is uniform on the compact saddle interval, so
finite-order Taylor expansion of the logarithm preserves the
uniform remainder in \eqref{eq:entropy-reference-all-orders}.
Using \eqref{eq:entropy-smooth-matching} at order $P+1$ transfers
that expansion through $c^{-P}$, with remainder $O(c^{-P-1})$, to
the constructed entropy. This completes the detailed proof of
Proposition~\ref{prop:btz-entropy}.

For a fixed Gaussian the transform $K_\phi$ is entire and bounded
on each fixed positive vertical line, and its local derivatives
satisfy the same compact-saddle bounds. The inversion estimate and
the coefficient calculation therefore apply after the tail
comparison in Appendix~\ref{app:entropy-transfer}.
All coefficient and remainder statements concern each fixed finite
order. They assert neither convergence of the series nor uniformity
as the order grows. They identify the smoothed microcanonical
count with a single perturbative saddle, without assuming canonical
BTZ dominance or identifying the full atomic thermal transform
with that saddle.

\subsection{Smoothed counts with a macroscopic gap}
\label{app:entropy-positive-kappa}

We prove the statements in Section~\ref{sec:entropy-positive-kappa}
for the explicit schedule of Section~\ref{sec:construction-schedule}.
Fix all construction parameters, including $R_0$, $R_{\mathrm{init}}$
and the tail-degree constants, subject to Appendix~\ref{app:global-estimates};
then fix a sufficiently small $\kappa>0$. Write
$b=\kappa a$, $U=\vartheta a$, $d=\sqrt\kappa$ and
$D=\sqrt\vartheta$, where $\vartheta=R_{\mathrm{init}}\kappa$.
The degree in every initial row is $k=\lfloor aD\rfloor$.
The near-gap results use the compact lexicographic selector.
The prescribed first-level data have total weight at most fixed $M_*$
and spins $|J|\le J_*$, with $J_*$ fixed and $b\ge J_*$.
Constants and the lower threshold on $a$ may depend on these bounds
and the fixed construction parameters.

\paragraph{The vacuum coefficient and integer sampling.}
Let $v_2(N)$ be the coefficient of $e^{-\beta N}$ in
$V(\beta)$. The exact relation to the descendant coefficients is
\[
  v_2(N)=p_2(N)-2p_2(N-1)+p_2(N-2),
\]
where negative arguments have zero coefficient. The two-color
Rademacher formula \cite[Theorem~4 and Lemma~2]{PribitkinWilliams2018}
gives the principal term
\begin{equation}
  p_2(N)=\frac{\pi}{6(N-1/12)}
      I_2\!\left(2\pi\sqrt{\frac{N-1/12}{3}}\right)
      +O\!\left(N^C e^{\pi\sqrt{N/3}}\right).
  \label{eq:entropy-partition-principal-bessel}
\end{equation}
In that series, a denominator $q$ divides the Bessel argument by
$q$, and the finite phase sum has modulus at most $q$.
For $2\le q\le\sqrt N$ this bounds the terms by a polynomial times
the half-sized principal exponential. For $q>\sqrt N$, the
small-argument estimate $I_2(z)=O(z^2)$ makes the remaining
$O(q^{-2})$ bound summable. This justifies the error in
\eqref{eq:entropy-partition-principal-bessel} at coefficient level.

Apply the large-argument Bessel expansion
\cite[Eq.~10.40.1]{NISTDLMF} to its principal term before taking
the exact second difference. The arithmetic error stays
exponentially smaller, whereas the leading second difference
multiplies the $p_2$ asymptotic by $\pi^2/(3N)$.
Consequently,
\begin{equation}
  \begin{aligned}
    p_2(N)&=\frac{N^{-5/4}}{4\,3^{3/4}}
        e^{2\pi\sqrt{N/3}}\bigl(1+O(N^{-1/2})\bigr),\\
    v_2(N)&=\frac{\pi^2N^{-9/4}}{4\,3^{7/4}}
        e^{2\pi\sqrt{N/3}}\bigl(1+O(N^{-1/2})\bigr).
  \end{aligned}
  \label{eq:entropy-kappa-partition-coefficients}
\end{equation}
An expansion of the principal Bessel term through any fixed order
justifies the cancellation; a relative error estimate for $p_2$
alone would not suffice.

At a dimension center $t=cx$ an extensive distance below
$\Delta_g=(1+\kappa)a$, a compact kernel sees only
$\sum_{N\ge0}v_2(N)\phi(N-t)$. Every term has $N=t+O(1)$.
Its common relative coefficient error is $O(t^{-1/2})$, so the
integer sum remains exactly the periodization $P_\phi(t)$ in
\eqref{eq:entropy-kappa-periodization}. This proves
\eqref{eq:entropy-kappa-vacuum-count}. For the chosen bump with
$w>1/2$, some integer is always in its positive interior;
continuity and periodicity give a positive minimum for $P_\phi$.
Its logarithm is therefore bounded, as used in the entropy formula.

For a fixed Gaussian the same expansion holds on $|N-t|\le t/2$:
the Taylor remainders are polynomial in $|N-t|$ and summable against
a slightly wider Gaussian. The bound
$v_2(N)\le e^{C\sqrt{N+1}}$ controls the complement by
$e^{-\eta t^2}$. Nonvacuum leakage is also negligible. Indeed,
complete tail cells, the integrated spin bound and the descendant
bound give, for a dimension cutoff $Z\ge\Delta_g$,
\begin{equation}
  \#\{\text{nonvacuum states with }\Delta\le Z\}
     \le C(1+a+Z)^C
          e^{Ca+C\sqrt{a(1+Z)}+C\sqrt{1+Z}}.
  \label{eq:entropy-kappa-global-envelope}
\end{equation}
The $Ca$ term includes all initial packets, whose weight is at most
$\operatorname{poly}(a)e^{4\pi aD}$ by
\eqref{eq:global-initial-variation}. For the tail part, a packet
meeting the cutoff ends at most one unit above it; sum its full
weight with \eqref{eq:entropy-integrated-spin-bound} and then use
\eqref{eq:entropy-descendant-cumulative-bound}.
All nonvacuum states lie at offset at least $\eta c$ from $t$.
Multiplying the bound on successive unit bins by the Gaussian gives
$e^{-\eta'c^2}$, including arbitrarily large positive offsets.
Thus the vacuum formula applies to the full Gaussian count too.

\paragraph{Stability under changes of the first level.}
For an initial row
$|J|\le T$, let $\Lambda_{a,J,V}$ be its continuous moment functional
on the full normalized interval $[0,1]$, with $V\in[U,U+1]$.
Let $\Lambda^{(1)}_{a,J,V}$ use the leading vacuum density on the
same interval.

\begin{lemma}[Stability of the initial moment forms]
\label{lem:entropy-first-level-stability}
With the preceding parameter choices, there are constants
$\eta_*,C_*>0$ such that, uniformly over the allowed first-level
data, initial rows and endpoints,
\begin{equation}
  \begin{gathered}
  |\Lambda_{a,J,V}(p)-\Lambda^{(1)}_{a,J,V}(p)|
       \le\epsilon_a\Lambda^{(1)}_{a,J,V}(p),\\
  \epsilon_a=C_{\kappa,M_*,J_*}a^{C_*}e^{-\eta_*aD},
  \qquad p\in\Pi_k,\quad p\ge0\text{ on }[0,1].
  \end{gathered}
  \label{eq:entropy-first-level-form-stability}
\end{equation}
The sufficient range \eqref{eq:global-kappa-choice} is unchanged.
\end{lemma}

\begin{proof}
Relative to the scalar-unit initialization, set
$\eta_a=\nu_{\mathrm{first}}-\delta_{(b,0)}$ and
$v=\mathcal I_b[\eta_a]$. Linearity gives the complete change in
the continuous response,
\begin{equation}
  \delta\rho_J
     =-t(v)\delta_{J,0}\rho_b^\psi
       +\bigl(R[v-t(v)\theta_b]\bigr)_J.
  \label{eq:entropy-first-level-response}
\end{equation}
The first term is the direct anchor density on $[2b,3b]$.
Since $\|\eta_a\|_{\TV}\le M_*+1$, the inverse, anchor and kernel
estimates bound this response by
$C(1+M_*)b^{C_*}\sqrt{\Ehat}\,e^{C_{\mathrm{ref}}b}$.
The original initialization has the larger exponential rate
\[
  \zeta=4\pi\sqrt\kappa+C_{\mathrm{ref}}\kappa
     <7\sqrt{R_0\kappa}\le\frac74D.
\]
Integration through a spin opening preserves this rate, since
$\sqrt{\Ehat}\,\dd\omega_J
\le(\Ehat-|J|)^{-1/2}\dd\Ehat$.
Thus the full initialization response in each initial cell has
total variation at most $C_{\kappa,M_*}a^{C_*}e^{a\zeta}$.

Let $S_p$ be the maximum of $p$ on the upper quarter of the
square-root interval. Chebyshev extrapolation gives
$\|p\|_{\infty,[0,1]}\le C14^kS_p$. In that quarter,
$\sqrt{(\Ehat-|J|)/a}\ge(3/4)\sqrt{15/16}\,D$.
The leading density and the polynomial derivative bound therefore give
\[
  \Lambda^{(1)}_{a,J,V}(p)
     \ge c_\kappa(k+1)^{-3}S_p
            e^{3\pi\sqrt{15/16}\,aD}.
\]
Here \eqref{eq:moment-polynomial-integral} is applied on the upper
quarter. Its polynomial cost does not change the exponential comparison.
The response and the higher-denominator vacuum remainder are smaller
as moment functionals because both
\[
  3\pi\sqrt{15/16}-\tfrac74-\log14>0,
  \qquad
  3\pi\sqrt{15/16}-2\pi-\log14>0.
\]
Choose $\eta_*$ smaller than these two margins. The remaining
polynomial factors give \eqref{eq:entropy-first-level-form-stability}.
This comparison does not require the signed cell density to be
pointwise positive.
\end{proof}

\paragraph{Endpoint capacity and the integer selector.}
The endpoint ceiling changes an initial row's integrated weight by
less than one. Its density on $[U,U+1]$ is exponentially large,
uniformly in the initial rows, so
$0\le V_J-U\le C_\kappa a^C e^{-\eta aD}$.
The moment comparison in
\eqref{eq:entropy-first-level-form-stability} uses this same actual
endpoint $V_J$ on both sides. For a fixed row its cell weight
$M_{a,J}$ has logarithm $4\pi aD+O(\log a)$.
Thus $h_J=[64M_{a,J}(k+1)^3]^{-1}$ is exponentially small, and the
energy at $h_J$ is $b+\xi_{a,J}$ with
$0<\xi_{a,J}=O(ah_J)\to0$.

Write $\Lambda=\Lambda_{a,J,V_J}$ and use the normalized cell
coordinate $s$. The initial polynomial inequality and
Lemma~\ref{lem:moment-inward-positivity} give
\begin{equation}
  \Lambda(p)\ge\beta_{h,J}\operatorname{Var}_{[h_J,1]}p,
  \qquad p\in\Pi_k,\quad p\ge0\text{ on }[h_J,1],
  \qquad \beta_{h,J}\ge a^{-C}e^{\eta aD}.
  \label{eq:entropy-inward-moment-bound}
\end{equation}
The inward restriction costs only $1/16$ in the coefficient of
\eqref{eq:global-initial-reserve}. Define $\lambda_{a,J}$ as the
largest real weight that can be subtracted at $h_J$ while leaving
a positive functional on this degree-$k$ cone. The moments still
come from the entire original cell.
Lemma~\ref{lem:moment-endpoint-multiplicity} proves that the actual
lexicographic multiplicity satisfies
\begin{equation}
  \left\lfloor(1-\beta_{h,J}^{-1})\lambda_{a,J}\right\rfloor
       \le m_{a,J}\le\lambda_{a,J}.
  \label{eq:entropy-first-level-capacity}
\end{equation}
In particular its relative error is at most
$\beta_{h,J}^{-1}+\lambda_{a,J}^{-1}$.
The elementary endpoint packing from the same lemma already gives
$\lambda_{a,J}\ge\lfloor\beta_{h,J}/2\rfloor$.
Hence rounding and the integer constraint change each fixed-row
capacity by a relatively exponentially small amount.

Here is the quadratic-form expression for that real capacity.
Let $v_n(s)=(1,s,\ldots,s^n)^T$,
$H_n=\Lambda[v_n v_n^T]$ and
$H_n^+=\Lambda[(1-s)v_n v_n^T]$. Then
\begin{equation}
  \lambda_{a,J}=
  \begin{cases}
    [v_n(h_J)^TH_n^{-1}v_n(h_J)]^{-1},& k=2n,\\
    [(1-h_J)v_n(h_J)^T(H_n^+)^{-1}v_n(h_J)]^{-1},& k=2n+1.
  \end{cases}
  \label{eq:entropy-capacity-quadratic-forms}
\end{equation}
For even degree, the Markov--Luk\'acs theorem
\cite[Corollary~3.24]{Schmudgen2017} represents a nonnegative interval
polynomial as a sum of squares plus $(s-h_J)(1-s)$ times a sum of squares.
Endpoint
subtraction affects only the first form. For odd degree the two
factors are $s-h_J$ and $1-s$, so only the second form changes.
The maximal permitted rank-one subtraction gives
\eqref{eq:entropy-capacity-quadratic-forms}. These interval
factorizations follow by mapping to $\cos\theta$ and factoring a
nonnegative trigonometric polynomial as a squared modulus.
Strict positivity in \eqref{eq:entropy-inward-moment-bound}
ensures that the matrices being inverted are positive definite.
The active test polynomials $q^2$ and $(1-s)q^2$ are nonnegative
on the original $[0,1]$; thus the relative comparison in
\eqref{eq:entropy-first-level-form-stability} applies to them.
The other interval form is controlled by the inward inequality,
without discarding any part of the input measure.

\paragraph{The scalar capacity: exponential weight.}
In row zero put $\Ehat=ax^2$, $d\le x\le D$.
The reference measure, including its $\dd\Ehat/\Ehat$ factor, is
\[
  \rho^{(1)}_{a,0}(\Ehat)\dd\omega_0
     =g(x)e^{4\pi ax}[1+O(a^{-1})]\dd x,
  \qquad g(x)=\frac{(1-e^{-2\pi x})^2}{x}.
\]
The error is uniform on this fixed positive interval. Change to
$y=4\pi a(D-x)$, and set
$Y=4\pi a(D-d)$ and $C_a=e^{4\pi aD}/(4\pi a)$.
The measure becomes $C_a g(D-y/(4\pi a))e^{-y}\dd y$ on $[0,Y]$.
For the moment omit $g$. The orthonormal polynomials for
$y^\alpha e^{-y}\dd y$ on $[0,\infty)$ are
\[
  p_m^{(\alpha)}(y)
      =\frac{(-1)^m L_m^\alpha(y)}
                    {\sqrt{\Gamma(m+\alpha+1)/m!}},
  \qquad \alpha=0,1.
\]
Their positive leading coefficients fix the sign convention below.
The generating function and orthogonality used here are the
standard Laguerre identities \cite[Secs.~18.3, 18.9, 18.12]{NISTDLMF}.

For $z>4$ define $r=r(z)\in(0,1)$ by
$z=2+r+r^{-1}$ and $\Psi(z)=1+r-\log r$.
The coefficient integral from
$\sum_mL_m^\alpha(y)t^m=(1-t)^{-\alpha-1}e^{-yt/(1-t)}$
has a saddle at $t=-r$. On $t=-re^{i\theta}$ the unique maximum
of its real exponent is at $\theta=0$, with value $1+r$ and
negative second derivative $(1-r)/(1+r)$. Thus
\begin{equation}
  L_n^\alpha(nz)
    =\frac{(-1)^ne^{n\Psi(z)}}{\sqrt{2\pi n}}
       (1+r)^{-\alpha-1/2}(1-r)^{-1/2}
       [1+O(n^{-1})].
  \label{eq:entropy-laguerre-saddle}
\end{equation}
All constants are uniform on compact subsets of $z>4$.

The half-line measure can be used despite the finite interval.
The same coefficient contour bounds
$|L_m^\alpha(y)|\le C r^{-m}e^{yr/(1+r)}$ for $m\le n$ and
$y\ge nz$. Writing
$\mathcal K_n^{(\alpha)}(y,y')=\sum_{m=0}^np_m^{(\alpha)}(y)
 p_m^{(\alpha)}(y')$, integration gives
\begin{equation}
  \int_{nz}^\infty\mathcal K_n^{(\alpha)}(y,y)
                          y^\alpha e^{-y}\dd y
       \le n^C e^{-n\mathcal J(z)},
  \quad \mathcal J(z)=z-2\Psi(z)>0.
  \label{eq:entropy-laguerre-projection-tail}
\end{equation}
Indeed, $\mathcal J(4)=0$ and
$\mathcal J'(z)=(1-r)/(1+r)>0$.
The left side bounds the operator norm of the discarded positive
Gram form, so truncation changes capacities only relatively
exponentially little. A bounded multiplier preserves this bound.

The recurrence for the $p_m^{(\alpha)}$ has diagonal
$2m+\alpha+1$ and off-diagonal $\sqrt{m(m+\alpha)}$.
Choose any $q\in(r(z),1)$ with $z>2+q+q^{-1}$.
The ratio recurrence then bounds
$p_{m-1}^{(\alpha)}(nz)/p_m^{(\alpha)}(nz)\le q$ for $m\le n$
and large $n$. Each fixed number of ratios down from degree $n$
tends to $r(z)$ by \eqref{eq:entropy-laguerre-saddle}.
The normalized reversed evaluation vector therefore converges in
$\ell^2$ to $(1,r,r^2,\ldots)$, with a geometric bound on its tail.
Consequently $\mathcal K_n^{(\alpha)}(nz,nz)$ is asymptotic to
$p_n^{(\alpha)}(nz)^2/(1-r^2)$.

For $k=2n$ the capacity without $g$ is
$C_a/\mathcal K_n^{(0)}(Y,Y)$, hence
\[
  \lambda_{2n}^{\exp}
    =\frac{n}{2a}(1-r^2)^2
         e^{4\pi aD-2n\Psi(Y/n)}[1+o(1)].
\]
For $k=2n+1$ it is
$C_a/[Y\mathcal K_n^{(1)}(Y,Y)]$.
The additional factor relative to the even expression at the same
$n$ is $(n+1)(1+r)^2/Y=r[1+o(1)]$.
Equivalently replace $n$ by $k/2$ in the even expression.
There is no remaining parity-dependent constant.
Since $z\Psi'(z)=1+r$, differentiation at fixed $Y$ gives
$\partial_k[-k\Psi(2Y/k)]=\log r$.
Putting $k=aD-\{aD\}$ yields
\begin{equation}
  \log\lambda_k^{\exp}
     =aH-\{aD\}\log r_{\mathrm{gap}}
        +\log\!\left[\frac D4(1-r_{\mathrm{gap}}^2)^2\right]+o(1),
  \label{eq:entropy-unweighted-capacity}
\end{equation}
where
\begin{equation}
  \begin{gathered}
    z=8\pi(1-d/D),\qquad
    r_{\mathrm{gap}}=\frac{2}{z-2+\sqrt{z(z-4)}},\\
    H=D(4\pi-1-r_{\mathrm{gap}}+\log r_{\mathrm{gap}}),\qquad
    4\pi\sqrt\kappa<H<4\pi\sqrt\vartheta.
  \end{gathered}
  \label{eq:entropy-kappa-rate-constants}
\end{equation}
To verify the stated bounds, $R_{\mathrm{init}}\ge16R_0>80$
implies $4<z<8\pi$ and $r_{\mathrm{gap}}>1/(8\pi)$.
Thus $D[4\pi-2-\log(8\pi)]<H<4\pi D$; the lower bound exceeds
$4\pi d$ because $D/d>\sqrt{80}$.
This also proves the bounds on $\epsilon_H$ in
\eqref{eq:entropy-kappa-nonbtz-endpoint}.

\paragraph{The scalar capacity: the multiplier.}
A two-sided comparison with $g$ would not determine its bounded
contribution to the logarithm. We compute the required inverse-form
limit. Extend $g(D-y/(4\pi a))$ by the constant $g(d)$ for $y>Y$.
The extension has a common strictly positive lower bound and a
finite upper bound. In the Laguerre basis, reverse indices about
$n$ and scale the Jacobi matrix by $n$. Regard the reversed operator
as acting on $\ell^2(\Z)$ by adjoining a decoupled positive operator
on the missing indices $i>n$. On finitely supported vectors its
coefficients tend to those of the bilateral operator
$2I+S+S^*$, where $S$ shifts the index by one. This implies strong
resolvent convergence: apply the resolvent identity first to the
dense set $(2I+S+S^*-z)u$, with $u$ finitely supported and
$\Im z\ne0$, then use the uniform resolvent bound.

All the Jacobi operators are nonnegative, so only $t\ge0$ enters
their functional calculus. In the scaled variable $t=y/n$, the multipliers
$g(\max\{D-nt/(4\pi a),d\})$ converge uniformly to
$g(\max\{D-Dt/(8\pi),d\})$. Subtracting $g(d)$ leaves continuous
functions with common compact support on this half-line; extend them
continuously with compact support to $\R$. Functional calculus, approximating
these functions by resolvent polynomials, therefore gives strong
convergence of their multiplication operators in the reversed
basis. Compression to indices $0,1,\ldots$ after this step gives
the Toeplitz operator $T(\mathfrak g)$ with symbol
\[
  \mathfrak g(\theta)
       =g\!\left(D-\frac{D}{8\pi}(2+2\cos\theta)\right).
\]
Pad each finite Gram matrix $A_n$ by a fixed positive scalar on
the remaining indices. They have a common positive lower bound,
so strong convergence also holds for the inverses, by
$A_n^{-1}-T^{-1}=A_n^{-1}(T-A_n)T^{-1}$.
Together with the strong convergence of the evaluation vectors,
this determines the multiplier correction to the kernel.

For clarity, its scalar value can be calculated directly. Define
\[
  \mathcal O(\zeta)=\exp\!\left[
     \frac1{4\pi}\int_{-\pi}^{\pi}
       \frac{e^{i\theta}+\zeta}{e^{i\theta}-\zeta}
                \log\mathfrak g(\theta)\dd\theta\right],
       \qquad |\zeta|<1.
\]
On the unit circle $|\mathcal O|^2=\mathfrak g$.
Multiplication by $\mathcal O$ identifies the analytic-function
norm weighted by $\mathfrak g$ with the ordinary Hardy norm.
Its evaluation kernel at $r$ consequently has squared norm
$[(1-r^2)|\mathcal O(r)|^2]^{-1}$.
The capacity is therefore multiplied by $|\mathcal O(r)|^2$.
The Poisson kernel identity
\[
  \frac{1-r^2}{1-2r\cos\theta+r^2}
      =\frac{\sqrt{z(z-4)}}{z-2-2\cos\theta}
\]
and $t=2+2\cos\theta$ give the constants
\begin{equation}
  \begin{aligned}
    g(x)&=\frac{(1-e^{-2\pi x})^2}{x},\qquad
       s_*=D\left(1-\frac1{2\pi}\right),\\
    \mathcal G&=\frac{\sqrt{z(z-4)}}\pi
       \int_0^4\frac{\log g(D-Dt/(8\pi))}
                         {(z-t)\sqrt{t(4-t)}}\dd t,\\
    C_{\mathrm{sc}}&=\log\!\left[\frac D4(1-r_{\mathrm{gap}}^2)^2\right]+\mathcal G,\\
    \chi&=\frac{2\pi\sqrt{(s_*-d)(D-d)}}d
             \left(\frac1{\sqrt{s_*D}}-\frac1D\right)>0.
  \end{aligned}
  \label{eq:entropy-kappa-gap-constants}
\end{equation}
Here $\mathcal G=\log|\mathcal O(r_{\mathrm{gap}})|^2$.
The evaluation is at positive $r_{\mathrm{gap}}$ because the
Laguerre polynomials were given positive leading coefficients.

The relative form error in
\eqref{eq:entropy-first-level-form-stability}, the $O(a^{-1})$
reference-density error, and the discarded tail in
\eqref{eq:entropy-laguerre-projection-tail} preserve this inverse
limit. The upper endpoint may be replaced by
$D_V=\sqrt{V_0/a}$ while holding the integer degree fixed:
$k-aD_V=-\{aD\}+O(V_0-U)$. Thus its exponentially small displacement
does not alter the floor term, even near a degree crossing.
The lower displacement from $h_0$ is also exponentially small.
The saddle and evaluation-vector estimates are uniform on these
shrinking neighborhoods. Finally
\eqref{eq:entropy-first-level-capacity} transfers the capacity to
the integer multiplicity. We have proved
$\log m_{a,0}=aH-\{aD\}\log r_{\mathrm{gap}}+C_{\mathrm{sc}}+o(1)$,
which is the scale $M_a$ in \eqref{eq:entropy-kappa-packet-scale}.
Strong inverse convergence supplies $o(1)$ here, without a
quantitative $O(a^{-1})$ rate.

\paragraph{Fixed-spin response.}
Let $j=|J|$ be fixed. In row $J$ use
$x=\sqrt{\Ehat/a-j/a}$ and translate to $v=x+j/(2aD)$.
The upper endpoint is $D+O(a^{-2})$ and the lower endpoint is
$d+(j/(2a))(D^{-1}-d^{-1})+O(a^{-2})$.
The reference exponent becomes
\[
  2\pi a\bigl(x+\sqrt{x^2+2j/a}\bigr)
      =4\pi av+2\pi j(v^{-1}-D^{-1})+O(a^{-1}),
\]
while the amplitude tends to $g(v)$ with no additional
order-one Jacobian. The preceding multiplier calculation applies
to $g(v)e^{2\pi j(v^{-1}-D^{-1})}$.
Writing $\ell_* =\sqrt{(s_*-d)(D-d)}$, its Poisson average in
this variable uses the probability measure
\[
  \dd\varpi_d(v)
    =\frac{\ell_*\dd v}
        {\pi(v-d)\sqrt{(v-s_*)(D-v)}},\qquad s_*<v<D.
\]
The elementary arcsine integral gives
\[
  \int v^{-1}\dd\varpi_d(v)
      =\frac1d\left(1-\frac{\ell_*}{\sqrt{s_*D}}\right),
  \qquad \partial_dH=\frac{8\pi r_{\mathrm{gap}}}{1+r_{\mathrm{gap}}}.
\]
Thus the coefficient of $j$ in the capacity ratio is
\[
  \frac{\partial_dH}{2}(D^{-1}-d^{-1})
      +2\pi\left(\int v^{-1}\dd\varpi_d-D^{-1}\right)
       =-\chi.
\]
The simplification uses
$\ell_*=(D-d)(1-r_{\mathrm{gap}})/(1+r_{\mathrm{gap}})$.
This proves, with the integer transfer and endpoint displacement,
\begin{equation}
  \frac{m_{a,J}}{M_a}\longrightarrow e^{-\chi|J|},
  \qquad \xi_{a,J}\longrightarrow0^+.
  \label{eq:entropy-kappa-endpoint-atoms}
\end{equation}
The degree is common to all rows, so its floor phase cancels in
this fixed-spin ratio. The initial data need not be invariant under
$J\mapsto-J$: their effect on the forms is exponentially small,
while the leading vacuum reference is even in $J$.

\paragraph{Other atoms in a fixed gap neighborhood.}
For a fixed row, let $P_J$ be the nonnegative extremal polynomial
for the endpoint capacity, normalized by $P_J(h_J)=1$.
Exact moments give the identity
\begin{equation}
  \sum_{s_i\ne h_J}P_J(s_i)=\lambda_{a,J}-m_{a,J}.
  \label{eq:entropy-nonendpoint-moment-identity}
\end{equation}
A positive lower bound on $P_J$ is needed to turn this weighted
sum into an atom count. It follows from the off-diagonal kernel,
not merely from equivalence of quadratic-form norms.
For a fixed physical offset $x=\Ehat-b\in[0,W]$, the Laguerre
coordinate has displacement $y(x)-Y=-2\pi x/d+o(1)$.
The normalized reversed evaluation vectors $w_n(y(x))$ converge
uniformly to the same $(1,r_{\mathrm{gap}},r_{\mathrm{gap}}^2,\ldots)$
as at $Y$, by the ratio bound used above. If $A_n$ is the multiplier
Gram matrix, the endpoint-normalized kernel polynomial factors as
\[
  \frac{\mathcal K_{n,a}(y(x),Y)}{\mathcal K_{n,a}(Y,Y)}
    =\frac{p_n(y(x))}{p_n(Y)}
       \frac{w_n(y(x))^TA_n^{-1}w_n(Y)}
            {w_n(Y)^TA_n^{-1}w_n(Y)}.
\]
Uniform upper and lower bounds on $A_n$ make the second factor
tend uniformly to one. The saddle gives the first factor's limit
$\exp[-2\pi x r_{\mathrm{gap}}/(d(1+r_{\mathrm{gap}}))]$.
In the odd case the additional factor $y(x)/Y$ tends to one.
Relative form perturbations change these kernel ratios by $o(1)$,
using the same bounded evaluation norms. Consequently
\begin{equation}
  P_J(b+x)\longrightarrow
     \exp\!\left[-\frac{4\pi r_{\mathrm{gap}}}{d(1+r_{\mathrm{gap}})}x\right]
       \quad\text{uniformly on }[0,W],
  \label{eq:entropy-gap-extremal-polynomial}
\end{equation}
where $P_J$ is now evaluated through the physical row coordinate.
Equations~\eqref{eq:entropy-first-level-capacity} and
\eqref{eq:entropy-nonendpoint-moment-identity} show that all other
atoms in this row and band have total weight $o(m_{a,J})$.
Protection ensures that later repairs do not change them or add
atoms to this initial band.

\paragraph{Uniform control of the spin sum.}
Fixed-spin limits alone do not permit summing all rows. For
$\gamma=|J|/a\le\kappa$, the dimensionless energy at normalized
coordinate $s$ is
\[
  X_\gamma(s)=(1-s)\kappa+s\vartheta
       -s(1-s)\bigl(\sqrt{\vartheta-\gamma}
                          -\sqrt{\kappa-\gamma}\bigr)^2.
\]
It decreases with $\gamma$. Its action
$\Phi_\gamma(s)=2\pi(\sqrt{X_\gamma(s)+\gamma}
                         +\sqrt{X_\gamma(s)-\gamma})$
is strictly smaller than $\Phi_0(s)$ for every $\gamma>0$.
At zero spin,
\[
  \left.\partial_\gamma\Phi_\gamma(s)\right|_{\gamma=0}
    =-\frac{2\pi(D-d)^2}{dD}
             \frac{s(1-s)}{d+(D-d)s}.
\]
Choose $0<s_0<(s_*-d)/(D-d)$.
For some $q_0>0$ and small $0<\gamma_0<\kappa/2$, this gives
\[
  \Phi_\gamma(s)\le4\pi[d+(D-d)s]-q_0\gamma(1-s),
       \quad s\ge s_0,\quad 0\le\gamma\le\gamma_0.
\]
Use the nonnegative endpoint trial polynomial for the comparison
exponential with slope $B_\gamma=4\pi(D-d)+q_0\gamma$.
Its capacity has action
$4\pi D-D\Psi(2B_\gamma/D)$, whose derivative in $\gamma$ is
$-2q_0\Psi'(2B_\gamma/D)<0$. Its prefactors are bounded uniformly
on this compact slope interval.

The part $s<s_0$ also needs control. In the comparison variable
$y=aB_\gamma(1-s)$, its lower cutoff obeys $y/n>4+\eta_0$.
Cauchy--Schwarz for the projection kernel and
\eqref{eq:entropy-laguerre-projection-tail} bound the fraction of
the trial norm in this part by $a^C e^{-\eta_1a}$.
The actual exponential exceeds the comparison there by at most
$e^{q_0\gamma a}$. Shrink $\gamma_0$ until
$q_0\gamma_0<\eta_1/2$; this omitted norm is still exponentially
small. The bounded reference prefactors and
\eqref{eq:entropy-first-level-form-stability} now give trial
polynomials $P_{a,J}$, normalized to one at the left endpoint, with
\begin{equation}
  \Lambda_{a,J}(P_{a,J})
      \le C M_a e^{-\chi_1|J|},
      \qquad |J|\le\gamma_0a,\qquad \chi_1>0.
  \label{eq:entropy-small-spin-trial-bound}
\end{equation}
Their values on a fixed physical band $[b,b+W]$ have a uniform
positive lower bound by the same exterior kernel ratios.
Thus \eqref{eq:entropy-small-spin-trial-bound} bounds all atoms
there, not only the endpoint multiplicity, by
$C_WM_a e^{-\chi_1|J|}$.

For $\gamma_0\le\gamma\le\kappa$, compactness and the strict
action inequality give $\Phi_\gamma\le\Phi_0-\eta_2$ on the
whole interval. The scalar trial polynomial therefore has row
norm at most $M_a a^C e^{-\eta_2a}$.
This remains true at a spin opening: in the square-root coordinate
the integrable threshold measure has at most a polynomial cost,
and the signed correction is bounded as a form.
For $\kappa<\gamma\le\kappa+W/a$ the row starts at zero
square-root coordinate and
$X_\gamma(s)=\gamma+(\vartheta-\gamma)s^2$.
This map joins continuously at $\gamma=\kappa$, preserving a
fixed action gap. More distant rows cannot contribute to the band.

At such an opening, $[b,b+W]$ occupies $s=O_W(a^{-1/2})$.
The scalar trial's zeros stay a fixed distance from $s=0$:
the Christoffel--Darboux formula writes its polynomial factor as
a Laguerre combination with one exterior root $Y$; the remaining
roots interlace the Laguerre roots
\cite[Secs.~18.2(v)--(vi)]{NISTDLMF} and are at most $4n+O(1)$.
The latter bound follows also from row sums of the Laguerre Jacobi
matrix. Since $Y/n>4$ by a fixed margin, the zeros in $s$ are
bounded below by some $s_{\min}>0$ independent of $a$.
The product over $O(a)$ zeros bounds the trial from below by
$e^{-C_W\sqrt a}$ on the physical band. Its odd-degree factor
$1-s$ obeys the same bound. Summing the $O(a)$ macroscopic rows
therefore gives at most $M_a a^C e^{-\eta_2a/2+C_W\sqrt a}=o(M_a)$.

Now truncate to finitely many spins, use
\eqref{eq:entropy-kappa-endpoint-atoms} and
\eqref{eq:entropy-gap-extremal-polynomial}, and remove the spin
truncation using the geometric majorant and macroscopic bound.
The primary measure at offsets $x=\Delta-\Delta_g$ divided by
$M_a$ converges against every compact continuous test function to
$\coth(\chi/2)\delta_0$. Only finitely many descendant levels
contribute on a fixed compact offset interval. Their convolution
therefore gives \eqref{eq:entropy-kappa-gap-count}.
The prescribed first level contributes exactly
$d_{\mathrm{first}}\mathcal B_\phi(u)$, independently of its spins
because all total spins are counted. It and the vacuum descendants
are $o(M_a)$. For the bump with $w>1/2$ and $u\ge0$, a nonnegative
integer lies in its positive support, so $\mathcal B_\phi(u)>0$;
taking logarithms proves \eqref{eq:entropy-kappa-gap-entropy}.
The available remainder is $o(M_a)$ and does not identify the
bounded first-level contribution as a controlled next term.

\paragraph{A cumulative bound and the Gaussian near-gap limit.}
The preceding trial polynomials give more than a fixed-band bound.
Their zeros in the normalized coordinate satisfy $s_i\ge s_{\min}>0$
uniformly over the small-spin comparison slopes and the scalar
trial. The product over at most $O(a)$ zeros gives
$P(s)\ge e^{-Cas}$ for $0\le s\le s_{\min}/2$.
For $|J|\le\gamma_0a$ and a physical band $[b,b+X]$,
$s\le CX/a$. Combining this with
\eqref{eq:entropy-small-spin-trial-bound} and exact positive node
moments gives
\[
  Q_{a,J}([b,b+X])\le C M_a e^{-\chi_1|J|+CX},
       \qquad 0\le X\le\epsilon_0a.
\]
Here $\epsilon_0>0$ is small and fixed. For the remaining rows,
the weaker coordinate bound $s\le C\sqrt{X/a}$ gives a loss
$e^{C\sqrt{aX}}$, while their norms have the strict loss
$e^{-\eta_2a}$. Choose
$\epsilon_0<(R_0-1)\kappa$ and $\epsilon_0<\vartheta-\kappa$
so that all relevant rows are initial rows, and decrease it until
$C\sqrt{\epsilon_0}<\eta_2/2$.
The action gap extends continuously through openings up to
$\gamma=\kappa+\epsilon_0$. The sum of these rows is bounded by
$M_a a^C e^{-\eta_2a/2}$. Summing the small-spin geometric series
and adding the at most $M_*$ prescribed primaries proves the
primary cumulative bound $C M_a e^{CX}$ in this band.

Let $\mathsf M_a(X)$ count all nonvacuum states whose dimensions
lie in $[\Delta_g,\Delta_g+X]$. Convolution with descendants costs
only $\sum_Np_2(N)e^{-LN}=G(L)$ after increasing $L>0$.
For $X\ge\epsilon_0a$, the global bound
\eqref{eq:entropy-kappa-global-envelope} is at most $Ce^{LX}$:
use $a\le X/\epsilon_0$ to absorb its square roots and polynomial
factors. Since $M_a\ge1$ eventually, both ranges give
\begin{equation}
  \mathsf M_a(X)\le C M_a e^{LX},\qquad X\ge0,
  \label{eq:entropy-kappa-gap-tail}
\end{equation}
with $C,L$ independent of large $a$, $X$ and the bounded first-level
data. This controls every actual nonvacuum atom, including ones
away from the endpoint levels.

For a fixed Gaussian the normalized contribution above offset $R$
is bounded by
\[
  C\sum_{m\ge\lfloor R\rfloor}e^{L(m+1)}
       \sup_{x\in[m,m+1]}\phi(x-u),
\]
which tends to zero uniformly in large $a$ as $R\to\infty$.
Thus the compact-offset limit can be integrated against the full
Gaussian. Vacuum descendants contribute $e^{O(\sqrt a)}$ after
Gaussian weighting: use
$\sqrt N\le\sqrt{\Delta_g}+\sqrt{|N-\Delta_g|}$ and the
coefficient bound above. They remain negligible compared with
$M_a$. This proves the fixed-Gaussian version of
\eqref{eq:entropy-kappa-gap-count} and its entropy, with the same
$o(1)$ precision.

\paragraph{Descendants in the intermediate region.}
Set $X_a=a(\epsilon-\kappa)$ for fixed $\epsilon>\kappa$.
Count only descendants of the endpoint levels in rows $|J|\le b$.
Their inward energy shifts tend uniformly to zero: the initial
cell weights have a uniform exponential lower bound, so the
formula for $h_J$ gives $\sup_{|J|\le b}\xi_{a,J}\le a^Ce^{-\eta a}$.
The spin bounds just proved imply
$\sum_{|J|\le b}m_{a,J}=M_a[\coth(\chi/2)+o(1)]$.
For a compact kernel, apply
\eqref{eq:entropy-kappa-partition-coefficients} to
$N=X_a+O_\phi(1)$ and use the uniform shift bound. This gives
\begin{equation}
  \begin{aligned}
    N_\phi^{\mathrm{gap}}(a\epsilon-1/12)
      ={}&\frac{M_a\coth(\chi/2)}{4\,3^{3/4}}
          X_a^{-5/4}e^{2\pi\sqrt{X_a/3}}\\
       &\times\bigl[P_\phi(X_a)+o(1)\bigr].
  \end{aligned}
  \label{eq:entropy-kappa-descendant-lower-bound}
\end{equation}
For a fixed Gaussian, the coefficient bound and quadratic decay
justify the same expansion, as in the vacuum calculation above.
This is a contribution to the positive full spectrum, so
$N_\phi\ge N_\phi^{\mathrm{gap}}$.
When $P_\phi(X_a)$ has a uniform positive lower bound, comparison
with the continuous BTZ reference gives
\begin{equation}
  \begin{aligned}
    \log\frac{N_\phi(a\epsilon-1/12)}
                  {N_\phi^{\mathrm{BTZ}}(a\epsilon-1/12)}
      \ge{}&a(H-4\pi\sqrt\epsilon)
          +\frac{2\pi}{\sqrt3}\sqrt{a(\epsilon-\kappa)}\\
       &-\frac34\log a-C.
  \end{aligned}
  \label{eq:entropy-kappa-nonbtz-bound}
\end{equation}
Indeed, Appendix~\ref{app:entropy-reference-expansion}, uniformly
at $\lambda_a=(a\epsilon-1/12)/(12a+1)$, gives reference logarithm
$4\pi a\sqrt\epsilon-\tfrac12\log a+O(1)$.
The descendant prefactor supplies $-\tfrac54\log a$;
subtracting produces $-\tfrac34\log a$.
The floor phase and the exact cylinder-energy shift affect only
bounded terms. The chosen bump with $w>1/2$, and every fixed
Gaussian, satisfy the sampling condition. Dividing by $a$ proves
Proposition~\ref{prop:entropy-kappa-nonbtz}. At
$\epsilon=\epsilon_H$ the positive $\sqrt a$ term still forces
the ratio to diverge. No matching upper bound for the full
intermediate profile follows from this subfamily of states.

\paragraph{Transfer above the initial scale.}
Fix a compact range with $12\lambda\ge\vartheta+\eta$, where
$\eta>0$, and set $Y_0=\lambda c+1/12$ exactly.
The front and short-cell rules in \eqref{eq:entropy-tail-contract}
still hold, but now $T$ and $U$ grow linearly with $a$.
The initial weight must therefore be estimated afresh.
At $W=U+1$, \eqref{eq:global-initial-variation} bounds it by a
polynomial times the sum of
$e^{4\pi\sqrt{a(U+1)}}$ and
$e^{a(4\pi\sqrt\kappa+C_{\mathrm{ref}}\kappa)}$.
The admissibility inequalities make the second exponent smaller
than $4\pi aD$. Summing the $O(a)$ initial rows and adding the
bounded prescribed weight yields
$M_{\mathrm{init}}\le C_{\kappa,M_*,J_*}a^Ce^{4\pi aD}$.
For a compact kernel at $E=\lambda c$, all their descendants thus
contribute at most
\begin{equation}
  C_{\kappa,M_*,J_*}\operatorname{poly}(a)
     \exp\!\left(4\pi a\sqrt\vartheta+O(\sqrt a)\right).
  \label{eq:entropy-kappa-initial-descendants}
\end{equation}
Here \eqref{eq:entropy-descendant-cumulative-bound} bounds all
allowed levels by $e^{O(\sqrt a)}$. The target action
$S_a(Y_0)$ exceeds $4\pi aD$ by a fixed positive multiple of $a$,
so these descendants are negligible, including those which reach
the observation window. The vacuum module is smaller still.

Choose fixed $0<\eta'<\eta/4$ small enough that for
$N\le\eta'a$, the parent centers $Y_N=Y_0-N$ remain above
$U$ by an extensive margin. On every central row
$|J|\le Y_N/2$, the reference dominates $H_a$ exponentially:
the smallest reference exponent in that cone has coefficient
$2\pi(\sqrt{3/2}+\sqrt{1/2})>7$.
Thus the current cell measure is positive there. Equations
\eqref{eq:entropy-cell-derivatives} and
\eqref{eq:entropy-uniform-polynomial-error} apply unchanged to
the smooth test function, with constants independent of these
cells. Exact moments bound the error by $C_PM_na^{-P}$.
Summing cell weights with the reference Gaussian spin bound gives
$C_Pa^{-P-1/2}e^{S_a(Y_N)}$.
The excluded spins have an extensive action loss, with their
threshold measures integrated using
\eqref{eq:entropy-integrated-spin-bound}.

Concavity then bounds the sum of all retained errors by
\[
  C_Pa^{-P-1/2}e^{S_a(Y_0)}
                  \sum_{N\ge0}p_2(N)e^{-\beta_0N}
     =C_Pa^{-P-1/2}e^{S_a(Y_0)}G(\beta_0).
\]
For $N>\eta'a$, tail parents lie at least $\eta'a-O_\phi(1)$
below the target. The complete-cell envelope bounds their primary
action by $S_a(Y_0-\eta'a+O_\phi(1))$, with polynomial losses.
The full descendant factor is $e^{O(\sqrt a)}$, leaving an
extensive deficit. Initial parents have already been bounded by
\eqref{eq:entropy-kappa-initial-descendants}. The same split
controls the reference tail. A positive subinterval of the kernel
and $|J|\le\sqrt a$ give reference lower bound
$C^{-1}a^{-1/2}e^{S_a(Y_0)}$.
This proves every fixed relative inverse-power accuracy for
transfer to the integer-spin reference.

For its spin integral, the explicit derivative bound
\eqref{eq:entropy-poisson-derivative-bound} is uniform on this
parent-energy compact. Choose more than $2P$ derivatives in
\eqref{eq:entropy-poisson-error} and sum descendants with the
same majorant. The exact continuous-spin transform is still
\eqref{eq:entropy-reference-btz}, with physical $c=12a+1$.
Appendix~\ref{app:entropy-reference-expansion} therefore proves
Proposition~\ref{prop:entropy-kappa-btz}, including its identical
entropy coefficients at every fixed finite order.

For a fixed Gaussian, fix the requested order $P$ and cut off
at offsets $R_a=A_P\sqrt{\log a}$, with a smooth transition up
to $2R_a$. Its derivative norms of any fixed order remain bounded.
The number of relevant cells is $O(R_a)$ per row and their total
unweighted mass grows by at most $e^{C R_a}$ relative to the
central scale. Increasing the fixed polynomial approximation
order absorbs this factor. In the remaining band
$R_a<|s|<\epsilon a$, choose $\epsilon$ small enough to retain
the extensive margin above $U$. Unit-bin complete-cell and
descendant estimates cost at most $a^Ce^{C|s|}$ relative to the
central scale. Taking $A_P$ large makes the Gaussian-weighted
sum $O(a^{-P})$. For $|s|\ge\epsilon a$,
\eqref{eq:entropy-kappa-global-envelope}, with the analogous
vacuum and reference bounds, has logarithm at most
$C(a+\sqrt{a s_+}+\sqrt{s_+})+C\log(1+a+s_+)$.
Quadratic Gaussian decay dominates this uniformly, including the
initial weight $e^{O(a)}$ and arbitrarily large positive offsets.
The far contribution is smaller than every fixed inverse power.
This completes the fixed-Gaussian high-energy proof.

All these limits keep the kernel, $\kappa$ and the construction
parameters fixed. The proofs give no quantitative inverse-$a$
near-gap expansion, full intermediate profile, or uniform result
as the high-energy margin closes or $\kappa\to0$.

\bibliographystyle{modular-spectrum}
\bibliography{references}
\end{document}